\documentclass{LMCS}

\def\doi{8 (3:10) 2012}
\lmcsheading%
{\doi}
{1--56}
{}
{}
{May~\phantom.25, 2011}
{Aug.~13, 2012}
{}

\newif\ifmylmcs
\newif\ifmylmcsfinalversion
\newif\ifmyacm
\newif\ifmyieee
\newif\ifhyperdoc
\mylmcstrue
\mylmcsfinalversiontrue
\myacmfalse
\myieeefalse
\ifmyieee
  \documentclass[times, 10pt,twocolumn]{article} 
  \usepackage{latex8}
  \usepackage{times}
\fi
\hyperdocfalse
\ifmylmcs\hyperdoctrue\fi
\def\LMCSPaperBib{\ifmylmcs -lmcs-paper\fi}

\newif\ifdraft
\draftfalse
\ifdraft
   \usepackage{ps-timestamp}
   \pstimestampextra{OMEGA-REGULAR}
\fi

\ifhyperdoc  %% \ifhyperdoc
   \usepackage{hyperref}
   \let\myhypertarget=\hypertarget
   \let\myhyperlink=\hyperlink
\else
   \def\myhypertarget#1#2{#2}
   \def\myhyperlink#1#2{#2}
\fi

\ifmylmcs
  \usepackage[square,numbers]{natbib}
  \let\cite=\citep  % Force citations to be in form (author year).
  \let\citeyear=\citeyearpar % Force year to be of form (year).
\else
  \usepackage{cite}
\fi

\usepackage{itlsym}  % Custom virtual fonts needed in document

\usepackage[reqno,centertags]{amsmath}
\ifmyacm
  
\else
  \ifmylmcs
  \else
    \usepackage{amsthm}
  \fi
\fi

\theoremstyle{plain}

\ifmylmcs

  \renewcommand{\proofname}{Proof.}
\else
  \newtheorem{mytheorem}{Theorem}[section]
  
  \newtheorem{mycorollary}[mytheorem]{Corollary}
  \newtheorem{mylemma}[mytheorem]{Lemma}
  \newtheorem{mydefin}[mytheorem]{Definition}
  \newtheorem{myremark}[mytheorem]{Remark}
\fi

\usepackage{graphicx,color}

\usepackage{url}
\usepackage{keyval}

\usepackage{exscale}

\usepackage{array}

\usepackage{enumerate}
\usepackage{longtable}
\usepackage{version}
  \excludeversion{mycomment}
  \excludeversion{hidetheorems}
  \excludeversion{jlc-omit}
  \includeversion{jlc-include}

\newbox\tempbox
\newbox\tempboxa
\newdimen\tempdimen
\newcount\tempcount
\newtoks\temptoks
\newtoks\temptoksa
\newtoks\temptoksb
\newtoks\temptoksc

\let\theoremStandardDef=\theorem
\input{bcmmacro2e.sty}
\makeltother
\let\theoremBEN=\theorem
\let\theorem=\theoremStandardDef

\let\D=\displaystyle
\newenvironment{myarray}[1][t]{\begin{array}[#1]{@{}>{\D}l@{}}}
   {\end{array}\myignoretrue}

\def\csnamelet#1#2{\expandafter\csnameletAux\csname#2\endcsname#1\relax}
\def\csnameletAux#1#2{\let#2=#1}

\newdimen\dispskipdimen
\setbox\tempbox=\hbox{1} %
\dispskipdimen=\wd\tempbox
\newif\ifWithinProveIt  % Keep track of when within scope of proveit

\newenvironment{proveit}[1][]{\begingroup
\newbox\tempProveItBox
\setbox\tempProveItBox=\hbox{#1} %
\global\WithinProveIttrue
\useslashslash
\ifdim\wd\tempProveItBox>0pt %
  \unskip\vspace{-5pt} %
  \noindent\textsc{#1.} %
\fi
\halign\bgroup \hbox to \dispskipdimen{$##$\hfil}\null
   \tabskip=1em plus 2em minus 0.5em %
   &\hbox to \wd\proveitformulabox{$\Theorem ##$\hfil}
   &\enskip ##\hfil\cr
}{\cr\egroup\endgroup
\global\WithinProveItfalse
\myignoretrue}

\newbox\proveitformulabox
\setbox\proveitformulabox=
    \hbox to 4in{}
{\makeatletter
\global\let\myignoretrue=\@ignoretrue
}

\newcolumntype{C}{>{$}c<{$}}
\newcolumntype{L}{>{$}l<{$}}
\newcolumntype{R}{>{$}r<{$}}

\def\Exists#1{\exists #1\mathpunct{.}}

\def\acc{\mathit{acc}}
\def\Andd{\;\And\;}

\def\Atoms{\ifmmode \textrm{\emph{Atoms}}\else $\textrm{\emph{Atoms}}$\fi}

\def\chop{\mathord{^\frown}}
\def\ChopstarSymbol{{\text{\large $\boldsymbol\scriptstyle*$}}}

\def\DC{\ifmmode \textrm{DC}\else $\textrm{DC}$\fi}
\def\DD#1{\langle #1\rangle}
\def\eh{\ifmmode \mathord{\textit{eh}}\else $\mathord{\textit{eh}}$\fi}
\let\emphbf=\textbf

\def\ETL{\ifmmode \textrm{ETL}\else $\textrm{ETL}$\fi}
\def\FE{\ifmmode \textrm{FE} \else $\textrm{FE}$\fi}
\def\FEV{\ifmmode \textrm{FE}_V \else $\textrm{FE}_V$\fi}
\def\FL{\ifmmode \textrm{FL} \else $\textrm{FL}$\fi}
\def\FLV{\ifmmode \textrm{FL}_V \else $\textrm{FL}_V$\fi}
\def\FLVV{\ifmmode \textrm{FL}_{V,V} \else $\textrm{FL}_{V,V}$\fi}
\def\Fusion{\ifmmode \textrm{Fusion}\else $\textrm{Fusion}$\fi}
\def\init{\mathit{init}}
\def\INT{\ifmmode \textrm{INT} \else $\textrm{INT}$\fi}
\def\ITL{\ifmmode \textrm{ITL}\else $\textrm{ITL}$\fi}
\def\Inv{\ifmmode \textit{Inv} \else $\textit{Inv}$\fi}

\def\ldepth{\ifmmode \textrm{ldepth}\else $\textrm{ldepth}$\fi}
\def\lh{\ifmmode \mathord{\textit{lh}}\else $\mathord{\textit{lh}}$\fi}
\def\LRL{\ifmmode \textrm{LRL}\else $\textrm{LRL}$\fi}
\def\LTL{\ifmmode \textrm{LTL}\else $\textrm{LTL}$\fi}
\def\ModalSys#1{\textit{#1}}
\def\NL{\ifmmode \textrm{NL}\else $\textrm{NL}$\fi}
\def\NLone{\ifmmode \NL^{\!1}\else $\NL^{\!1}$\fi}
\def\nuTL{\ifmmode \nu\textrm{TL}\else $\nu\textrm{TL}$\fi}
\def\PDL{\ifmmode \textrm{PDL}\else $\textrm{PDL}$\fi}
\def\PITL{\ifmmode \textrm{PITL}\else $\textrm{PITL}$\fi}
\def\PITLK
  {\ifmmode \textrm{PITL}^{\!\textsc{k}} \else $\textrm{PITL}^{\!\textsc{k}}$\fi}
\def\PITLV{\ifmmode \textrm{PITL}_V \else $\textrm{PITL}_V$\fi}

\def\PTL{\ifmmode \textrm{PTL}\else $\textrm{PTL}$\fi}
\def\PTLF{\ifmmode \textrm{PTL}_{\lt \omega}
   \else $\textrm{PTL}_{\lt \omega}$\fi}
\def\PTLI{\ifmmode \textrm{PTL}_\omega\else $\textrm{PTL}_\omega$\fi}
\def\PTLU{\ifmmode \textrm{PTL}\!^{\textsc{u}} \else
   $\textrm{PTL}\!^{\textsc{u}}$\fi}
\def\QPITL{\ifmmode \textrm{QPITL}\else $\textrm{QPITL}$\fi}
\def\QPTL{\ifmmode \textrm{QPTL}\else $\textrm{QPTL}$\fi}
\def\Reg{\ifmmode \textrm{Reg} \else $\textrm{Reg}$\fi}
\let\state=s
\def\SAT{\ifmmode \textrm{SAT} \else $\textrm{SAT}$\fi}
\def\SChopstarSymbol{{\text{\Large $\scriptstyle\boldsymbol\star$}}}
\def\SChopstar{^\SChopstarSymbol}
\def\SChopstarSChopstar{^{\SChopstarSymbol\SChopstarSymbol}}
\let\ConventionalChopstar=\Chopstar

\def\SOneS{\ifmmode \textrm{S1S}\else $\textrm{S1S}$\fi}
\def\sh{\ifmmode \mathord{\textit{sh}} \else $\mathord{\textit{sh}}$\fi}
\def\Test{?}
\def\U{\textbin{\mathcal{U}}}
\def\V{\ifmmode {\mathcal{V}} \else $\mathcal{V}$\fi}
\def\Var{\ifmmode \textit{Var} \else $\textit{Var}$\fi}
\def\VV{\ifmmode V\! \else $V\!$\fi}
\def\UntilOp{\mathit{until}}
\def\Var{\ifmmode \mathit{Var}\else $\mathit{Var}$\fi}

\def\VV{\ifmmode V\! \else $V\!$\fi}

\def\vld{\mathrel{\raisebox{\depth}{$\scriptstyle \models$}}}

\def\thmR{\vdash_{\mathrm{rt}}}
\def\theoremR{\thmR\;}
\def\TheoremR{\thmR\quad}

\def\x#1{{\rm (#1)}}

\let\dotspace=\null
\def\badVerbChars{\catcode`\&=\other
   \catcode`\^=\other
   \catcode`\_=\other
}
\def\badref{\bgroup\badVerbChars\badrefAux}
\def\badrefAux#1{\ref{#1} %
   (??? {\hyphenpenalty=0 %Don't warn about hyphenation problems
   \hbadness=10000   % Don't warn about underflow hboxes
   \textbf{\sffamily #1}} ???)\egroup}
\def\badeqref{\bgroup\badVerbChars\badeqrefAux}
\def\badeqrefAux#1{\eqref{#1} %
   (??? {\hyphenpenalty=0 %Don't warn about hyphenation problems
   \hbadness=10000   % Don't warn about underflow hboxes
   \textbf{\sffamily #1}} ???)\egroup}

\newif\ifmynote
\newif\ifmywarnings % See if any of my own warnings issued
\newif\ifmyidwarnings % See if any changes in my ids
\newif\ifmyidundefined % See if any undefined ids
\newif\ifNotEOF   % Used to indicate when reach end of a file
\newif\ifprintmyids % Include printed versions of my ids
\newif\ifPrintRefIdXrefs % Print cross reference, etc of refids
\newif\ifUseHardcopyIds % If true, keep generated names for ids from *.myids_p
\mynotefalse
\mywarningsfalse    % Initially none of my own warnings issued.  Don't change.
\myidwarningsfalse % Initially no changes in my references.  Don't change.
\myidundefinedfalse % Initially no undefined ids.  Don't change.
\printmyidstrue
\PrintRefIdXrefsfalse
\UseHardcopyIdsfalse

\def\defof#1{def.~of~$#1$}

\def\mywrite#1#2{{\edef\tempmac{{#2}}\write#1\tempmac}}

\makeatletter
\newif\ifrefidPageNo
\newcount\myidundefcount  \myidundefcount=0 %
\def\refid{\@ifnextchar[\refidAux{\refidAux[]}}
\def\refidAux[#1]#2{\relax
   \ifInTextBody  % needed for MyThms support
     \immediate\write\MyThmsOut
       {\noexpand\MyThmsUsedInTextBody
	 #2\noexpand\EndMyThmsUsedInTextBody}\relax
   \fi
   \refidPageNofalse
   \setkeys{refid}{#1}\relax
   \JustDidXreffalse
   \ifrefidPageNo
      \forceidpagenum{#2}\relax
   \else
      \xrefpagenum{#2}\relax
   \fi
   \ifrefidPageNo
      \let\refidHypercommand=\myhypertarget
   \else
      \expandafter\ifx\csname #2-dummy-id\endcsname\relax
         \let\refidHypercommand=\myhyperlink
      \else
         \let\refidHypercommand=\MyHyperDummy
      \fi
   \fi
   \expandafter\ifx\csname defidDef-#2\endcsname\relax
      \expandafter\ifx\csname #2-ref\endcsname\relax
         \mywarning{Undefined ITL id #2}\relax
         \WriteIdUndefined{#2}\relax
         \myhyperlink{undefid.#2}{#2???}\relax
         \ifinner
         \else
            \mynote{Undefined ITL id: #2}\relax
         \fi
      \else
         \refidHypercommand{refid.#2.0}{{\normalfont
                \csname #2-ref\endcsname}}\relax
            \MyShowRefsFirstPageNum{#2}\relax
      \fi
   \else
      \refidHypercommand{refid.#2.0}{{\normalfont
            \csname defidDef-#2\endcsname}}\relax
         \MyShowRefsFirstPageNum{#2}\relax
   \fi
}
\define@key{refid}{pagenum}[true]{\csname refidPageNo#1\endcsname}
\makeatother

\newif\ifJustDidXref  % Used to tell \WriteIdUndefined to produce xref

\def\WriteIdUndefined#1{\relax
   \expandafter\ifx\csname #1-undefined\endcsname\relax
      \global\myidwarningstrue
      \global\myidundefinedtrue
      \global\advance\myidundefcount by 1 %
      \mywrite\MyIdsOut{\noexpand\string\noexpand\MyIdUndefined
         ,#1,\noexpand\the\noexpand\pagenum
         ,\the\myidundefcount,}%
      \ifJustDidXref
      \else
         \xrefpagenum{#1}%
      \fi
      \expandafter\gdef\csname #1-undefined\endcsname{1}%
   \fi
}

\def\forceidpagenum#1{{\relax
   \expandafter\ifx\csname #1-number-of-defs\endcsname\relax
      \expandafter\gdef\csname #1-number-of-defs\endcsname{0}%
   \fi
   \tempcount=\csname #1-number-of-defs\endcsname
   \advance \tempcount by 1 %
   \expandafter\xdef\csname #1-number-of-defs\endcsname{\the\tempcount}%
   \myhypertarget{iddef.#1.\the\tempcount}{}%
   \write\MyIdsOut{\string\MyIdPageNum,#1,\the\pagenum,}%
}}

\newcount\mygencount
\global\mygencount=0

\def\xrefpagenum#1{{\expandafter\ifx\csname #1-disable-xrefs\endcsname\relax
   \expandafter\ifx\csname #1-xref-max\endcsname\relax
      \tempcount=1 %
   \else
      \tempcount=\csname #1-xref-max\endcsname
      \advance \tempcount by 1 %
   \fi
   \expandafter\xdef\csname #1-xref-max\endcsname{\the\tempcount}%
   \mywrite\MyIdsOut{\noexpand\string\noexpand\MyIdXrefPageNum
      ,#1,\the\tempcount,\noexpand\the\noexpand\pagenum,}%
   \ifWithinProveIt
      \ifx\MyCurrentStateItId\relax
      \else
         \immediate\write\MyIdsOut{\string\MyIdXrefInProof
           ,#1,\the\tempcount,\MyCurrentStateItId,}%
      \fi
   \fi
   \myhypertarget{myxref.#1.\csname #1-xref-max\endcsname}{}%
   \global\JustDidXreftrue
\fi
}}

\def\MyShowRefsFirstPageNum#1{\relax
   \ifWithinProveIt
      \expandafter\ifx\csname #1-disable-xrefs\endcsname\relax
         \expandafter\ifx\csname #1-first-pagenum\endcsname\relax
         \else
            {\small\ (p.\ \csname #1-first-pagenum\endcsname)}\relax
         \fi
      \fi
   \fi
}

\let\MyCurrentStateItId=\relax

\def\ifundefined#1{\relax\expandafter\ifx\csname #1\endcsname\relax}
\def\newidcount#1{\expandafter\newcount\csname idcount#1\endcsname
   \csname idcount#1\endcsname=\ifUseHardcopyIds 300\else 1\fi}
\newidcount{T}
\newif\ifnewidPageNo
\newif\ifnewidPrevDefOk
\newif\ifnewidRef
\makeatletter
\def\newid{\@ifnextchar[\newidAux{\newidAux[]}}
\def\newidAux[#1]#2{\let\IdPrefix=\DefaultIdPrefix
   \let\IdBase=\relax
   \newidPrevDefOkfalse
   \newidPageNofalse
   \newidReffalse
   \setkeys{newid}{#1}\relax
   \expandafter\ifx\csname #2-hardcopy.latex-code\endcsname\relax
      \newidAuxB{\IdPrefix}{#2}\relax
   \else
      \newidUseHardcopyDef{#2}\relax
   \fi
   \ifnewidRef
      \ifnewidPageNo
         \refid[pagenum]{#2}\unskip
      \else
         \refid{#2}\unskip
      \fi
   \else
      \ifnewidPageNo \forceidpagenum{#2}\fi
   \fi
}
\define@key{newid}{base}{\def\IdBase{#1}}
\define@key{newid}{pagenum}[true]{\csname newidPageNo#1\endcsname}
\define@key{newid}{prefix}{\def\IdPrefix{#1}}
\define@key{newid}{prevdef-ok}[true]{\csname newidPrevDefOk#1\endcsname}
\define@key{newid}{ref}[true]{\csname newidRef#1\endcsname}
\makeatother

\def\newidAuxB#1#2{%
   \ifx\IdBase\relax
      \expandafter\ifx\csname #2\endcsname\relax
         \expandafter\xdef\csname #2\endcsname
            {\noexpand\textbf{#1\the\csname idcount#1\endcsname}}\relax
         \fixednumcount=\the\csname idcount#1\endcsname
         \expandafter\xdef\csname #2-sortkey\endcsname{#1,\fixednum,}\relax
         \immediate\write\MyIdsOut
            {\string\MyIdEntry,#2,\csname #2-sortkey\endcsname
                ,\string\textbf{#1\the\csname idcount#1\endcsname
                }\string\myend}\relax
         \global\advance\csname idcount#1\endcsname by 1\relax
      \else
         \ifnewidPrevDefOk
         \else
            \mywarning{ITL id #2 already defined}\relax
         \fi
      \fi
   \else  % Handle the optional base=...
      \expandafter\ifx\csname \IdBase\endcsname\relax
         \mywarning{Undefined ITL id \IdBase\space needed in def of #2}%
         \JustDidXreffalse  % Force \WriteIdUndefined to produce xref
         \WriteIdUndefined{\IdBase}%
         \def\IdBase{---TempBaseFor#2}%
         \forceidpagenum{#2}%
      \fi
      \expandafter\xdef\csname #2\endcsname{\noexpand\csname
         \IdBase\noexpand\endcsname$\noexpand\boldsymbol{'}$}\relax
      \expandafter\xdef\csname #2-sortkey\endcsname
                     {\csname \IdBase-sortkey\endcsname'}\relax
      \immediate\write\MyIdsOut
         {\string\MyIdEntry,#2,\csname #2-sortkey\endcsname
             ,\expandafter\string\csname \IdBase\endcsname
                   $\string\boldsymbol{'}$\string\myend}\relax
   \fi
}

\def\newidUseHardcopyDef#1{{\relax
   \csnamelet{\tempmac}{#1-hardcopy.triple.1}\relax
   \temptoks=\expandafter{\tempmac}\relax
   \csnamelet{\tempmac}{#1-hardcopy.triple.2}\relax
   \temptoksa=\expandafter{\tempmac}\relax
   \csnamelet{\tempmac}{#1-hardcopy.triple.3}\relax
   \temptoksb=\expandafter{\tempmac}\relax
   \csnamelet{\tempmac}{#1-hardcopy.latex-code}\relax
   \temptoksc=\expandafter{\tempmac}\relax

   \immediate\write\MyIdsOut
        {\string\MyIdEntry,#1,\the\temptoks
         ,\the\temptoksa,\the\temptoksb
         ,\the\temptoksc\noexpand\myend}\relax

   \expandafter\xdef\csname #1\endcsname{\the\temptoksc}\relax
   \expandafter\xdef\csname #1-sortkey\endcsname
        {\the\temptoks,\the\temptoksa,\the\temptoksb,\the\temptoksc}\relax
}}

\def\dummyid#1{\expandafter\gdef\csname #1-dummy-id\endcsname{dummy-id}}

\newcount\fixednumcount
\def\fixednum{%
   \ifnum 100>\fixednumcount 0\ifnum 10>\fixednumcount 0\fi\fi
   \the\fixednumcount}

\newif\ifdefidPageNo
\newif\ifdefidRef
\newif\ifdefidNoXrefs
\newif\ifdefidBold
\makeatletter
\def\defid{\@ifnextchar[\defidAux{\defidAux[]}}
\define@key{defid}{pagenum}[true]{\csname defidPageNo#1\endcsname}
\define@key{defid}{ref}[true]{\csname defidRef#1\endcsname}
\define@key{defid}{noxrefs}[true]{\csname defidNoXrefs#1\endcsname}
\define@key{defid}{bold}[true]{\csname defidBold#1\endcsname}
\makeatother

\def\defidAux[#1]#2#3{\defidPageNofalse
   \defidReffalse
   \defidNoXrefsfalse
   \defidBoldtrue
   \setkeys{defid}{#1}%
   \defidAuxB#2,#3;%
   \ifdefidRef
      \ifdefidPageNo
         \refid[pagenum]{#2}
      \else
         \refid{#2}
      \fi
   \else
      \ifdefidPageNo \forceidpagenum{#2}\fi
   \fi
   \ifdefidNoXrefs
      \expandafter\gdef\csname #2-disable-xrefs\endcsname{1}
   \fi
}

\def\defidAuxB#1,#2,#3,#4;{\ifdefidBold
        \def\tempid{\textbf{#2#3#4}}%
      \else
        \def\tempid{#2#3#4}%
      \fi
      \expandafter\ifx\csname #1-ref\endcsname\tempid
      \else
         \global\myidwarningstrue
         \mywarning{Changed ITL id in \string\defid: \noexpand #1}%
      \fi
      \global\expandafter\let\csname defidDef-#1\endcsname\tempid
      \fixednumcount=0#3  % Save number part for use with \fixednum
      \expandafter\xdef\csname #1-sortkey\endcsname{#2,\fixednum,#4}%
      \immediate\write\MyIdsOut
         {\string\MyIdEntry,#1,#2,\fixednum
                ,#4,\ifdefidBold\string\textbf\fi{#2#3#4}\string\myend}%
}

\def\defidthree#1#2#3{\message{[#1\space #2\space #3]}\defid[#1]{#2}{#3}}

\def\mynote#1{\ifdraft\ifmynote\marginpar{\raggedright
   \footnotesize *** #1 ***}\fi\fi}
\def\MyIdEntry,#1,#2,#3,#4,#5\myend{\relax
   \expandafter\gdef\csname #1-ref\endcsname{#5}}
\def\MyIdPageNum,#1,#2,{{\relax
   \expandafter\ifx\csname #1-iddef-max\endcsname\relax
      \expandafter\xdef\csname #1-iddef-max\endcsname{0}%
   \fi
   \tempcount=\csname #1-iddef-max\endcsname
   \advance \tempcount by 1 %
   \expandafter\xdef\csname #1-iddef-max\endcsname{\the\tempcount}%
   \expandafter\xdef\csname #1-iddef.\the\tempcount\endcsname{#2}%
}}

\def\MyIdXrefPageNum,#1,#2,#3,{}
\def\MyIdXrefInProof,#1,#2,#3,{}

\def\MyIdUndefined,#1,#2,#3,{}

\def\mywarning#1{\global\mywarningstrue
   \typeout{^^JMy Warning: #1 on line \the\inputlineno.^^J}\relax}

\newif\ifNeedEndingLineOfStars
\NeedEndingLineOfStarsfalse
\def\LineOfStars{\typeout{**********************************************************************}}
\AtEndDocument{\relax
\ifPrintRefIdXrefs
\else
   \ifmylmcsfinalversion
   \else
     \LineOfStars
     \typeout{***\space\space\space\space\space\space\space\space\space\space
        Cross references, etc. of refids not printed.
        \space\space\space\space\space\space\space\space ***}
     \NeedEndingLineOfStarstrue
   \fi
\fi
\ifmylmcsfinalversion
\else
  \ifUseHardcopyIds
     \LineOfStars
     \typeout{*** Hardcopy feature on and should be disabled before new printout ***}
     \NeedEndingLineOfStarstrue
  \else
     \LineOfStars
     \typeout{*** Hardcopy feature off but should be enabled until next printout ***}
     \NeedEndingLineOfStarstrue
  \fi
\fi
\ifmyidwarnings
   \LineOfStars
   \typeout{***\space\space\space\space
        ITL id definitions missing/changed.
        Rerun through Latex\space\space\space\space\space ***}
   \NeedEndingLineOfStarstrue
\fi
\ifmywarnings
   \LineOfStars
   \typeout{***\space\space\space\space\space\space\space
                \space\space\space\space\space\space\space
                \space\space\space\space\space
                 Some of my warnings issued\space\space\space
                \space\space\space\space\space\space\space
                \space\space\space\space\space\space\space\space\space***}
\fi
\ifNeedEndingLineOfStars
   \LineOfStars
\fi
}

\newbox\DisplayRulebox

\def\DisplayRule#1{\ifmynote\rlap{\hspace{5in}\DisplayInsertHyphens{#1}}\fi}

\def\DisplayInsertHyphens#1{\relax
   \savebox{\DisplayRulebox}{\relax
      \parbox[t]{1.8in}{
      \InsertHyphens{#1}}} %
   \dp\DisplayRulebox=0pt %
   \usebox{\DisplayRulebox} %
}

\def\InsertHyphens#1{{\relax
   \hyphenpenalty=0 %Don't warn about hyphenation problems
   \hbadness=10000   % Don't warn about underflow hboxes
   \newif\ifFirstChar  \FirstChartrue
   \InsertHyphensAux#1|\relax
}}

\def\InsertHyphensAux#1{\if #1|\let\InsertHyphensAuxB=\relax
   \else
      \ifFirstChar
         #1\relax
         \global\FirstCharfalse
      \else
         \expandafter\ifx\csname MyHyphen-#1\endcsname\relax
            #1\relax
         \else
            \csname MyHyphen-#1\endcsname
         \fi
      \fi
      \let\InsertHyphensAuxB=\InsertHyphensAux
   \fi
   \InsertHyphensAuxB}

\def\DefineMyHyphenNames#1{\if #1|\let\DefineMyHyphenNamesAux=\relax
   \else \expandafter\def\csname MyHyphen-#1\endcsname{\-#1}
      \let\DefineMyHyphenNamesAux=\DefineMyHyphenNames\fi
   \DefineMyHyphenNamesAux}
\DefineMyHyphenNames ABCDEFGHIJKLMNOPQRSTUVWXYZ|

\def\TheoremPrefix{T}
\def\DerivedRulePrefix{DR}
\def\showsomething[#1]#2#3{\gdef\MyCurrentStateItId{#3}\relax
   \ifundefined{#3}\else \llap{*}\fi
   \DisplayRule{#3}\relax
   \setbox\tempbox=\hbox{\newid[prefix=#2,pagenum,ref,prevdef-ok,#1]{#3}}\relax
   \myhyperlink{symtab.#3.0}{\box\tempbox}\relax
}
\def\showtheorem[#1]#2{\showsomething[#1]{\TheoremPrefix}{#2}}
\def\showderivedrule[#1]#2{\showsomething[#1]{\DerivedRulePrefix}{#2}}

\def\stateit[#1]#2{\begingroup
  \temptoks={#2}% needed for MyThms support
  \immediate\write\MyThmsOut
    {\noexpand\MyThmsOne \the\temptoks\noexpand\EndMyThmsOne}\relax
\useslashslash
\tabskip=0pt %
\def\omfil{\omit\hfil}  % See usage explained for \statetheorem, etc.
\vspace{7pt}
\halign{\showit[#1]{##}\hfil\quad\enskip &$##$\hfil\cr
#2\cr}\endgroup\unskip}

\makeatletter
\def\statetheorem{\@ifnextchar[\statetheoremAux{\statetheoremAux[]}}
\def\statetheoremAux[#1]#2{\let\showit=\showtheorem
   \stateit[#1]{#2}\ignorespaces}
\def\statederivedrule{\@ifnextchar[\statederivedruleAux
   {\statederivedruleAux[]}}
\def\statederivedruleAux[#1]#2{\let\showit=\showderivedrule
   \stateit[#1]{#2}\ignorespaces}
\makeatother

\makeatletter
\let\pagenum=\c@page
\makeatother

{\begingroup

\def\ReadMyIds#1{\relax
   \newread\MyIdsIn \openin\MyIdsIn = \jobname.#1
   \ifeof\MyIdsIn
      \closein\MyIdsIn
   \else
      \closein\MyIdsIn
      \input{\jobname.#1}
   \fi
}

\input{\jobname.myids}

\def\MyIdEntry,#1,#2,#3,#4,#5\myend{\tempcount=#3\relax
   \expandafter\xdef\csname #1-printed-ref\endcsname
      {#2\ifnum\tempcount>0 \the\tempcount \fi #4}\relax
}

\def\MyIdPageNum,#1,#2,{
   \ifundefined{#1-printed-pagenums}\relax
      \expandafter\xdef\csname #1-printed-pagenums\endcsname{#2}
   \else
      \expandafter\xdef\csname #1-printed-pagenums\endcsname%
                                {\csname #1-printed-pagenums\endcsname,#2}
   \fi
}

\def\MyIdUndefined,#1,#2,#3,{}

\ifprintmyids
   \input{\jobname.myids_p}
\fi

\def\MyIdEntry,#1,#2,#3,#4,#5\myend{\relax
   \expandafter\gdef\csname #1-hardcopy.triple.1\endcsname{#2}\relax
   \expandafter\gdef\csname #1-hardcopy.triple.2\endcsname{#3}\relax
   \expandafter\gdef\csname #1-hardcopy.triple.3\endcsname{#4}\relax
   \expandafter\gdef\csname #1-hardcopy.latex-code\endcsname{#5}\relax
}
\def\MyIdPageNum,#1,#2,{}

\ifUseHardcopyIds
   \input{\jobname.myids_p}
\fi

\endgroup}

\newwrite\MyIdsOut \immediate\openout\MyIdsOut = \jobname.myids

\ifUseHardcopyIds
   \def\MyThmsSuffix{mythms_p}
\else
   \def\MyThmsSuffix{mythms}
\fi
\def\MyThmsUsedInTextBody#1\EndMyThmsUsedInTextBody{\expandafter\gdef
  \csname #1-used-in-mythms\endcsname{1}}
\def\MyThmsMacro#1\EndMyThmsMacro{\def\MyThmsMacroA{#1}} %
\let\MyThmsMacroA=\relax
\IfFileExists{\jobname.\MyThmsSuffix}{\input{\jobname.\MyThmsSuffix}}{} %

\begin{document}

\newif\ifInTextBody  % needed for MyThms support
\InTextBodyfalse % needed for MyThms support
\newwrite\MyThmsOut \immediate\openout\MyThmsOut = \jobname.mythms %
\ifx\MyThmsMacroA\relax
\else
%   \message{1: Found MyThmsMacroA!!} %
\fi

\ifmyacm
  \bibliographystyle{acmtrans}
\else
  \bibliographystyle{alpha}  %% Preferred for LMCS (see lmcs.cls)
  \ifmylmcs
    \let\citeyear=\cite
  \fi
\fi

%% \title[short title]{full title}
\title[A Complete Axiom System for Propositional Interval Temporal Logic]%
{A Complete Axiom System for 
         Propositional Interval Temporal Logic with Infinite Time}
\author[B.~Moszkowski]{Ben Moszkowski}
\address{Software Technology Research Laboratory,
%%   Bede Island Building,
   De Montfort University, Leicester, UK}


\begin{abstract}
  \emph{Interval Temporal Logic ($\ITL$)} is an established temporal formalism
  for reasoning about time periods.  For over 25 years, it has been applied in
  a number of ways and several $\ITL$ variants, axiom systems and tools have
  been investigated.  We solve the longstanding open problem of finding a
  complete axiom system for basic quantifier-free propositional $\ITL$
  ($\PITL$) with infinite time for analysing nonterminating computational
  systems.  Our completeness proof uses a reduction to completeness for
  $\PITL$ with finite time and conventional propositional linear-time temporal
  logic.  Unlike completeness proofs of equally expressive logics with
  nonelementary computational complexity, our semantic approach does not use
  tableaux, subformula closures or explicit deductions involving encodings of
  omega automata and nontrivial techniques for complementing them.  We believe
  that our result also provides evidence of the naturalness of interval-based
  reasoning.
\end{abstract}

\maketitle

\InTextBodytrue % needed for MyThms support

\section{Introduction}

\label{introduction-sec}

Intervals and discrete linear state sequences offer a natural and flexible way
to model both sequential and parallel aspects of computational processes
involving hardware or software.  \emph{Interval Temporal Logic
  ($\ITL$)}~\cite{Moszkowski86} (see also~\cite{ITLwebsite}) is an established
formalism for rigorously reasoning about such intervals.  $\ITL$ has a basic
construct called \emph{chop} for the sequential composition of two arbitrary
formulas as well as an analogue of Kleene star for iteration called
\emph{chop-star}.  Although originally developed for digital hardware
specification~\cite{Moszkowski83a,Moszkowski83,HalpernManna83,Moszkowski85},
$\ITL$ is suitable for logic-based executable
specifications~\cite{Moszkowski86}, compositional reasoning about concurrent
processes~\cite{Moszkowski94,Moszkowski95a,Moszkowski98,Moszkowski11-TIME},
refinement~\cite{CauZedan97}, as well as for runtime
analysis~\cite{ZhouZedan99}.

Until now, in spite of research over many years involving $\ITL$ and its
applications, there was no known complete axiom system for quantifier-free
propositional ITL ($\PITL$) with infinite time.  We present one and prove
completeness by a reduction to our earlier complete $\PITL$ axiom system for
finite time~\cite{Moszkowski04a} (see also~\cite{BowmanThompson03}) and
conventional propositional linear-time temporal logic ($\PTL$).  We do not use
subformula closures, tableaux, or explicit deductions involving encodings of
omega automata and nontrivial techniques for complementing them.  Such encodings
are typically found in completeness proofs for comparable logics discussed later
on (see \S\ref{omega-regular-logics-with-nonelementary-complexity-subsec}),
which like $\PITL$ have omega-regular expressiveness.  See
Thomas~\citeyear{Thomas90,Thomas97} for more about omega-regular languages,
omega automata and some associated logics. Our simple axiom system avoids
complicated inference rules and proofs such as axiom systems for an equally
expressive version of $\PITL$ with restricted sequential
iteration~\cite{Paech89} and a less expressive version of $\PITL$ lacking
sequential iteration~\cite{RosnerPnueli86}.  In the future we plan to use our
axiom system as a hierarchical basis for obtaining completeness for some $\PITL$
variants.  We also believe it can be applied to some other logics and discuss
this in Section~\ref{future-work-sec}.

Our earlier completeness proof for a larger, more complicated axiom system for
quantified $\ITL$ with finite domains and infinite time~\cite{Moszkowski00}
does not work if variables are limited to being just propositional.  So that
result, while serving as a stepping stone for further research on $\ITL$, even
fails to establish axiomatic completeness for a quantified version of $\PITL$
($\QPITL$) with infinite time!  For these reasons, we feel justified in
regarding the problem of showing axiomatic completeness for full $\PITL$ with
infinite time as a previously open problem.
% Interestingly, a much simpler companion axiomatic completeness result of
% ours~\cite{Moszkowski00a} for quantified $\ITL$ with finite domains and just
% finite time does not appear to have problems if only propositional variables
% are permitted.

We now mention some recent publications by others as evidence of $\ITL$'s
continuing relevance.  None specifically motivate our new completeness proof.
Nevertheless, they arguably contribute to making a case for the study of
$\ITL$'s mathematical foundations, which naturally include axiomatic
completeness.

The KIV interactive theorem prover~\cite{ReifSchellhorn98} has for a number of
years included a slightly extended version of $\ITL$ for interactive theorem
proving via symbolic execution both by itself (e.g.,
see~\cite{BaeumlerBalser10,BaeumlerSchellhorn09}) and also as a backend
notation which supports Statecharts~\cite{ThumsSchellhorn04} and
UML~\cite{BalserBaeumler04}.  KIV can employ ITL proof systems such as ours.
The concluding remarks of~\cite{BaeumlerSchellhorn09} note the following
advantages of ITL:
\begin{quote}
  Our $\ITL$ variant supports classic temporal logic operators as well as
  program operators.

  The interactive verifier KIV allows us to directly verify parallel programs
  in a rich programming language using the intuitive proof principle of
  symbolic execution. An additional translation to a special normal form (as
  e.g.\ in TLA [Temporal Logic of Actions~\cite{Lamport02}]) using explicit
  program counters is not necessary.
\end{quote}
Axiomatic completeness of $\PITL$ is not an absolute requirement for the KIV
tool but does offer some benefits.  This is because some axioms, inference rules
and associated deductions employed to prove completeness can be exploited in
KIV, thereby reducing the number of adhoc axioms and inference
rules.\footnote{Our claim is supported by email correspondence in 2011 with
  Gerhard Schellhorn of the KIV group.}
%% See following files for the email in directory ~/src/papers/schellhorn/:
%%  schellhorn-discussion-about-particular-itl-proof-email.Feb11.txt
%%  schellhorn-discussion-about-kiv-and-axiomatic-completeness-email.Sep11.txt

Various imperative programming constructs are expressible in $\ITL$ and
operators for projecting between time granularities are available (but not
considered here).  $\ITL$ influenced an assertion language called
\emph{temporal `e'}~\cite{Morley99} which is part of the IEEE Standard
1647~\cite{IEEE1647\LMCSPaperBib} for the system verification language
\emph{`e'}.

The \emph{Duration Calculus ($\DC$)} of Zhou, Hoare and
Ravn~\citeyear{ZhouHoare91} is an $\ITL$ extension for real-time and hybrid
systems.  The books by Zhou and Hansen~\citeyear{ZhouHansen2004} and Olderog
and Dierks~\citeyear{OlderogDierks2008} both employ $\DC$ with finite time and
discuss relatively complete axiom systems for it.  The second book utilises
$\DC$ with timed automata to provide a basis for specifying, implementing and
model checking suitable real-time systems. Indeed, Olderog and Dierks explain
how they regard an interval-oriented temporal logic as being better suited for
these tasks than more widely used point-based ones and timed process algebras.
Concerning point-based logics, they make this comment (on page 23): ``In our
opinion this leads to complicated reasoning similar to that \ldotss based on
predicate logic.''  As for timed process algebras, they note the following (on
page 25): ``A difficulty with these formalisms is that their semantics are
based on certain scheduling assumptions on the actions like urgency, which are
difficult to calculate with.''

Within the last ten years, other complete axiom systems for versions of
propositional and first-order $\ITL$ with infinite time have been presented.
These include two by Wang and Xu~\citeyear{WangXu2004} for first-order
variants with restricted quantifiers and no sequential iteration as well as a
probabilistic extension of theirs by Guelev~\citeyear{Guelev2007} which all
build on an earlier completeness result of Dutertre~\citeyear{Dutertre95} for
first-order $\ITL$ restricted to finite time.  Like Dutertre, Wang and Xu and
also Guelev use a nonstandard abstract-time semantics (e.g., without induction
over time) instead of $\ITL$'s standard discrete-time one.  Their proofs
employ Henkin-style infinite sets of maximal consistent formulas.  Duan et
al.~\citeyear{DuanZhang08,DuanZhang12} give a tableaux-like completeness proof
for a related omega-regular logic called \emph{Propositional Projection
  Temporal Logic (PPTL)}. The only primitive temporal operators in PPTL for
sequential composition have varying numbers of operands and concern multiple
time granularities. However, both chop and chop-star can be derived. The proof
system has over 30 axioms and inference rules, some rather lengthy and
intricate.  The completeness proof itself involves the nontrivial task of
complementing omega-regular languages which can be readily expressed in the
logic but it is not discussed.  Furthermore, the authors omit much of the
prior work in the area developed in the course of over forty years (which we
later survey in
Section~\ref{existing-completeness-proofs-for-omega-regular-logics-sec}).
More significantly, they do not explain how they bypass the associated hurdles
faced by previous completeness proofs for logics with comparable
expressiveness and nonelementary computational complexity.  These points make
checking the proof's handling of the complementation of omega-regular
languages, liveness and other issues rather challenging.  Mo, Wang and
Duan~\citeyear{MoWang2011} describe promising applications of Projection
Temporal Logic to specifying and verifying asynchronous communication.  Zhang,
Duan and Tian~\cite{ZhangDuan2012} investigate the modelling of multicore
systems in Projection Temporal Logic.  In view of this, the foundational issue
of axiomatic completeness for $\textrm{PPTL}$ should be addressed in the
future more thoroughly and systematically and better related to other
approaches. Incidentally, we already showed in~\cite{Moszkowski95a} that
axiomatic completeness for a version of $\PITL$ with a standard version of
temporal projection can be simply and hierarchically reduced to axiomatic
completeness for $\PITL$ without temporal projection.  Duan et
al.~\citeyear{DuanZhang08,DuanZhang12} however make no mention of this by now
long established and powerful technique in their review of prior work.

% $\PITL$ with infinite time can express the \emphbf{omega-regular languages}, a
% natural generalisation of regular languages to infinite words introduced by
% B\"uchi~\cite{Buechi62} (see also Thomas~\cite{Thomas90}).  For a formal
% definition of such languages, we take $\Sigma$ to be some finite alphabet.
% Now let the set $\Sigma^\omega$ contain exactly the infinite words with
% letters in $\Sigma$.  Any language $L\subseteq \Sigma^\omega$ is called an
% \emphbf{omega-language}.  It is \emphbf{omega-regular} iff there exists some
% $n\ge 0$ and regular languages $U_1$, $V_1$, \ldots, $U_n$, $V_n$ such that $L
% = \bigcup_{1\le i\le n} (U_i\cdot V_i^\omega)$, where $U\cdot U'$ denotes
% pairwise string concatenation of elements in regular languages $U$ and $U'$
% and $U^\omega$ denotes the omega-language obtained by $\omega$
% concatenations of regular language $U$'s elements.
%% Lichtenstein et al.~\cite{LichtensteinPnueli85}, Emerson~\cite{Emerson90} and

Here is the structure of the rest of this presentation: Section~\ref{pitl-sec}
overviews $\PITL$ and the new axiom system.
Section~\ref{right-instances-right-variables-and-right-theorems-sec} concerns
a class of $\PITL$ theorems from which we can also deduce suitable
substitution instances needed later on.
Section~\ref{some-lemmas-for-replacement-sec} gives some infrastructure for
systematically replacing formulas by other equivalent ones in deductions
arising in the completeness proof.  Section~\ref{useful-subsets-of-pitl-sec}
introduces some useful $\PITL$ subsets for later use in the completeness
proof.  Section~\ref{reduction-of-chop-omega-sec} reduces completeness for
$\PITL$ with a kind of infinite sequential iteration to completeness for a
subset without this.
Section~\ref{deterministic-semi-automata-and-automata-sec} shows how to
represent deterministic finite-state semi-automata and automata in $\PITL$.
Section~\ref{compound-semi-automata-for-suffix-recognition-sec} employs
semi-automata to test a given $\PITL$ formula in a finite interval's suffix
subintervals.  Section~\ref{reduction-of-pitl-to-ptlu-sec} shows completeness
for the $\PITL$ subset without infinite sequential iteration.
Section~\ref{some-observations-about-the-completeness-proof-sec} includes some
observations about the completeness proof.
Section~\ref{existing-completeness-proofs-for-omega-regular-logics-sec}
reviews existing complete axiom systems for omega-regular logics.
Section~\ref{future-work-sec} discusses some topics for future research.

\section{Propositional Interval Temporal Logic}

\label{pitl-sec}

We now describe the version of (quantifier-free) $\PITL$ used here.  More on
basic aspects of $\ITL$ can be found in \cite{Moszkowski83a,HalpernManna83,%
  Moszkowski85,Moszkowski86,Moszkowski04a} (see also Kr{\"o}ger and
Merz~\citeyear{KroegerMerz08}, Fisher~\cite{Fisher11} and the ITL
web pages~\cite{ITLwebsite}).

Below is the syntax of $\PITL$ formulas in BNF, where $p$ is any propositional
variable:
\begin{displaymath}
A ::= \,
\True \,\mid\,
p \,\mid\,
\Not A \,\mid\,
A \Or A \,\mid\,
\Skip \,\mid\,
A\chop A \,\mid\, %\enskip\mbox{{(\em chop\/)}}\quad
%A;B \quad %\enskip\mbox{{(\em chop\/)}}\quad
A\SChopstar %\enskip\mbox{{(\em chop-star\/)}}
\dotspace.
\end{displaymath}
%% We include $\True$ as a primitive so as to avoid a definition of it which
%% contains some specific variable.  This is not strictly necessary.
The last two constructs are called \emphbf{chop} and \emphbf{chop-star},
respectively.  The boolean operators $\False$, $A\And B$, $A\imp B$
(\emph{implies}) and $A\equiv B$ (\emph{equivalence}) are defined as usual.
We refer to $A\chop B$ as \emphbf{strong chop}, since a weak version $A;B$
also exists.  In addition, $A\SChopstar$ (\emphbf{strong chop-star}) slightly
differs from $\ITL$'s conventional \emphbf{weak chop-star}
$A\ConventionalChopstar$, although the two are interderivable.  The strong
variants of chop and chop-star taken as primitives here are chosen simply
because, without loss of generality, they help streamline the completeness
proof.

We use $p$, $q$, $r$ and variants such as $p'$ for propositional variables.
Variables $A$, $B$, $C$ and variants such as $A'$ denote arbitrary $\PITL$
formulas.  Let $w$ and $w'$ denote \emphbf{state formulas} without the
temporal operators $\Skip$, chop and chop-star.  We have $V$ denote a finite
set of propositional variables.  Also, $V_A$ denotes the finite set of the
formula $A$'s variables.

Time within $\PITL$ is discrete and linear.  It is represented by
\emph{intervals} each consisting of a sequence of one or more states.  More
precisely, an interval $\sigma$ is any finite or $\omega$-sequence of one or
more states $\sigma_0$, $\sigma_1$, \ldots.  Each state $\sigma_i$ in $\sigma$
maps each propositional variable $p$ to either $\True$ and $\False$.  This
mapping is denoted as $\sigma_i(p)$.  An interval $\sigma$ has an
\emph{interval length} $\intlen{\sigma}\ge 0$, which, if $\sigma$ is finite,
is the number of $\sigma$'s states minus 1 and otherwise $\omega$. So if
$\sigma$ is finite, it has states $\sigma_0$, \ldots,
$\sigma_{\intlen{\sigma}}$.  This (standard) version of $\PITL$, with
state-based propositional variables, is called \emphbf{local PITL}.

A \emphbf{subinterval} of $\sigma$ is any interval which is a
\emph{contiguous} subsequence of $\sigma$'s states. This includes $\sigma$
itself.

The notation $\sigma \vld A$, defined shortly by induction on $A$'s syntax,
denotes that interval $\sigma$ \emphbf{satisfies} formula $A$.  Moreover, $A$
is \emphbf{valid}, denoted $\vld A$, if all intervals satisfy it.

Below are the  semantics of the first five constructs:
%$\PITL$ except ``$\chop$'' and ``*'':
\begin{iteMize}{$\bullet$}
  
\item True: $\sigma \vld \True$ trivially holds for any $\sigma$.

\item Propositional variable: $\sigma \vld p \iff p$ is true in the initial
  state $\sigma_0$ (i.e., $\sigma_0(p)=\True$).

\item Negation: $ \sigma\vld \Not A  \iff \sigma \not\vld A $.

\item Disjunction: $ \sigma\vld A \Or B \iff
  \sigma\vld A \text{ or } \sigma\vld B $.
  
\item Skip:
  $ \sigma\vld \Skip \iff $
$\sigma$ has exactly two states.
\end{iteMize}
For natural numbers $i$, $j$ with $0\le i\le j\le\intlen{\sigma}$, let
$\sigma_{i:j}$ be the finite subinterval $\sigma_i \ldots \sigma_j$ (i.e.,
$j-i+1$ states). Define $\sigma_{i\uparrow}$ to be $\sigma$'s suffix
subinterval from state $\sigma_i$.

Below are semantics for the versions of chop and chop-star found most suitable
for the completeness proof.  As already noted, other versions can be readily
derived.

\begin{iteMize}{$\bullet$}
\item Chop: $ \sigma\vld A\chop B \iff $ for some natural number $i:0\le
  i\le\intlen{\sigma}$, both $\sigma_{0:i}\vld A$ and $\sigma_{i\uparrow}\vld
  B$.  This is called \emph{strong chop} because both $A$ and $B$ must be
  true.

% Let $\sigma\vld A\chop B$ be true iff either of the following holds:
% \begin{iteMize}
% \item For some $0\le i\le\intlen{\sigma}$, both $\sigma_{0:i}\vld A$ and
%   $\sigma_{i\uparrow}\vld B$.
% \item The interval $\sigma$ itself is infinite and $\sigma\vld A$ holds.
% \end{iteMize}
% In the first case, the subintervals share state $\sigma_i$.  We call this chop
% \emph{weak} and later define a \emph{strong} version.

\item Chop-star: $ \sigma \vld A\SChopstar \iff $ one of the following
  holds:
  \begin{iteMize}{$-$}
  \item Interval $\sigma$ has only one state (i.e., it is \emph{empty}).
  \item $\sigma$ is finite and either itself satisfies $A$ or can be split
    into a finite number of (finite-length) subintervals which share
    end-states (like chop) and all satisfy $A$.
  \item $\intlen{\sigma}=\omega$ and $\sigma$ can be split into $\omega$
    finite-length intervals sharing end-states (like chop) and each satisfying
    $A$.
  \end{iteMize}
  In this version of chop-star, each iterative subinterval has finite
  length. The third case is called \emphbf{chop-omega} and denoted as
  $A\Chopomega$.
\end{iteMize}

As an example, we depict the behaviour of variable $p$ in some 5-state
interval $\sigma$ and denote $\True$ and $\False$ by \texttt{t} and
\texttt{f}, respectively:
\begin{displaymath}
  \begin{array}{cccccc}
     & \sigma_0 & \sigma_1 & \sigma_2 & \sigma_3 & \sigma_4 \\\hline
     p & \texttt{t} & \texttt{f} & \texttt{t} & \texttt{f} & \texttt{t}
  \end{array}    
\end{displaymath}
This interval satisfies the following formulas:
\begin{displaymath}
  p \qquad \Skip\chop\Not p  \qquad p \And (\True\chop\Not p)
 \qquad (p \And (\Skip\chop\Skip))\SChopstar
\dotspace.
\end{displaymath}
For instance, the formula $\Skip\chop\Not p$ is true because
$\sigma_0\sigma_1$ satisfies $\Skip$ and $\sigma_1\ldots\sigma_4$ satisfies
$\Not p$ since $\sigma_1(p)=\False$.  The fourth formula is true because both
$\sigma_0\ldots\sigma_2$ and $\sigma_2\ldots\sigma_4$ satisfy $p\And
(\Skip\chop\Skip)$.  The interval does not satisfy the formulas below:
\begin{displaymath}
  \Not p \qquad \Skip\chop p
 \qquad \True\chop (\Not p \And \Not(\True\chop p))
\dotspace.
\end{displaymath}

\begin{mycomment}
Figure~\ref{informal-ITL-fig} pictorially illustrates the semantics of chop
and chop-star in both finite and infinite time and also shows some simple
$\PITL$ formulas using derived operators in
Table~\ref{pitl-derived-operators-table} together with finite intervals which
satisfy the formulas.  For some sample formulas we include in parentheses
versions using $\PTL$ operators.
\begingroup
\hyphenpenalty=10000 %
\begin{figure*}
\begin{center}
   \def\xdots{\scalebox{2.0}{\ldots}} %
   \setbox\tempbox=\hbox{\scalebox{0.75}{\relax
            \input{semantics-of-ITL-for-infinite-time-manna-journal-submission.half-pstex_t}}} %
   \subfigure[Informal semantics for finite time]{\vbox to \ht\tempbox{\relax
      \hbox{\scalebox{0.75}{\relax
            \input{semantics-of-ITL-for-finite-time-manna-journal-submission.half-pstex_t}}}\vfil}} %
      \hspace{1cm}
   \subfigure[Informal semantics for infinite time]{\copy\tempbox} %
      \hspace{1cm}
   \subfigure[Some finite-time examples]{\relax
      \scalebox{0.85}{\input{sample-ITL-formulas-for-manna-journal-submission.half-pstex_t}}}
\end{center}
\caption{Informal $\PITL$ semantics and examples}
\label{informal-ITL-fig}
\end{figure*}
\endgroup
\end{mycomment}

\begin{mycomment}
\begin{myremark}
  The behaviour of chop-star on empty intervals is a frequent source of
  confusion and it is therefore important to note that any formula
  $A\SChopstar$ (including $\False\SChopstar$) is true on a one-state interval.
  This is because in the semantics of chop-star for a one-state interval we
  can always set $n=0$ and therefore ignore the values of variables in the
  interval.
\end{myremark}
\end{mycomment}

\begin{table*}%[h]
  \begin{center}
    \begin{tabular}{L@{$\Defeqv$}Ll}
       \Next A & \Skip\chop A & Next \\
       \Diamond A & \True\chop A & Eventually \\
       \Box A & \Not\Diamond\Not A & Henceforth \\
       \More & \Next\True & More than one state \\
       \Empty & \Not\More & Only one state \\
       \Finite & \Diamond\Empty & Finite interval \\
       \Inf & \Not\Finite & Infinite interval \\
%       \SFin A & \Diamond(\Empty \And A) & Strong test of final state \\
       \Fin A & \Box(\Empty \imp A)& Weak test of final state \\
       A \Tassign B & \Finite \imp ((\Fin A)\equiv B) & Temporal assignment \\
       \Df\! A & A\chop\True & Some initial finite subinterval \\
       \Bf A & \Not\Df\Not A & All initial finite subintervals \\
       A;B & (A\chop B) \;\Or\; (A \And \Inf) & Weak chop \\
       \Di A & A;\True & Some initial subinterval (even infinite)\\
       \Bi A & \Not\Di\Not A & All initial subintervals (including infinite) \\
%  A\Chopkleene & \Empty \,\Or\, \DD{A\SChopstar}A & Finitary chop-star \\
        A\ConventionalChopstar
	  & A\SChopstar \Or \bigl(A\SChopstar \chop (A \And \Inf)\bigr)
	  & Conventional (weak) chop-star \\
        A\Chopomega &  A\SChopstar \And \Inf & Chop-omega
    \end{tabular}
    \caption{Some useful derived PITL operators}
    \label{pitl-derived-operators-table}
  \end{center}
\end{table*}
Table~\ref{pitl-derived-operators-table} shows some useful derived $\PITL$
operators, including the weak versions of chop $A;B$ and chop-star
$A\ConventionalChopstar$.
The derived construct $A\Tassign B$ for temporal assignment in
Table~\ref{pitl-derived-operators-table} perhaps requires some more
explanation. Its purpose is to specify that the value of $A$ in a finite
interval's last state equals the value of $B$ for the interval.  For example,
the formula $p\Tassign \Box q$ is true on an interval iff either (a) the
interval is infinite or (b) it is both finite and has one of the following
hold for the propositional variables $p$ and $q$:
\begin{iteMize}{$\bullet$}
\item The (finite) interval's last state has $p$ true and all states have $q$
  true.
\item The (finite) interval's last state has $p$ false and at least one state
  has $q$ false.
\end{iteMize}

\gdef\MergedAxiomSystem{\relax
\let\currentnametype=\nameaxiom
{\useslashslash\let\theorem=\theoremBEN
\mbox{\relax
\vtop{
\tabskip=0pt%n
\halign{\strut \defidthree{ref,pagenum}##\hfil\hskip 0.25em&$##$\hfil\cr
%\halign{\strut {\bf ##}\hfil\hskip 0.25em&$##$\hfil\cr
%   &\quad{\bf #}\hfil\hskip 0.25em&$#$\hfil\cr
\multicolumn{2}{l}{\strut\emph{Axioms:}} \\
%\multicolumn{2}{l}{\strut\vphantom{\vrule width 0pt height 1.4em}\emph{Axioms:}} \\
{VPTL}{VPTL,,}&
        \theorem \mbox{Substitution instances of valid $\PTL$ formulas}  \\
%\noalign{\vskip 3pt}
{ChopAssoc}{P,2,}&
  \theorem (A\chop B)\chop C \,\equiv\, A\chop (B\chop C) \cr
{OrChopImp}{P,3,}&
  \theorem  (A\Or A')\chop B
      \,\implies\, (A\chop B) \Or (A'\chop B)
%% & \hphantom{\theorem} \implies (A\chop B) \Or (A'\chop B)
 \\
{ChopOrImp}{P,4,}&
 \theorem A\chop (B\Or B')
       \,\implies\, (A\chop B) \Or (A\chop B')
%%  & \hphantom{\theorem} \implies (A\chop B) \Or (A\chop B')
 \\
{EmptyChop}{P,5,}&
  \theorem \Empty\chop A \enskip\equiv\enskip A \\
{FiniteImpChopEmpty}{P,6,}&
  \theorem \Finite \,\implies\, (A\chop\Empty \enskip\equiv\enskip A) \\
% \\
{StateImpBf}{P,7,}&
  \theorem w \implies \Bf w \\
{BfAndBoxImpChopImpChop}{P,8,}&
   \theorem \Bf(A\imp A') \Andd \Box(B\imp B')
     \Implies (A\chop B \implies A'\chop B')
 \\
{SChopStarEqv}{P,9,}&
  \theorem A\SChopstar \EQUIV
      \Empty \Or (A\And \More)\chop A\SChopstar
%%  & \hphantom{\theorem} \enskip\Empty \Or (A\And \More)\chop A\SChopstar
 \\
%% P10& \theorem A\SChopstar \And \Inf \EQUIV \null \\
%%  & \hphantom{\theorem} \quad \DD{A\SChopstar}(A \And\Inf) \,\Or\, A\Chopomega
%% \\
%% For following see notes of 3 July '08, page 1:
{ChopOmegaInduct}{P,10,}&
 \theorem A \And \Box\bigl(A \imp (B \And \More)\chop A\bigr)
       \,\implies\, B\Chopomega
  \\
% P11& \theorem  \begin{myarray}
% 		 \Box((A \imp (B \And \More)\chop A) \\
% 		   \enskip \implies B\SChopstar
%                  \end{myarray}
%  \\
\multicolumn{2}{l}{\strut\vphantom{\vrule width 0pt height 1.4em}\emph{Inference Rules:}} \\
{MP}{MP,,}&
 \theorem A\imp B, \quad \theorem A \infer \theorem B \\
%% ISubst& \theorem A \infer \theorem A_v^B \\
{BfFGen}{$\Bf$FGen,,}&
  \thm \Finite \imp A \infer \theorem \Bf A \\
{BoxGen}{$\Box$Gen,,}&
  \theorem A \infer \theorem \Box A \\
%$\Bi$Gen& \thm A \infer \theorem \Bi A \\
{BfAux}{$\Bf$Aux,,}&
   \thm \Bf\bigl((\Fin p)\equiv B\bigl) \,\implies\, A \infer \thm A \\
% {BoxAux}{$\Box$Aux,,}&
%    \thm \Box(p\equiv B) \,\implies\, A \infer \thm A \quad(!!!!!!??????)\\
\noalign{\vspace{4pt}}
\multicolumn{2}{l}{\strut
   \rlap{In \textbf{$\Bf$Aux}, %and \textbf{$\Box$Aux} (!!!!!!??????),
      propositional variable $p$ must not occur in $A$ or $B$.}}
\cr}}}}}\relax
\begin{table*}[t]
%\begin{center}
\centerline{\MergedAxiomSystem}
%\vspace{-15pt}
%\end{center}
\caption{Axiom system for PITL with finite and infinite time}
\label{axiom-system-for-pitl-with-finite-and-infinite-time-table}
\end{table*}
\begin{table*}[t]
\begin{center}
\bgroup
\def\defidX#1#2{\defidY#2}%
\def\defidY#1,#2,#3{\textbf{#1#2}}%
\mbox{\useslashslash
\qquad
\vtop{
\halign{\strut \defidthree{ref,pagenum}#\hfil\hskip 0.25em&$#$\hfil   \cr
%\halign{\strut \defidX#\hfil\hskip 0.25em&$#$\hfil   \cr
\multicolumn{2}{l}{\strut\emph{Axioms:}} \\
{Taut}{Taut,,}&
        \theoremBEN \mbox{Substitution instances of conventional
	 (nonmodal) tautologies}
%                       (see Def.~\ref{tautology-def})}
 \cr
%\noalign{\vskip 3pt}
{FChopAssoc}{F,2,}&
   \theoremBEN (A\chop B)\chop C \equiv A\chop (B\chop C)\cr
{FOrChopImp}{F,3,}& \theoremBEN (A\Or A')\chop B
   \implies (A\chop B) \Or (A'\chop B) \cr
{FChopOrImp}{F,4,}& \theoremBEN A\chop (B\Or B')
   \implies (A\chop B) \Or (A\chop B') \cr
{FEmptyChop}{F,5,}& \theoremBEN \Empty\chop A \EQUIV A \cr
{FChopEmpty}{F,6,}& \theoremBEN A\chop\Empty \EQUIV A \cr
{FStateImpBf}{F,7,}& \theoremBEN w \implies \Bf w \cr
{FBiAndBoxImpChopImpChop}{F,8,}&\theoremBEN
           \Bf(A\imp A') \Andd \Box(B\imp B')
           \Implies (A\chop B) \imp (A'\chop B')
 \cr
{FChopStarEqv}{F,9,}&\theoremBEN
     A\SChopstar \;\equiv\; 
       \Empty \;\Or\; (A\And \More)\chop A\SChopstar \cr
{FNextImpWeakNext}{F,10,}&\theoremBEN \Next A \implies \WeakNext A \cr
{FBoxInduct}{F,11,}&\theoremBEN A \And \Box(A \imp \WeakNext A) \implies \Box A
\cr
\noalign{\vspace{2pt}}
\multicolumn{2}{l}{\strut\vphantom{\vrule width 0pt height 1.2em}\emph{Inference Rules:}} \\
{FMP}{MP,,}&\theoremBEN A\imp B, \quad \theoremBEN A \infer \theoremBEN B \cr
{FBfGen}{$\Bf$Gen,,}&\thm A \infer \thm \Bf A \\
{FBoxGen}{$\Box$Gen,,}&\theoremBEN A \infer \theoremBEN \Box A \\
\noalign{\vspace{2pt}}
\multicolumn{2}{l}{Note:
  $\WeakNext A \defeqv \Not\Next\Not A$\quad (Weak next)}
\cr}}}
\egroup
\caption{Axiom system for PITL with just finite time}
\label{pitl-axiom-system-for-finite-time-table}
\end{center}
\end{table*}
%\clearpage

% Wang and Xu~\cite{WangXu2004}, Guelev~\cite{Guelev2007} and Duan and Zhang
% al.~\cite{DuanZhang08} use strong chop in axiomatic completeness
% proofs for variants of $\ITL$ with infinite time.  Only Duan and Zhang's
% variant, which is based on operators for multiple time granularities, is
% propositional.

% was earlier used by Chandra et
% al.~\cite{ChandraHalpern81,ParikhChandra85}, Harel and
% Peleg~\cite{ParikhChandra85}. Harel
% and Peleg used it with Thomas' theorem in a proof about a process logic's
% expressiveness.

% A dual of strong chop can be readily defined: $\BB{A}B \defeqv
% \Not\DD{A}\Not B$.  Note that we can express $\Inf$ as
% $\BB{\True}\Not\Df{\Skip}$ and therefore express the conventional weak chop
% in terms of strong chop:
% \begin{displaymath}
%   \valid A\chop B \Equiv \DD{A}B \,\Or\, (A \And \Inf)
% \dotspace.
% \end{displaymath}

Below are some sample valid $\PITL$ formulas:
\bgroup
  \let\valid=\relax
  \let\vld=\relax
\begin{center}
  \begin{tabular}{L}
    (\Finite \And \Bf A) \imp A
      \qquad
      \Skip\SChopstar
      \qquad
      A\SChopstarSChopstar \equiv A\SChopstar
      \qquad (w \And A)\chop B \EQUIV w \And (A\chop B) \\[2pt]
    \Bf(A \And B) \equiv (\Bf A \And \Bf B)
      \qquad (\Box\Bf A) \equiv (\Bf\Box A)
      \qquad  (\Bf\Bf A) \equiv \Bf A \\[2pt]
     \Df A \And \Df B \EQUIV \Df(\Df A \And \Df B)
       \qquad \Bf\bigl((\Fin p)\equiv A) \,\implies\, (\Bf A) \equiv (\Box p)
\dotspace.
  \end{tabular}
\end{center}
\egroup

Let \textbf{$\PTL$} be the subset of $\PITL$ with just $\Skip$ and the
(derived) temporal operators $\Next$ and $\Diamond$ shown in
Table~\ref{pitl-derived-operators-table}. We use $X$ and $X'$ for $\PTL$
formulas.

Although we do not need \emphbf{existential quantification} in our proof, it
is convenient to define here since it helps the exposition concerning
automata-based ways to represent $\PITL$ formulas
in~\S\ref{deterministic-finite-state-automata-subsec},
\S\ref{formal-equivalence-of-the-two-representations-of-runs-subsec}
and~\S\ref{feasibility-of-reduction-from-pitl-to-ptl-subsec} and also assists
us when we compare our approach with related proofs for logics with
quantification in
Section~\ref{existing-completeness-proofs-for-omega-regular-logics-sec}.  The
syntax is $\Exists{p}A$ for any propositional variable $p$ and formula $A$.
We let $\sigma\vld \Exists{p}A$ be true iff $\sigma'\vld A$ is true for some
interval $\sigma'$ identical to $\sigma$ except possibly for $p$'s behaviour.
Existential quantification together with $\PITL$ yields \textbf{$\QPITL$} and
together with $\PTL$ yields \textbf{$\QPTL$}.

\subsection{PITL Axiom System}

Table~\ref{axiom-system-for-pitl-with-finite-and-infinite-time-table} shows
the $\PITL$ axiom system with finite and infinite time.  Axiom~\refid{VPTL}
permits $\PITL$ substitution instances of valid $\PTL$ formulas with $\Skip$,
$\Next$ and $\Diamond$. For instance, from the valid $\PTL$ formula $\Next p
\imp \Diamond p$ follows $\thm \Next A \imp \Diamond A$, for any $\PITL$
formula $A$.  Axiom~\refid{ChopOmegaInduct} gives an inductive way to
introduce chop-omega.
\begin{mycomment}
  In version of the axiom system for $\PITL$ with infinite time but without
  chop-omega, the righthand operand of the equivalence in
  Axiom~\textbf{P10-old} can be shortened to $\DD{A\SChopstar}(A \And\Inf)$ and
  furthermore, Axiom~\textbf{P10} can be totally omitted.  All chop-stars
  would only be allowed when located somewhere within the lefthand side of
  chops.
\end{mycomment}
Our new Inference Rule~\refid{BfAux} permits auxiliary variables to capture
behaviour in finite-length prefix intervals and is only needed for infinite
time.

The axiom system in
Table~\ref{axiom-system-for-pitl-with-finite-and-infinite-time-table} for both
finite and infinite time is adapted from our earlier one~\cite{Moszkowski04a}
for just finite time (see
Table~\ref{pitl-axiom-system-for-finite-time-table}), itself based on a
previous one we originally presented in~\cite{Moszkowski94}.  That axiom
system contains some axioms of Rosner and Pnueli~\citeyear{RosnerPnueli86} for
$\PITL$ without chop-star and our own axioms and inference rule for the
operators $\Bi$ (defined using weak chop in
Table~\ref{pitl-derived-operators-table}) and chop-star.  The new $\PITL$
axiom system in
Table~\ref{axiom-system-for-pitl-with-finite-and-infinite-time-table} adapts
the axioms for $\Bi$ to use $\Bf$ instead to shorten the completeness proof
since $\Bf$ works better with the strong chop operator $\chop$.

For consistency with our usage here, the version of the earlier axiom system
for just finite time given in
Table~\ref{pitl-axiom-system-for-finite-time-table} uses strong chop $\chop$
instead of weak chop ``;'' and likewise uses $\Bf$ instead of $\Bi$. It
therefore very slightly differs from the original one in~\cite{Moszkowski04a}
in an inessential way since for finite time the two pairs of operators are
indistinguishable.  In~\cite{Moszkowski04a} we prove completeness by
reduction to $\PTL$.

Appendix~\ref{some-pitl-theorems-and-their-proofs-sec} contains a large
variety of representative $\PITL$ theorems, derived rules and their proofs.
Many are used directly or indirectly in our completeness proof.

Note that Inference Rule~\refid{BfFGen} in
Table~\ref{axiom-system-for-pitl-with-finite-and-infinite-time-table} for
$\Bf$ mentions $\Finite$ in it, whereas the analogous Inference
Rule~\refid{BoxGen} for $\Box$ does not.  A version of~\refid{BfFGen} without
$\Finite$ and called~\textbf{$\Bf$Gen} can be deduced (see the derived
inference rule~\refid{BfGen} in
Appendix~\ref{some-pitl-theorems-and-their-proofs-sec}).  If just finite time
is permitted, the two variants~\refid{BfFGen} and~\textbf{$\Bf$Gen} for $\Bf$
are in practice identical since $\Finite$ is valid and hence deducible by
Axiom~\refid{VPTL}.  In fact, our earlier axiom system for $\PITL$ with just
finite time in Table~\ref{pitl-axiom-system-for-finite-time-table} uses the
version without $\Finite$.

\begin{mycomment}
  The Inference Rule~\refid{BfAux} in
  Table~\ref{axiom-system-for-pitl-with-finite-and-infinite-time-table} is
  alternatively derivable from an inference rule $\thm A \Infer\unskip\; \thm
  A_{\Fin p}^B$ employing restricted \emph{interval-oriented} substitution of
  the formula $B$ into all instances in $A$ of $\Fin p$ which only occur in
  the left sides of chops but not totally outside chops, in their right sides
  or in chop-stars.  For this alternative inference rule, no other instances
  of $p$ can occur in $A$.  Note that here $\Fin p$ acts exactly like a
  variable which is replaced by $B$ in $A$ so $A_{\Fin p}^B$ requires a slight
  abuse of conventional notation for substitutions.
\end{mycomment}

\subsection{Theoremhood, Soundness and Axiomatic Completeness}

\label{theoremhood-soundness-and-axiomatic-completeness-subsec}

A formula $A$ deducible from the axiom system is a \emphbf{theorem}, denoted
$\thm A$. Additionally, a formula $A$ is \emphbf{consistent} if $\Not A$ is
\emph{not} a theorem, i.e., $\not\thm \Not A$.  We claim the axiom system is
\emphbf{sound}, that is, $\thm A$ \emph{implies} $\vld A$.  A logic is
\emphbf{complete} if each valid formula is deducible as a theorem in the
logic's axiom system.  In other words, if $\vld A$, then $\thm A$.  Our goal
is to show \emphbf{completeness} for $\PITL$.  However, we actually prove a
stronger result which requires some further definitions and we therefore defer
the formal statement until
Theorem~\ref{completeness-of-pitl-axiom-system-thm} in
Section~\ref{right-instances-right-variables-and-right-theorems-sec}.  We also
make use of the following variant way of expressing axiomatic completeness:
\begin{mylemma}[Alternative notion of completeness]
  \label{alternative-completeness-lem}
  A logic's axiom system is complete iff each consistent formula is
  satisfiable.
\end{mylemma}

We often use the next
Theorem~\ref{completeness-of-pitl-axiom-system-for-finite-time-thm} about
finite time:
\begin{mytheorem}[Completeness of $\PITL$ Axiom System for Finite Time]
  \label{completeness-of-pitl-axiom-system-for-finite-time-thm}
  Any valid $\PITL$ implication $\Finite \imp A$ is deducible as a $\PITL$
  theorem $\thm\Finite \imp A$ using the axiom system for $\PITL$ with both
  finite and infinite time in
  Table~\ref{axiom-system-for-pitl-with-finite-and-infinite-time-table}.
\end{mytheorem}
\begin{proof}
  This readily follows by deducing the axioms and inference rules of our
  earlier complete axiom system for $\PITL$ with just finite
  time~\cite{Moszkowski04a} in
  Table~\ref{pitl-axiom-system-for-finite-time-table}.  The axiom system and
  proofs of theorems are easily \emph{relativised} to make finite time
  explicit and deduced with the new axiom system for both finite and infinite
  time already presented in
  Table~\ref{axiom-system-for-pitl-with-finite-and-infinite-time-table}.  The
  relativisation can use the fact that the two axiom systems are quite
  similar.
\end{proof}
\emph{One can alternatively disregard
  Theorem~\ref{completeness-of-pitl-axiom-system-for-finite-time-thm} and
  instead treat our presentation as a \emph{self-contained} proof reducing
  completeness for $\PITL$ with both finite and infinite time to that for
  $\PITL$ with just finite time.}

\begin{mycomment}
\subsection{Why Study Axiomatic Completeness?}

At this stage, it seems worthwhile to enumerate reasons for considering
axiomatic completeness of $\PITL$:
\begin{iteMize}{$\bullet$}
\item It is a basic desirable property of any logic (e.g., see logic textbooks).
\item It provides formal way to explore the logic's analytical power.
\item It is a nontrivial, rigorous and systematic application of the logic.
\item Our work suggests \emph{unexpected benefits} of using $\PITL$ to obtain
  results for other logics.
\item This approach readily generalises to finite domains.
\item Abstracted versions of first-order properties can be investigated.
\item Theorem provers can embed complete axiom systems.
\end{iteMize}
\end{mycomment}

\subsection{Summary of the Completeness Proof}

Our proof of axiomatic completeness for $\PITL$ establishes that any
consistent $\PITL$ formula is satisfiable (see the earlier
Lemma~\ref{alternative-completeness-lem}).  The completeness proof makes use
of a $\PITL$ subset called $\PTLU$ (defined later in
\S\ref{ptl-with-until-subsec}) which is a version of $\PTL$ having an $\Until$
operator.  As we discuss in \S\ref{ptl-with-until-subsec}, axiomatic
completeness for $\PTLU$ readily follows from axiomatic completeness for basic
$\PTL$ so any consistent $\PTLU$ formula is satisfiable.

The $\PITL$ completeness proof can be roughly summarised as ensuring that for
any consistent $\PITL$ formula $A$, there exists a consistent $\PTLU$ formula
$Y_0$, which possibly contains auxiliary propositional variables, such that
the $\PITL$ implication $Y_0\imp A$ is deducible.  Completeness for $\PTLU$
guarantees that $Y_0$ is satisfiable.  The soundness of the $\PITL$ axiom
system then ensures that any model of $Y_0$ also satisfies $A$ thereby showing
axiomatic completeness for $\PITL$. Note that in the actual proof, we use make
use of a $\PTLU$ conjunction $Y\And X$ in place of $Y_0$.

In the course of the $\PITL$ completeness proof, we also employ another
$\PITL$ subset called $\PITLK$ (defined later in
\S\ref{pitl-without-omega-iteration-subsec}). It is a version of $\PITL$
without omega-iteration and serves as a kind of bridge between full $\PITL$
and $\PTLU$.  The $\PITL$ completeness proof first obtains from the $\PITL$
formula $A$ a $\PITLK$ formula $K$ such that we can deduce $A\equiv K$.  We
then show how to obtain the $\PTLU$ formula $Y_0$ such that the implication
$Y_0\imp K$ is deducible.  We further show that if $A$ is consistent, so are
$K$ and $Y_0$.  Axiomatic completeness for $\PTLU$ ensures that the consistent
$\PTLU$ formula $Y_0$ is satisfiable.  The implication $Y_0\imp K$ together
with the deduced equivalence $A\equiv K$ guarantees the deducibility of the
previously mentioned $\PITL$ implication $Y_0\imp A$. Hence, any model of
$Y_0$ also satisfies $A$, thereby establishing completeness for $\PITL$ since
every consistent $\PITL$ formula is indeed satisfiable.

Here is a very brief summary of the main reductions:
\begin{displaymath}
  \PITL
  \quad
  \xrightarrow{\text{Section~\ref{reduction-of-chop-omega-sec}}}\quad
  \PITLK
  \quad
  \xrightarrow{\text{Section~\ref{reduction-of-pitl-to-ptlu-sec}}}\quad 
  \PTLU
\dotspace.
\end{displaymath}
Only the reduction from $\PITLK$ to $\PTLU$ requires some explicit
automata-theoretic constructions which involve finite words and are expressed
in temporal logic.

Below is the structure of our reduction from $\PITL$ to $\PTLU$:
\begin{iteMize}{$\bullet$}
\item In Section~\ref{right-instances-right-variables-and-right-theorems-sec}
  we describe a class of $\PITL$ theorems with useful substitution instances.
\item In Section~\ref{some-lemmas-for-replacement-sec} we present lemmas for
  systematically replacing some of a formula's subformulas by others in
  proofs.
\item In Section~\ref{useful-subsets-of-pitl-sec} we formally introduce the
  very simple $\PTL$ subset $\NLone$ as well as the subsets $\PTLU$ and
  $\PITLK$.  Although $\PITLK$ lacks chop-omega, it still has the same
  expressiveness as $\PITL$.  We also describe three other classes of formulas
  called \emph{right-chops}, \emph{chain-formulas} and \emph{auxiliary
    temporal assignments}.
\item In Section~\ref{reduction-of-chop-omega-sec} we show that any $\PITL$
  formula is deducibly equivalent to one in $\PITLK$.
\item In Section~\ref{deterministic-semi-automata-and-automata-sec} we
  show how to represent semi-automata and automata in $\PITL$.
\item Section~\ref{compound-semi-automata-for-suffix-recognition-sec} utilises
  the material in the previous section to test for a given $\PITL$ formula in
  suffixes of a finite interval.
  Sections~\ref{deterministic-semi-automata-and-automata-sec}
  and~\ref{compound-semi-automata-for-suffix-recognition-sec} provide a basis
  for introducing suitable auxiliary variables via auxiliary temporal
  assignments.
\item In Section~\ref{reduction-of-pitl-to-ptlu-sec} we use the constructed
  auxiliary variables to reduce an arbitrary consistent $\PITLK$ formula $K$
  to one in $\PTLU$.  Axiomatic completeness for $\PITL$ with infinite time
  then readily follows from this.
%%
%%{ata-and-ptlu-formula-lem}
\end{iteMize}

% \begin{center}
% \begin{tabular}{l}
%   $\PITL$ \\
%   \enspace $\Downarrow$\quad
% 	\emph{\S\ref{reduction-of-chop-omega-sec}:
% 		\begin{tabular}[t]{l}
% 		  Replacement of chop-omega subformulas \\
% 		  by equivalent formulas without it
%                 \end{tabular}}
%  \\
%   $\PITL$ without chop-omega \\
% %%  \enspace $\Downarrow$\quad
% %%	\emph{\S\ref{re-expressing-of-pitl-formulas-subsec}:} \\
% %%  $\PITL$ with strong chop and without chop-omega \\
%   \enspace $\Downarrow$\quad
% 	\emph{\relax
% 	  \S\ref{reduction-of-pitl-to-ptlu-sec}:
% 	       Use auxiliary variables and automata} \\
%   $\PTL$
% \end{tabular}
% \end{center}

\noindent A large portion of the reasoning is done at the semantic level (for example,
all of Section~\ref{compound-semi-automata-for-suffix-recognition-sec}).  We
then employ axiomatic completeness for restricted versions of $\PITL$ (such as
$\PITL$ with finite time) to immediately deduce the theoremhood of key
properties expressible as valid formulas in these versions.  This
significantly shortens the completeness proof by reducing the amount of
explicit deductions.

\section{Right-Instances, Right-Variables and Right-\!Theorems}

\label{right-instances-right-variables-and-right-theorems-sec}

Before proceeding further, we need to introduce a class of $\PITL$ theorems
for which suitable substitution instances are themselves deducible as
theorems.  Now in the completeness proof for $\PITL$ later on, if a deducible
$\PITL$ formula has propositional variables not occurring in the left of chops
or in chop-stars (e.g., $p$ in the formula $p\imp \Diamond p$), then in each
step of the formula's deduction these particular variables likewise do not
occur in the left of chops or chop-stars.  We define more generally for any
$\PITL$ formula $A$ and subformula $B$ in $A$, a \emphbf{right-instance} of
$B$ in $A$ to be an instance of $B$ which does not occur within the left of a
chop or within some chop-star.  Consider for example the disjunction below:
\begin{equation}
  \label{right-instance-example-1-eq}
  (p \chop \Not q) \;\Or\; (p \chop p') \;\Or\; (p \chop p')\SChopstar
\dotspace.
\end{equation}
The subformulas $\Not q$, and $(p \chop \Not q)$ as well as the leftmost
occurrence of $p \chop p'$ are right-instances in the overall
formula~\eqref{right-instance-example-1-eq}.  However, all three occurrences
of $p$ and the rightmost occurrences of $p'$ and $p \chop p'$ are not
right-instances in~\eqref{right-instance-example-1-eq} because each is either
in the left of a chop or in a chop-star.

Now let a $\PITL$ formula $A$'s
\emphbf{right-variables} be the (finite) set $\mathit{RV}(A)$ of $A$'s
variables which have only right-instances in $A$, that is, do not occur in the
left of chops or chop-stars.

We now look at why the concept of right-variable is needed.  In the formula $p
\imp \Diamond p$, the variable $p$ is a right-variable.  Therefore, from the
validity of $p \imp \Diamond p$, we can infer the validity of the substitution
instance $\Skip \imp \Diamond\Skip$.
Lemma~\ref{substitution-instances-into-right-variables-lem}, which is shortly
presented, formalises this idea.  However, if a variable is not a
right-variable in a valid formula, we might incorrectly infer that a
substitution instance of the formula is also valid.  For instance, the
variable $p$ is not a right-variable in the formula $p \imp \Bf p$ which is an
instance of Axiom~\refid{StateImpBf} in
Table~\ref{axiom-system-for-pitl-with-finite-and-infinite-time-table}.  This
formula is valid but the substitution instance $\Skip \imp \Bf \Skip$ is not.

Now all propositional variables in a propositional formula with no temporal
operators are right-variables of that formula.  More generally, all
propositional variables in a $\PTL$ formula are right-variables.  In contrast, a
chop-star formula has no right-variables.

The next simple lemma concerns substitution into right-variables in valid
formulas:
\begin{mylemma}[Substitution Instances into Right-Variables]
  \label{substitution-instances-into-right-variables-lem}
  Suppose $A$ is a $\PITL$ formula, $p$ is one of $A$'s right-variables (i.e.,
  in $\mathit{RV}(A)$) and $B$ is some $\PITL$ formula.  Then if $A$ is valid,
  so is the substitution instance $A_p^B$.
\end{mylemma}
\begin{proof}[Proof by contradiction.]
  Let $q$ be a variable not occurring in $A$ or $B$ and let $C$ be a variant
  of $A$ with all instances of $p$ replaced by $q$ (i.e., $A_p^q$).  The
  variable $p$ is a right-variable of $A$ so $q$ is similarly a right-variable
  of $C$.  It follows by induction on $A$'s syntax that $A_p^B$ and $C_q^B$
  denote exactly the same $\PITL$ formula.  Consequently, in our reasoning
  about $A_p^B$, we can assume without loss of generality that $p$ itself does
  not occur in $B$.  This is because we can view $A_p^B$ as being $C_q^B$.

  Now suppose by contradiction that $A_p^B$ is not valid.  By our previous
  discussion, also assume that $p$ does not occur in $B$.  Then some
  interval $\sigma$ satisfies $\Not(A_p^B)$. We construct a variant $\sigma'$
  in which the value of variable $p$ in each state $\sigma'_i$ equals true iff
  the suffix subinterval $\sigma_{i\uparrow}$ satisfies $B$. Hence
  $\sigma'\vld\Box(p\equiv B)$ and $\sigma'\vld \Not(A_p^B)$.  It readily
  follows from this and $p$ being a right-variable that $\sigma'$ satisfies
  $\Not A$ since $A_p^B$ only examines $B$ in suffix subintervals.  From
  $\sigma'\vld \Not A$ we have that $A$ is not valid.
\end{proof}

Later in Section~\ref{reduction-of-chop-omega-sec}, our completeness proof
will need a deductive analogue of the semantically oriented
Lemma~\ref{substitution-instances-into-right-variables-lem} to permit us to
infer from a theorem $A$ and right-variable $p$ in $\mathit{RV}(A)$ another
theorem $A_p^B$.  One way to achieve this is by adding the next inference rule
to the $\PITL$ axiom system in
Table~\ref{axiom-system-for-pitl-with-finite-and-infinite-time-table} for any
formula $A$ and variable $p$ in $\mathit{RV(A)}$:
\begin{equation}
  \label{optional-inference-rule-for-right-variables-1-eq}
  \theoremBEN A \infer \theoremBEN A_p^B
\dotspace.
\end{equation}
Another possibility is an analogue of Inference Rule~\refid{BfAux} in
Table~\ref{axiom-system-for-pitl-with-finite-and-infinite-time-table}:
\begin{equation*}
  \theoremBEN \Box(p\equiv B) \,\implies\, A \infer \theoremBEN A
\dotspace,
\end{equation*}
where the propositional variable $p$ does not occur in $A$ or $B$.
However, it turns out that these are unnecessary since the axiom system in its
current form is already sufficient to allow a suitable class of such
substitutions.  We now present a formal basis for this.

A $\PITL$ formula $A$ which is theorem (i.e., $\thm A$) is called a
\emphbf{right-theorem} (denoted $\thmR A$) if there exists a deduction of $A$
in which $A$'s right-variables never occur on the left of chop or in chop-star
in any proof steps.  However, any of $A$'s variables not in $\mathit{RV}(A)$
as well as any subsequently introduced auxiliary variables in the deductions
are permitted to appear in some deduction steps in the left of chops or
chop-stars.  For example, if $p$ is a right-variable of $A$, then no proof
step can use $p$ with Axiom~\refid{StateImpBf} (e.g., $\thm p\imp \Bf p$)
since $p$ is not a right-variable here owing to $\Bf p$.

The completeness proof for $\PITL$ will ensure that any valid $\PITL$ formula
$A$ is indeed deducible as a right-theorem.  We will refer to this here as
\textbf{right-completeness}.  Below is our main theorem for axiomatic
completeness of $\PITL$ using right-completeness:
\begin{mytheorem}[Right-Completeness of $\PITL$ Axiom System]
  \label{completeness-of-pitl-axiom-system-thm}
  Any valid $\PITL$ formula $A$ is a right-theorem of the axiom system, that
  is, if $\vld A$, then $\thmR A$.
\end{mytheorem}
The proof of this, our main result, is described later and concludes in
Section~\ref{reduction-of-pitl-to-ptlu-sec}.

Right-theoremhood naturally yields the dual notion of
\emphbf{right-consistency} of a $\PITL$ formula $A$, that is, not $\thmR \Not
A$.  Our completeness proof for $\PITL$ can therefore be regarded as not only
showing that valid $\PITL$ formulas are right-theorems but also that any
right-consistent $\PITL$ formula is satisfiable (compare with
Lemma~\ref{alternative-completeness-lem}).

As already pointed out, the main reason we are interested in right-theorems is
that suitable substitution instances of them are $\PITL$ theorems.  Our need
for this occurs when in Section~\ref{reduction-of-chop-omega-sec} we reduce
right-completeness for $\PITL$ to right-completeness for its subset $\PITLK$
without chop-omega.  The lemma below formalises the substitution process:
\begin{mylemma}[Substitution Instances of Right-Theorems]
  \label{substitution-instances-of-right-theorems-lem}
  Let $A$ and $B_1,\ldots, B_n$ be $PITL$ formulas and $p_1,\ldots, p_n$ be
  some of $A$'s right-variables.  If $A$ is a right-theorem, then so is the
  substitution instance $A_{p_1,\ldots,p_n}^{B_1,\ldots,B_n}$, that is, $\thmR
  A_{p_1,\ldots,p_n}^{B_1,\ldots,B_n}$.
%   Let $A$ be a right-theorem in $\PITL$, $p_1$, \ldots, $p_n$ be some of $A$'s
%   right-variables (i.e., in $\mathit{RV}(A)$) and $B_1$, \ldots, $B_n$ be some
%   $\PITL$ formulas.  Then the substitution instance
%   $A_{p_1,\ldots,p_n}^{B_1,\ldots,B_n}$ is a $\PITL$ theorem, that is, $\thm
%   A_{p_1,\ldots,p_n}^{B_1,\ldots,B_n}$.
\end{mylemma}
\begin{proof}
  We assume that auxiliary variables in $A$'s proof (i.e., ones not in $V_A$)
  do not occur in $B_1,\ldots, B_n$.  In each step of $A$'s proof, we replace
  each $p_i$ by $B_i$ to obtain $\thmR A_{p_1,\ldots,p_n}^{B_1,\ldots,B_n}$.
\end{proof}

%%We already noted that all variables in $\PTL$ formulas are right-variables.

Many $\PITL$ theorems in
Appendix~\ref{some-pitl-theorems-and-their-proofs-sec} can be checked to be
right-theorems by inspection of the proof steps.  For example, those with no
right-variables are immediate right-theorems.  We have not indicated in the
appendix which theorems are right-theorems and will normally only designate
formulas as right-theorems in the completeness proof when this is needed.

The next lemma concerns the relationship between derived rules and
right-theorems:
\begin{mylemma}[Right-Theorems from Some Derived Rules]
  \label{right-theorems-from-some-derived-rules-lem}
  Suppose the assumptions of a derived rule which deduces some $\PITL$ formula
  $A$ are right-theorems.  Furthermore, suppose that in the derived rule's own
  proof of $A$, none of $A$'s right-variables occur on the left of chop or in
  chop-star (including in any nested deduced $\PITL$ theorems and derived
  rules).  If $A$'s right-variables are a subset of the union of the
  assumptions' right-variables, then $A$ itself is a right-theorem.
\end{mylemma}
We omit the proof.
% \begin{proof}
%   This follows from that fact that the only reasoning involved in such
%   deductions concerns the righthand sides of chop.
% \end{proof}
For example, Derived Rule~\refid{BoxImpInferBoxImpBox} in
Appendix~\ref{some-pitl-theorems-and-their-proofs-sec} (see also the abbreviated
Table~\ref{list-of-pitl-theorems-and-derived-rules-mentioned-before-app-table}
found later in
\S\ref{formal-equivalence-of-the-two-representations-of-runs-subsec}) lets us
infer from the theorem $\thm \Box\!A \imp B$ the theorem $\Box\!A \imp \Box\!B$.
It only requires the kind of reasoning mentioned in
Lemma~\ref{right-theorems-from-some-derived-rules-lem}.  Consequently, from
$\thmR \Box\!A \imp B$ we can infer $\thmR \Box\!A \imp \Box\!B$.

\vspace{5pt}

\emph{Readers are strongly encouraged to initially try to understand our
  completeness proof without consideration of right-theoremhood by simply
  viewing it as ordinary theoremhood and ignoring the prefix ``right-''.  This
  can even be rigorously done by assuming that the optional inference
  rule~\eqref{optional-inference-rule-for-right-variables-1-eq} is part of the
  $\PITL$ axiom system. A subsequent, more thorough study of the material can
  then better take right-theoremhood into account.  Indeed, we can then regard
  our completeness proof as two parallel proofs, a simpler one
  with~\eqref{optional-inference-rule-for-right-variables-1-eq} and another
  more sophisticated one which is based on right-theoremhood and
  Lemma~\ref{substitution-instances-of-right-theorems-lem} and hence does not
  assume~\eqref{optional-inference-rule-for-right-variables-1-eq}.
  Incidentally, our completeness proof ultimately ensures
  that~\eqref{optional-inference-rule-for-right-variables-1-eq} is obtainable
  as a derived inference rule even if it is not in the axiom system.}

\section{Some Lemmas for Replacement}

\label{some-lemmas-for-replacement-sec}

We now consider some techniques used in the completeness proof to replace
selected right-instances in a $\PITL$ formula by other formulas.
\begin{mylemma}
  \label{basic-limited-replacement-lemma-for-pitl-lem}
  Let $A_1$, $A_2$, $B_1$ and $B_2$ be $\PITL$ formulas.  If $A_2$ can be
  obtained from $A_1$ by replacing zero or more right-instances of $B_1$ in
  $A_1$ by $B_2$, then the next implication is deducible as a right-theorem:
  \begin{equation*}
%    \label{basic-limited-replacement-lemma-for-pitl-1-eq}
    \TheoremR \Box(B_1\equiv B_2) \Implies A_1\equiv A_2
\dotspace.
  \end{equation*}
\end{mylemma}
\begin{proof}
%   The case for finite time follows almost immediately from axiomatic
%   completeness for finite time.  This is because if we have $\thm B_1\equiv
%   B_2$, then the equivalence $A_1\equiv A_2$ is valid and hence the valid
%   implication $\Finite \implies A_1\equiv A_2$ is a theorem.

  The proof involves induction on the syntax of formula $A_1$, with each
  instance of $B_1$ regarded as atomic.  We consider the cases when $A_1$ is
  $B_1$ itself, $\True$, a propositional variable $p$, $\Not C$, $C_1 \Or
  C_2$, $\Skip$, $C_1\chop C_2$, and $C\SChopstar$.  The first three of these
  involve quite routine conventional propositional reasoning.  The case for
  $\Skip$ is trivial since $A_1$ and $A_2$ are identical.  The case for
  chop-star is likewise trivial since this lemma does not permit replacement
  in its scope.

  For the case for chop, assume $A_1$ and $A_2$ have the forms $C_1\chop C_2$
  and $C_1\chop C'_2$, respectively.  Note that no replacements are done in
  the left of chop.  By induction on $A_1$'s syntax, we deduce the next
  implication:
  \begin{equation*}
    \theoremR \Box(B_1\equiv B_2) \Implies C_2\equiv C'_2
\dotspace.
  \end{equation*}
  This and $\PTL$ reasoning (see Derived Rule~\refid{BoxImpInferBoxImpBox} in
  Appendix~\ref{some-pitl-theorems-and-their-proofs-sec} and also in the
  abbreviated
  Table~\ref{list-of-pitl-theorems-and-derived-rules-mentioned-before-app-table}
  found later in
  \S\ref{formal-equivalence-of-the-two-representations-of-runs-subsec}) yields
  the implication below:
  \begin{equation*}
    \theoremR \Box(B_1\equiv B_2) \Implies \Box(C_2\equiv C'_2)
\dotspace.
  \end{equation*}
  Lemma~\ref{right-theorems-from-some-derived-rules-lem} ensures that our use
  here of Derived Rule~\refid{BoxImpInferBoxImpBox} indeed yields a
  right-theorem.

  We can also deduce the next implication using
  Axiom~\refid{BfAndBoxImpChopImpChop} and some further temporal reasoning
  (see $\PITL$ Theorem~\refid{BoxChopEqvChop} in
  Appendix~\ref{some-pitl-theorems-and-their-proofs-sec} and also in
  Table~\ref{list-of-pitl-theorems-and-derived-rules-mentioned-before-app-table}
  in \S\ref{formal-equivalence-of-the-two-representations-of-runs-subsec}):
  \begin{equation*}
    \theoremR \Box(C_2\equiv C'_2)
      \Implies (C_1\chop C_2) \equiv (C_1\chop C'_2)
\dotspace.
  \end{equation*}
  These two implications together yield our goal below:
  \begin{equation*}
    \theoremR \Box(B_1\equiv B_2)
      \Implies (C_1\chop C_2) \equiv (C_1\chop C'_2)
\dotspace.
  \end{equation*}
  This concludes Lemma~\ref{basic-limited-replacement-lemma-for-pitl-lem}'s proof.
\end{proof}

Lemma~\ref{basic-limited-replacement-lemma-for-pitl-lem} yields a derived
inference rule for \emphbf{Right Replacement} of formulas:
\begin{mylemma}[Right Replacement Rule]
  \label{right-replacement-derived-rule-lemma-for-pitl-lem}
  Let $A_1$, $A_2$, $B_1$ and $B_2$ be $\PITL$ formulas.  Suppose that $A_2$
  can be obtained from $A_1$ by replacing zero or more right-instances of
  $B_1$ in $A_1$ by $B_2$.  If $B_1$ and $B_2$ are deducibly equivalent as a
  right-theorem (i.e., $\thmR B_1\equiv B_2$), then so are $A_1$ and $A_2$.
\end{mylemma}
\begin{proof}
  By Lemma~\ref{basic-limited-replacement-lemma-for-pitl-lem}, we deduce the
  next implication:
  \begin{equation*}
    \TheoremR \Box(B_1\equiv B_2) \Implies A_1\equiv A_2
\dotspace.
  \end{equation*}
  Also, $\thmR B_1\equiv B_2$ and Inference Rule~\refid{BoxGen} yield $\thmR
  \Box(B_1\equiv B_2)$.  Then modus ponens yields $\thmR A_1\equiv A_2$.
\end{proof}

\section{Useful Subsets of PITL}

\label{useful-subsets-of-pitl-sec}

We now describe five subsets of $\PITL$ and some associated properties which
will be extensively used later on in different parts of the $\PITL$
completeness proof.  We have chosen to collect material about the subsets here
instead of introducing each subset as the need arises.  This should make it
easier for readers to review the definitions and features when required and
also make the main steps of the completeness proof shorter and more focused.
In addition, when taken as a whole, the combined presentation of the $\PITL$
subsets enables us to give a technical overview of some of the proof steps
encountered.  Table~\ref{kinds-of-variables-table} later lists variables used
for the subsets and other subsequently defined categories.

\subsection{PTL with only Unnested Next Constructs}

\label{nl-one-formulas-subsec}

Let \emphbf{$\NLone$} denote the subset of $\PTL$ formulas in which the only
temporal operators are unnested $\Next$s (e.g., $p \Or \Next\Not p$ but not $p
\Or \Next\Next\Not p$). It is not hard to see that $\NLone$ formulas only
examine an interval's first two states.  They are therefore useful for
describing automata transitions from one state to the next.  The variables $T$
and $T'$ denote formulas in $\NLone$.

Below are some theorems which contain $\NLone$ formulas and are required in
the completeness proof.  None of these theorems are themselves in $\NLone$.
The proofs are in Appendix~\ref{some-pitl-theorems-and-their-proofs-sec}.
\begin{center}
\begin{tabular}{l@{\qquad}L}
  \refid{DfMoreAndNLoneEqvMoreAndNLone} &
    \Theorem \Df(\More \And T) \Equiv \More \And T \\
  \refid{DfSkipAndNLoneEqvMoreAndNLone} &
    \Theorem \Df(\Skip \And T) \Equiv \More \And T \\
  \refid{NLoneAndSkipChopEqvNLoneAndNext} &
    \Theorem (\Skip \And T)\chop A \Equiv T \Andd \Next A
\end{tabular}
\end{center}

\subsection{PTL with Until}

\label{ptl-with-until-subsec}

Recall that for our purposes we define $\PTL$ to be the subset of $\PITL$ with
just $\Skip$ and the derived temporal operators $\Next$ and $\Diamond$ shown
in Table~\ref{pitl-derived-operators-table}.

We also use a more expressive version of $\PTL$ denoted here as
\emphbf{$\PTLU$} with a \emph{strong} version of the standard temporal
operator $\Until$, derivable in $\PITL$:
\begin{displaymath}
  T \Until A \Defeqv (\Skip \And T)\SChopstar\chop A
\dotspace.
\end{displaymath}
We limit $\Until\!$'s lefthand operand to be a formula in $\NLone$ (defined
previously in \S\ref{nl-one-formulas-subsec}).  Note that this definition of
$\Until$ using chop and chop-star results in any variable in the left operand
of $\Until$ not being a right-variable.  Let $Y$ and $Y'$ denote $\PTLU$
formulas.
% and also use a derived \emph{``strict Until''} operator:
% \begin{displaymath}
%   w \UntilPlus A \Defeqv w \And \Next(w \Until A)
% \dotspace.
% \end{displaymath}

We establish right-completeness for $\PITL$ by a reduction to $\PTLU$, instead
of directly to $\PTL$.  It is not hard to show that our axiom system is
complete for $\PTLU$ formulas.  This is because we can deduce the next two
$\PTLU$ axioms known to capture this kind of $\Until\!$'s behaviour (the
$\PITL$ proofs are in Appendix~\ref{some-pitl-theorems-and-their-proofs-sec}):
\begin{equation*}
  \refid{UntilEqv}\enspace
  \Theorem T\Until A \Equiv A \,\Or\, \bigr(T \And \Next(T\Until A)\bigr) \\
  \qquad\quad
  \refid{UntilImpDiamond}\enspace
  \Theorem T\Until A \Implies \Diamond A
\dotspace.
\end{equation*}
% \begin{align*}
%   & \Theorem T\Until A \Equiv A \,\Or\, \bigr(T \And \Next(T\Until A)\bigr) \\
%   & \Theorem T\Until A \Implies \Diamond A
% \dotspace.
% \end{align*}
Consequently, we can reduce completeness for $\PTLU$ to it for $\PTL$.  In
fact every $\PTLU$ theorem is a right-theorem.  This is because the
right-variables in $T\Until A$ remain so in ~\refid{UntilEqv}
and~\refid{UntilImpDiamond}, Hence, the two $\PTLU$ axioms ensure that these
variables remain right-variables in the proof steps for deducing a $\PTLU$
theorem in the $\PITL$ axiom system.  See Kr{\"o}ger and
Merz~\citeyear{KroegerMerz08} for more about axioms for a variety of such
binary temporal operators.

\subsection{PITL without Omega-Iteration}

\label{pitl-without-omega-iteration-subsec}

Our completeness proof includes a step in which any chop-omega (defined in
Table~\ref{pitl-derived-operators-table}) is eliminated by re-expressing any
chop-star not in the left of chop or another chop-star.  This exploits a
convenient alternative characterisation of omega-regular languages described
by Thomas at the end of~\cite{Thomas79} which does not involve
omega-iteration.  It instead employs \emph{closure} under some other
operations which include \emph{complementation}:
\begin{mytheorem}[Omega-Regularity using Closures]
  \label{thomas'-theorem-thm}
  The omega-regular languages of an alphabet $\Sigma$ are exactly the closure
  of $\{\emptyset\}$ under
  the following: (1) union, (2) complementation (with respect to
  $\Sigma^\omega$) and (3) left concatenation by $\Sigma$'s regular languages.
\end{mytheorem}
Here $\emptyset$ denotes the omega-language with no elements.

Let $\mathbf{PITL^K}$ denote the $\PITL$ subset in which chop-star only occurs
on the left of chops (like (3) in Thomas' theorem above) and is therefore
restricted to finite intervals.  The $\textrm{K}$ in $\PITLK$ stands for
``Kleene star''.  For example, the next two formulas are in $\PITLK$:
\begin{displaymath}
  (\Skip \And p)\SChopstar\chop q
    \qquad (\Skip\SChopstar\chop \Skip) \Or \Next p
\dotspace.
\end{displaymath}
In contrast, the two formulas below are not in $\PITLK$:
\begin{displaymath}
  (\Skip \And p)\SChopstar
    \qquad p \imp \Diamond (\Skip \And q)\SChopstar
\dotspace.
\end{displaymath}
Observe that a $\PITLK$ formula can contain chop-star subformulas, which by
the definition of $\PITLK$ are not themselves in it.  An example is $(\Skip
\And p)\SChopstar$ in $(\Skip \And p)\SChopstar\chop q$.

With just finite time, any $\PITL$ formula $A$ is easily re-expressed in
$\PITLK$ as $A\chop\Empty$ (compare with Axiom~\refid{FiniteImpChopEmpty} in
Table~\ref{axiom-system-for-pitl-with-finite-and-infinite-time-table}).
However this technique does not work for infinite time.  We also need Thomas'
theorem (Theorem~\ref{thomas'-theorem-thm}) to ensure that any $\PITL$ formula
$A$ has a semantically equivalent $\PITLK$ formula $K$ for both finite and
infinite time (i.e., $\vld A\equiv K$).  For example, one way to re-express
the $\PITL$ formula $(\Skip \And p)\SChopstar$ in $\PITLK$ is $\Box(\More \imp
p)$.  It follows that any chop-omega formula is re-expressible in $\PITLK$.
For instance, for any $\PITL$ formula $B$, the formula $(\Skip \And
B)\Chopomega$ is semantically equivalent to $\Box\Df(\Skip \And B)$.

Later on in Section~\ref{reduction-of-chop-omega-sec} we employ Thomas'
theorem to easily reduce axiomatic completeness for $\PITL$ to that for
$\PITLK$.  More precisely, we will formally establish there that for any
$\PITL$ formula $A$, there exists a semantically equivalent $\PITLK$ formula
$K$ such that the formula $A\equiv K$ is deducible as a $\PITL$ theorem.
Hence, by simple propositional reasoning, if $A$ is consistent, so is $K$ and
any model for $K$ is also one for $A$.  The remainder of the overall
completeness proof then reduces completeness for $\PITLK$ to it for $\PTLU$.

Choueka and Peleg~\citeyear{ChouekaPeleg83} give a simpler proof of Thomas'
theorem using standard deterministic omega automata. Readers favouring an
automata-theoretic perspective can therefore regard the theorem in the context
of $\PITL$ as a basis for \emph{implicitly} determinising the original $\PITL$
formula, resulting in a semantically equivalent one in $\PITLK$.

\subsection{Right-Chops and Chain Formulas}

\label{right-chops-and-chain-formulas-subsec}

For any $\PITL$ formula $A$, we call a chop formula in $A$ a
\emphbf{right-chop} if it is not in another chop's left operand or in a
chop-star.  Right-chops help reduce $\PITLK$ to $\PTLU$. We illustrate them
with the formula below:
\begin{equation}
  \label{right-chop-example-1-eq}
  \bigl((p\chop p') \chop \Not(q\chop q')\bigr) \;\Or\; (p \chop p')
\dotspace.
\end{equation}
The following three formulas all occur as right-chops in this:
\begin{displaymath}
  (p\chop p') \chop \Not(q\chop q')
  \qquad
  q\chop q'
  \qquad
  p \chop p'
\dotspace.
\end{displaymath}
Only the second instance of $p\chop p'$ in
formula~\eqref{right-chop-example-1-eq} is a right-chop.  In contrast, the
first instance of $p\chop p'$ is not a right-chop since it is within the left
operand of another chop.  Observe that the right-chops of a $\PITL$ formula
$A$ are exactly those subformulas in $A$, including possibly $A$ itself, which
have chop as their main operator and are right-instances (previously defined
in Section~\ref{right-instances-right-variables-and-right-theorems-sec}).

\begin{mycomment}
  For any $n\ge 0$, let $\mathbf{PITL^K_n}$ denote the set of all $\PITLK$
  formulas with \emph{exactly} $n$ right-chops.  Observe the $\PITLK$ itself
  is the union $\bigcup_{i\ge 0}\PITLK_i$.  Also, the only temporal operator
  in formulas in $\PITLK_0$ is $\Skip$.  We will then inductively obtain the
  alternative version of axiomatic completeness
  (Lemma~\ref{alternative-completeness-lem}) for all of $\PITLK$.  In
  particular, we ultimately show by induction on $j\ge 0$ that for any
  $\PITLK_j$ formula $K$ and $\PTLU$ formula $Y$, if the conjunction $K\And Y$
  is consistent, then it is satisfiable.  Observe that for the base case with
  $j=0$, the formula $K\And Y$ is in $\PTLU$ and hence we can use the
  axiomatic completeness for $\PTLU$ previous discussed in
  \S\ref{ptl-with-until-subsec}.
\end{mycomment}

In addition to right-chops, the reduction of a $\PITLK$ formula to $\PTLU$
employs a class of $\PTLU$ formulas involving disjunctions and sequential
\emphbf{chains} of restricted constructs.
%\begin{mydefin}[Chain Formulas]
%  \label{chain-formula-def}
Let a \emphbf{chain formula} be any $\PTLU$ formula with the syntax below,
where $w$ is a state formula, $T$ is an $\NLone$ formula and $G$ and $G'$ are
themselves chain formulas:
\begin{displaymath}
    \Empty  \qquad w \And G \qquad  G \Or G' \qquad   T \Until G
%%  \Empty \And w \qquad  (\Skip \And T)\SChopstar
%%    \qquad G \Or G' \qquad G\chop G'
\dotspace.
\end{displaymath}
%\end{mydefin}  %% of chain formulas

The operator $\UntilOp$ in chain formulas involves a quite limited version of
the $\PITL$ operator chop-star which is much easier to reason about than full
chop-star.  The next lemma exploits this and shows that a chop in which the
left operand is a chain formula and the right one is in $\PTLU$ can be
re-expressed as a deducibly equivalent $\PTLU$ formula.
\begin{mylemma}
  \label{re-expressing-g-chop-x-in-ptlu-lem}
  For any chain formula $G$ and $\PTLU$ formula $Y$, there exists some $\PTLU$
  formula $Y'$ such that the equivalence $(G\chop Y)\equiv Y'$ is deducible as
  a right-theorem.
\end{mylemma}
\begin{proof}
  We do induction on $G$'s syntax using the deducible equivalences below in
  which $w$ is a state formula, $T$ is an $\NLone$ formula and $G'$ and $G''$
  are themselves chain formulas:
  \begin{displaymath}
    \begin{array}{@{\qquad}l@{\qquad}l@{\qquad}l@{\qquad}l@{}}
      \TheoremR \Empty\chop Y \EQUIV Y
        & \TheoremR (G' \Or G'')\chop Y\EQUIV (G'\chop Y) \Or (G''\chop Y) \\
      \TheoremR (w \And G')\chop Y \EQUIV w \And (G'\chop Y)
        & \TheoremR (T \Until G')\chop Y \EQUIV T\Until(G'\chop Y)
\dotspace.
    \end{array}
  \end{displaymath}
% \begin{displaymath}
%     \begin{array}[b]{l@{\qquad}l}
%       \TheoremR \Empty\chop Y \EQUIV Y \\
%       \TheoremR (w \And G)\chop Y \EQUIV w \And (G\chop Y) \\
%       \TheoremR (G \Or H)\chop Y\EQUIV (G\chop Y) \Or (H\chop Y) \\
%       \TheoremR (T \Until G)\chop Y \EQUIV T\Until(G\chop Y)
% %       \TheoremR (\Empty \And w) \chop Y \EQUIV w \And Y \\
% %       \TheoremR (\Skip \And T)\SChopstar\chop Y \EQUIV T \Until Y \\
% %       \TheoremR (G' \Or G'')\chop Y\EQUIV (G'\chop Y) \Or (G''\chop Y) \\
% %       \TheoremR (G'\chop G'')\chop Y\EQUIV G'\chop (G''\chop Y)
% \dotspace.
%     \end{array}
%   \end{displaymath}
  The first of these is an instance of $\PITL$ Axiom~\refid{EmptyChop}.  The
  second and third are respective instances of $\PITL$
  Theorems~\refid{StateAndEmptyChop} and~\refid{OrChopEqv} in
  Appendix~\ref{some-pitl-theorems-and-their-proofs-sec} (see also the
  abbreviated
  Table~\ref{list-of-pitl-theorems-and-derived-rules-mentioned-before-app-table}
  found later in
  \S\ref{formal-equivalence-of-the-two-representations-of-runs-subsec}).  The
  fourth uses the earlier $\ITL$-based definition of the temporal operator
  $\UntilOp$ in~\S\ref{ptl-with-until-subsec} and Axiom~\refid{ChopAssoc}
  which itself concerns chop's associativity.
\end{proof}
For example, the left chop operand in the $\PITL$ formula $\bigl(p \And
(q\Until \Empty))\chop \Skip$ is a chain formula.  The chop itself is
deducibly equivalent to the $\PTLU$ formula $p \And (q\Until \Skip)$.

Our completeness proof will ultimately apply
Lemma~\ref{re-expressing-g-chop-x-in-ptlu-lem} when in
Section~\ref{reduction-of-pitl-to-ptlu-sec} we later replace the left operands
of a consistent $\PITLK$ formula's right-chops with chain formulas.  For this
to work, we will also need auxiliary variables of the kind now described.

\subsection{Auxiliary Temporal Assignments}

\label{auxiliary-temporal-assignments-subsec}

When we later represent automata runs in $\PITL$, it is convenient to
generalise formulas of the form $p\Tassign B$ (the temporal assignment
construct defined in Table~\ref{pitl-derived-operators-table}) to conjunctions
of several of these.  Please refer back to Section~\ref{pitl-sec} for a brief
explanation about the meaning of temporal assignment.
%\begin{mydefin}[Auxiliary Temporal Assignments]
%  \label{ata-def}
We call such a conjunction an \emphbf{Auxiliary Temporal Assignment (ATA)}.
It has the form given below:
\begin{displaymath}
  \textstyle
  \bigwedge_{1\le i\le n} (q_i \Tassign A_i)
\dotspace,
\end{displaymath}
for some $n\ge 0$, where each $A_i$ is a $\PITL$ formula, there are $n$
distinct \emphbf{auxiliary} propositional variables $q_1$, \ldots $q_n$ and
the only ones of them permitted in each $A_i$ are $q_1$, \ldots $q_{i-1}$.
All other propositional variables are allowed in any $A_i$.  Here is a sample
ATA with one nonauxiliary variable $r$ and two auxiliary variables $p$ and
$q$:
\begin{equation*}
  (p\Tassign \Next r) \And (q\Tassign \Box (r\imp \Diamond p))
\dotspace.
\end{equation*}
Variables such as $D$ and $D'$ denote ATAs.  Two ATAs are \emphbf{disjoint} if
they have distinct auxiliary variables.

\label{some-derived-inference-rules-for-auxiliary-variables-sec}

Let us now look at how to formally introduce ATAs containing auxiliary
variable into deductions for later use within the completeness proof in
\S\ref{proof-of-the-main-completeness-theorem-subsec}.

\begin{mylemma}[Temporal Operators $\Bf$, $\Tassign$ and Right-Consistency]
  \label{temporal-operator-bf-and-consistency-lem}
  Let $A$ and $B$ be $\PITL$ formulas with no instances of propositional
  variable $p$.  If $A$ is right-consistent, so is the conjunction $A \Andd
  \Bf(p\Tassign B)$.
\end{mylemma}
\begin{proof}[Proof by contradiction.]
  Suppose $A \Andd \Bf(p\Tassign B)$ is not right-consistent. Then
  $\Bf(p\Tassign B) \imp \Not A$ is a right-theorem.  We re-express
  $\Bf(p\Tassign B)$ as $\Bf\bigl((\Fin p)\equiv B\bigr)$.  By this and
  Inference Rule~\refid{BfAux}, the formula $\Not A$ is a right-theorem.
  Therefore $A$ is not right-consistent.
\end{proof}

Lemma~\ref{temporal-operator-bf-and-consistency-lem} readily generalises to
reduce a formula's right-consistency to that for a conjunction of it and a
suitable ATA:
\begin{mylemma}[The Temporal Operator $\Bf$, ATAs and Right-Consistency]
  \label{temporal-operator-bf-atas-and-consistency-lem}
  Let $A$ be a $\PITL$ formula and $D$ an ATA with no auxiliary variables in
  $A$.  If $A$ is right-consistent, so is the formula $A \Andd \Bf D$.
\end{mylemma}
\begin{proof}
  For some $n\ge 0$, the ATA $D$ contains $n$ auxiliary variables and has the
  form $\bigwedge_{1\le i\le n} (q_i \Tassign B_i)$.  We first apply
  Lemma~\ref{temporal-operator-bf-and-consistency-lem} $n$ times to reduce the
  formula $A$'s right-consistency to that for the next formula:
  \begin{equation}
    \label{temporal-operator-bf-atas-and-consistency-1-eq}
    \textstyle
    A \Andd \bigwedge_{1\le i\le n} \Bf(q_i \Tassign B_i)
\dotspace.
  \end{equation}
  The conjunction of $\Bf$-formulas is then re-expressed with a single $\Bf$
  (see $\PITL$ Theorem~\refid{BfAndEqv} found in
  Appendix~\ref{some-pitl-theorems-and-their-proofs-sec} and also included in
  the more abbreviated
  Table~\ref{list-of-pitl-theorems-and-derived-rules-mentioned-before-app-table}
  later in
  \S\ref{formal-equivalence-of-the-two-representations-of-runs-subsec}) to
  obtain the formula below which is deducibly equivalent
  to~\eqref{temporal-operator-bf-atas-and-consistency-1-eq}:
  \begin{equation*}
    \textstyle
    A \Andd \Bf\bigl(\bigwedge_{1\le i\le n} q_i \Tassign B_i\bigr)
\dotspace.
  \end{equation*}
  This is the same as our goal $A \Andd \Bf D$.
\end{proof}

\subsection{Overview of Role of PITL Subsets in Rest of Completeness Proof}

\label{overview-of-role-of-pitl-subsets-in-completeness-proof-subsec}

The $\PITL$ completeness proof can now be summarised using the $\PITL$ subsets
just presented. \emph{Some readers may prefer to skip this material and
  proceed directly to the proof which starts in
  Section~\ref{reduction-of-chop-omega-sec}.}  Our goal here is to show that
any right-consistent $\PITL$ formula $A$ is satisfiable.
Here is an informal sequence of the transformations involved:
\begin{displaymath}
  A
  \quad
  \xrightarrow{\text{Section~\ref{reduction-of-chop-omega-sec}}}
  \quad
  K
  \quad
  \xrightarrow{\text{Section~\ref{reduction-of-pitl-to-ptlu-sec}}}
  \quad 
  K'\And \Bf D'
  \quad
  \xrightarrow{\text{Section~\ref{reduction-of-pitl-to-ptlu-sec}}}
  \quad 
  Y\And X
\dotspace,
\end{displaymath}
where $K$ is a $\PITLK$ formula, $K'$ is a $\PITLK$ formula in which the left
operands of all right chops are chain formulas, $D'$ is an ATA and $Y$ and $X$
are respectively in $\PTLU$ and $\PTL$.  If $A$ is right-consistent, then so
are the formulas in all steps. From the completeness of the $\PTLU$ axiom
system as discussed in \S\ref{ptl-with-until-subsec} we have that the
conjunction $Y\And X$ is satisfiable.  Furthermore, our techniques ensure that
the models of a formula obtained from one of the transformations also satisfy
the immediately preceding formula and hence by transitivity the original
$\PITL$ formula $A$ as well.

\begin{mycomment}
  Observe that by the definition of $\PITLK$, any chop-star formula which
  occurs in a $\PITLK$ formula $K$ is located in the left of some chop.  It is
  not hard to see that any such chop-star formula is in fact located somewhere
  in the left side of a right-chop.  Also, it is not hard to check that any
  right-chop of a $\PITLK$ formula is itself in $\PITLK$ since any chop-stars
  in the right-chop must occur in its left operand.
\end{mycomment}

Important automata-theoretic techniques presented in
Sections~\ref{deterministic-semi-automata-and-automata-sec}
and~\ref{compound-semi-automata-for-suffix-recognition-sec} help with the
reductions to $K'\And \Bf D'$ and $Y\And X$ in
Section~\ref{reduction-of-pitl-to-ptlu-sec}.  We show in
Section~\ref{reduction-of-pitl-to-ptlu-sec} that the formulas $K\And \Bf D'$,
$K'\And \Bf D'$ and $Y\And X$ are deducibly equivalent.

Note that in the actual completeness proof (in
Lemma~\ref{completeness-for-pitlk-lem} in
\S\ref{proof-of-the-main-completeness-theorem-subsec}), which for technical
reasons involves a sequence of transformations from $K$ to $K'$, we make use
of a $\PITLK$ formula denoted $K'_{m+1}$ rather than simply $K'$.

\begin{mycomment}
We start by deducing $A$'s equivalence with a $\PITLK$ formula $K$ as a
right-theorem in Section~\ref{reduction-of-chop-omega-sec}.  Let $m$ be the
number of $K$'s right-chops. We obtain from $K$ another $\PITLK$ formula $K'$
(actually $K'_{m+1}$
in~\S\ref{proof-of-the-main-completeness-theorem-subsec}).  Left operands
$B_1,\ldots, B_m$ of all of $K$'s $m$ right-chops are replaced in $K'$ by
chain formulas $G_1,\ldots, G_m$ using techniques developed in
Sections~\ref{deterministic-semi-automata-and-automata-sec}--\relax
\ref{reduction-of-pitl-to-ptlu-sec}.  This exploits an automata-theoretic way
to do \emph{infix recognition}, that is, to test which of an interval's
finite-time infix subintervals satisfy some given $\PITL$ formula.  We employ
$m$ deducible implications of the form below (see later
Lemma~\ref{infix-recognition-lem})), where $B$ is the left operand of one of
$K$'s $m$ right-chops and both the ATA $D$ and chain formula $G$ are
constructed from $B$:
\begin{equation}
  \label{summary-of-role-of-pitl-subsets-in-rest-of-paper-1-eq}
%    \label{ata-and-ptlu-formula-2-eq}
  \TheoremR \Bf D \Implies \Box\Bf(B\equiv G)
  \dotspace.
\end{equation}
The subformula $B$ in the formula $\Box\Bf(B\equiv G)$ can occur in some
finite infix subintervals of the implication's own interval.  The overall
result of replacement in $K'$ is expressed by the right-theorem $\thmR \Bf D'
\imp (K\equiv K')$.  Here the ATA $D'$ is itself a conjunction of the $m$
disjoint ATAs $D_1\ldots, D_m$ obtained from the $m$ instances
of~\eqref{summary-of-role-of-pitl-subsets-in-rest-of-paper-1-eq} for the left
operands $B_1,\ldots, B_m$ of the $K$'s $m$ right-chops.  Now the formula $K
\And \Bf D'$ is right-consistent (by
Lemma~\ref{temporal-operator-bf-atas-and-consistency-lem}) and therefore so is
$K' \And \Bf D'$.  Formulas $K'$ and $\Bf D'$ are then re-expressed as
equivalent $\PTLU$ and $\PTL$ formulas $Y$ and $X$, respectively, so we have
the right-consistent $\PTLU$ formula $Y \And X$.  As we stated in
Section~\ref{right-instances-right-variables-and-right-theorems-sec},
right-completeness is identical to every right-consistent formula being
satisfiable.  So right-completeness for $\PTLU$ yields a model for $Y \And
X$. The model also satisfies $K' \And \Bf D'$ and $K$ and so the original
right-consistent formula $A$ as well. This shows right-completeness for
$\PITL$.

For example, if $K$ is the conjunction $\bigl(B_1\chop (B_2\chop\Empty)\bigr)
\And \Not\Skip$ with two right-chops, we construct from $B_1$ and $B_2$ two
ATAs $D_1$ and $D_2$ and chain formulas $G_1$ and $G_2$.  Take $D'$ to be
$D_1\And D_2$.  Now $K$'s right-consistency and the earlier
Lemma~\ref{temporal-operator-bf-atas-and-consistency-lem} about ATAs ensure $K
\And \Bf D'$ is right-consistent.  We take $K'$ to be $\bigl(G_1\chop
(G_2\chop\Empty)\bigr) \And \Not\Skip$ and have the right-theorem $\thmR \Bf
D' \imp (K\equiv K')$ by suitably using
implication~\eqref{summary-of-role-of-pitl-subsets-in-rest-of-paper-1-eq}.
Hence $K' \And \Bf D'$ is itself right-consistent.  We now re-express it in
$\PTLU$ as $Y \And X$ which is satisfiable and yields a model for $K$ and the
original formula $A$.

Here is an informal sequence of the transformations just presented:
\begin{displaymath}
  A
  \quad
  \xrightarrow{\text{Section~\ref{reduction-of-chop-omega-sec}}}
  \quad
  K
  \quad
  \xrightarrow{\text{Section~\ref{reduction-of-pitl-to-ptlu-sec}}}
  \quad 
  K'\And \Bf D'
  \quad
  \xrightarrow{\text{Section~\ref{reduction-of-pitl-to-ptlu-sec}}}
  \quad 
  Y\And X
\dotspace.
\end{displaymath}
As we noted earlier,
Sections~\ref{deterministic-semi-automata-and-automata-sec}
and~\ref{compound-semi-automata-for-suffix-recognition-sec} contain important
automata-theoretic techniques which help with the reductions to $K'\And \Bf D'$
and $Y\And X$ in Section~\ref{reduction-of-pitl-to-ptlu-sec}.  For example, in
Section~\ref{compound-semi-automata-for-suffix-recognition-sec} we show the
validity of a version of
implication~\eqref{summary-of-role-of-pitl-subsets-in-rest-of-paper-1-eq}
which is restricted to finite time and recognises \emph{suffix subintervals}
rather than infix ones:
\begin{equation*}
  \Valid \Finite \And \Bf D \Implies \Box(B\equiv G)
  \dotspace.
\end{equation*}

\end{mycomment}

% In essence, a formula $Y$ in $\PTLU$ can be re-expressed as a $\PTL$ formula
% $X$ with some auxiliary variables which mimic the behaviour of the Untils in
% $Y$.  If $p_1$, $\ldots$, $p_n$ are these variables, then $Y$ and
% $\Exists{p_1\ldots p_n}X$ are semantically equivalent and the implication
% $X\imp Y$ is valid.  This implication can be used to reduce axiomatic
% completeness for $Y$ in $\PTLU$ to it for $X$ in $\PTL$.  We first show
% that the consistency of $Y$ ensures the consistency of $X$.  Completeness for
% $\PTL$ yields a model for $X$ and the valid implication $X\imp Y$ then
% guarantees that this model also satisfies $Y$.

\begin{mycomment}
Table~\ref{summary-of-completeness-proof-table} shows the various remaining
steps in the completeness proof for $\PITL$.  These
primarily concern automata-theoretic techniques to obtain axiomatic
completeness for $\PITLK$.
\begin{table*}%[h]
  \begin{center}
    \begin{tabular}{lll}
      Reduction of right-completeness for $\PITL$ to it for $\PITLK$:
      & Section~\ref{reduction-of-chop-omega-sec} \\
      Introduction of deterministic finite-state semi-automata and automata:
      & Section~\ref{deterministic-semi-automata-and-automata-sec} \\
      Introduction of a compound semi-automaton for suffix recognition:
      & Section~\ref{compound-semi-automata-for-suffix-recognition-sec} \\
      Reduction of right-completeness for  $\PITLK$ to $\PTLU$:
      & Section~\ref{reduction-of-pitl-to-ptlu-sec}
    \end{tabular}
    \caption{Summary of remainder of $\PITL$ axiomatic completeness proof}
    \label{summary-of-completeness-proof-table}
  \end{center}
\end{table*}
\end{mycomment}

\section{Reduction of Chop-Omega}

\label{reduction-of-chop-omega-sec}

If we assume right-completeness for $\PITLK$ (later proved as
Lemma~\ref{completeness-for-pitlk-lem}
in~\S\ref{proof-of-the-main-completeness-theorem-subsec}), then obtaining from
a $\PITL$ formula a deducibly equivalent $\PITLK$ one is relatively easy.  We
first look at re-expressing chop-omega formulas in $\PITLK$ and then extend
this to arbitrary $\PITL$ formulas.

\begin{mylemma}[Deducible Re-Expression of Chop-Omega in $\PITLK$]
  \label{deducible-re-expression-of-chop-omega-lem}
  Suppose we have right-completeness for $\PITLK$.  Then for any $\PITL$
  formula $B$, there exists a $\PITLK$ formula $K$ with the same variables and
  no right-variables and for which the equivalence $K\equiv B\Chopomega$ is a
  right-theorem (i.e., $\thmR K\equiv B\Chopomega$).
\end{mylemma}
\begin{proof}
  Thomas' theorem (Theorem~\ref{thomas'-theorem-thm}) ensures that there
  exists some $\PITLK$ formula which is semantically equivalent to
  $B\Chopomega$ and contains the same variables.  From that formula we obtain
  one denoted here as $K$ which has no right-variables by conjoining a
  trivially true $\Df$-formula containing a disjunction of all of $B$'s
  variables and their negations.  We therefore have $\vld K\equiv B\Chopomega$
  and now deduce $\thmR K\equiv B\Chopomega$:

  \textbf{Case for showing $\boldsymbol{\theoremR K\imp B\Chopomega}$:}

  The first step involves an instance of Axiom~\refid{ChopOmegaInduct}:
  \begin{equation}
    \label{deducible-re-expression-of-chop-omega-1-eq} 
    \theoremR K \Andd \Box(K \imp (B \And \More)\chop K) \Implies B\Chopomega
\dotspace.
  \end{equation}
  In addition, the next formula is valid:
  \begin{displaymath}
    \Valid B\Chopomega \implies (B \And \More)\chop B\Chopomega
\dotspace.
  \end{displaymath}
  From this and $\vld K \equiv B\Chopomega$, we have $\vld K \imp (B \And
  \More)\chop K$.  We then use the assumed right-completeness of $\PITLK$ to
  deduce the implication as a right-theorem.  Now invoke $\Box$-generalisation
  (Axiom~\refid{BoxGen}) on this to obtain $\thmR \Box(K \imp (B \And
  \More)\chop K)$.  Simple propositional reasoning involving that and the
  earlier deduced
  implication~\eqref{deducible-re-expression-of-chop-omega-1-eq} establishes
  our immediate goal $\theoremR K\imp B\Chopomega$.

  \textbf{Case for showing $\boldsymbol{\theoremR B\Chopomega \imp K}$}:

  Let $p$ be a propositional variable not in $B\Chopomega$ or $K$.  The next
  formula is valid (and an instance of Axiom~\refid{ChopOmegaInduct}):
  \begin{displaymath}
    \Valid p \Andd \Box(p \imp (B \And \More)\chop p) \Implies B\Chopomega
\dotspace.
  \end{displaymath}
  We then replace $B\Chopomega$ by the semantically equivalent $K$:
  \begin{equation}
    \label{deducible-re-expression-of-chop-omega-2-eq} 
    \Valid p \Andd \Box(p \imp (B \And \More)\chop p) \Implies K
\dotspace.
  \end{equation}
  Now $K$ is a $\PITLK$ formula and furthermore $(B \And \More)\chop p$ is as
  well since even if $B$ does contain some chop-stars, $B$ is located within
  the left of a chop.  The valid
  formula~\eqref{deducible-re-expression-of-chop-omega-2-eq} is in $\PITLK$
  and hence a right-theorem by the assumed right-completeness for $\PITLK$:
  \begin{equation*}
    \TheoremR p \Andd \Box(p \imp (B \And \More)\chop p) \Implies K
\dotspace.
  \end{equation*}
  Therefore, we can use
  Lemma~\ref{substitution-instances-of-right-theorems-lem} to obtain the
  theoremhood of the next $\PITL$ implication which has the formula
  $B\Chopomega$ substituted into the right-variable $p$:
  \begin{equation}
    \label{deducible-re-expression-of-chop-omega-3-eq} 
    \theoremR\; B\Chopomega
      \And \Box\bigl(B\Chopomega \imp (B \And \More)\chop B\Chopomega\bigr)
        \Implies K
\dotspace.
  \end{equation}
  We also deduce the following from the definition of chop-omega in terms of
  chop-star together with Axiom~\refid{SChopStarEqv} and some simple temporal
  reasoning:
  \begin{displaymath}
     \TheoremR B\Chopomega \implies (B \And \More)\chop B\Chopomega
\dotspace.
  \end{displaymath}
  We now do $\Box$-generalisation (Axiom~\refid{BoxGen}) on this and then use
  propositional reasoning on it with the previous
  formula~\eqref{deducible-re-expression-of-chop-omega-3-eq} to obtain the
  right-theorem $\thmR B\Chopomega \imp K$, which is our immediate goal.
\end{proof}

\begin{mylemma}[Reduction of $\PITL$ to $\PITLK$]
  \label{reduction-of-pitl-pitlk-lem}
%  \label{deducible-elimination-of-full-chop-omega-lem}
%  \label{reduction-of-completeness-for-pitl-to-it-for-pitlk-lem}
  If right-completeness holds for $\PITLK$, then for any $\PITL$ formula $A$,
  there exists an equivalent $\PITLK$ formula $K$ with exactly the same
  propositional variables and right-variables such that $\thmR A\equiv K$.
\end{mylemma}
\begin{proof}
  We first re-express each of $A$'s chop-stars $B_i\SChopstar$ not in the left
  of chop or another chop-star using the next deducible equivalence (see
  $\PITL$ Theorem~\refid{SChopStarEqvSChopstarChopEmptyOrChopOmega} found in
  Appendix~\ref{some-pitl-theorems-and-their-proofs-sec} and also included in
  the more abbreviated
  Table~\ref{list-of-pitl-theorems-and-derived-rules-mentioned-before-app-table}
  in \S\ref{formal-equivalence-of-the-two-representations-of-runs-subsec}):
  \begin{equation}
    \label{reduction-of-pitl-pitlk-1-eq}
    \TheoremR B_i\SChopstar \Equiv (B_i\SChopstar\chop\Empty) \Or B_i\Chopomega
\dotspace.
  \end{equation}
  This splits $B_i\SChopstar$ into cases for finite and infinite time.  Note
  that there there are no right-variables
  in~\eqref{reduction-of-pitl-pitlk-1-eq} since any
  variables occur in a chop-star.  Hence the equivalence, once deduced, is
  trivially a right-theorem.

  Lemma~\ref{deducible-re-expression-of-chop-omega-lem} ensures some $\PITLK$
  formula $K'_i$ exists with the same variables as $B_i$, no right-variables
  and the right-theorem $\thmR K'_i \equiv B_i\Chopomega$.
%   Without loss of generality assume
%   that $K$, like $B_i\Chopomega$, has no right variables.
%   It can be achieved
%   by ensuring $K$ contains a spurious conjunction of disjunctions of each
%   variable and its negation somewhere on the left of a chop.  For example, if
%   $V_B_i=\{q,r\}$, then $K$ itself can be the conjunction below:
%   \begin{equation*}
%     K'_i \Andd \Df((q \Or \Not q) \And (r\Or \Not r)\bigr)
% \dotspace,
%   \end{equation*}
%   for some $K'_i$ with $V_{K'_i}=V_B_i$ and $\vld K'_i\equiv B_i\Chopomega$.
  Hence like~\eqref {reduction-of-pitl-pitlk-1-eq},
  the next equivalence is a right-theorem and both sides have the same
  variables and no right-variables:
  \begin{equation*}
    \TheoremR B_i\SChopstar \Equiv (B_i\SChopstar\chop\Empty) \Or K'_i
\dotspace.
  \end{equation*}
  Then Right Replacement
  (Lemma~\ref{right-replacement-derived-rule-lemma-for-pitl-lem}) in $A$ of
  each $B_i\SChopstar$ by $(B_i\SChopstar\chop\Empty)\Or K'_i$ yields a
  $\PITLK$ formula $K$ which the same variables as $A$ and equivalent to it
  (i.e., $\thmR A\equiv K$).  No right-variables in $A$ are in any replaced
  $B_i\SChopstar$. Hence $A$ and $K$ have the same right-variables.
%   Alternatively, the next
%   deducible equivalence, based on of all this, can be used with Right
%   Replacement:
%   \begin{equation*}
%     \Theorem B_i\SChopstar \Equiv (B_i\SChopstar\chop\Empty) \Or K
% \dotspace.
%   \end{equation*}
\end{proof}

\section{Deterministic Finite-State Semi-Automata And Automata}

\label{deterministic-semi-automata-and-automata-sec}

The remainder of our axiomatic completeness proof for $\PITL$ mostly concerns
reducing $\PITLK$ to $\PTLU$.  Now $\PITL$ with finite time expresses the
regular languages and can readily encode regular expressions (see for
example~\cite{Moszkowski04a} which reproduces our results with J.~Halpern
in~\cite{Moszkowski83a}).  We can therefore employ some kinds of
\emph{deterministic finite-state semi-automata and automata} which provide a
convenient low-level framework for finite time to encode the behaviour of an
arbitrary $\PITL$ formula.  Our completeness proof utilises these
semi-automata and automata to build a variant semi-automaton discussed in the
next Section~\ref{compound-semi-automata-for-suffix-recognition-sec} to assist
in reducing $\PITL$ formulas on the left of right-chops to chain formulas in
$\PTLU$.  The reduction applying these techniques to go from $\PITLK$ to
$\PTLU$ is presented in Section~\ref{reduction-of-pitl-to-ptlu-sec}.

After introducing the semi-automata and automata, we will consider various
semantically equivalent ways to represent them in temporal logic, each with
its benefits.  Some require $\PITL$ and others just $\PTL$.  The
representations in $\PITL$ are at a higher level and fit well with our proof
system, especially since we can assume completeness for $\PITL$ with finite
time.  In some later sections, we consider deducing some of the properties as
theorems.

In order to define an \emphbf{alphabet} for our semi-automata and automata, we
introduce a special kind of state formula which serves as a \emphbf{letter}
and is called here an \emphbf{atom}.  An atom is any finite conjunction in
which each conjunct is some propositional variable or its negation and no two
conjuncts share the same variable.  The Greek letters $\alpha$ and $\beta$
denote an individual atom.  For any finite set of propositional variables
$\VV$, let $\boldsymbol{\Sigma_V}$ be some set of $2^{\size{V}}$ logically
distinct atoms containing exactly the variables in $\VV$.  For example, if
$\VV=\{p,q\}$, we can let $\Sigma_{\VV}$ be the set of the four atoms shown
below:
\begin{displaymath}
  p \And q
  \qquad
  p \And \Not q
  \qquad
  \Not p \And q
  \qquad
  \Not p \And \Not q
\dotspace.
\end{displaymath}
One simple convention is to assume that the propositional variables in an atom
occur from left to right in lexical order.  If $V$ is the empty set, then 
 $\Sigma_V$ contains just the formula $\True$.

 A finite, nonempty sequence of atoms form a \emphbf{word}.  Each possible
 word corresponds to some collective state-by-state behaviour of the selected
 variables in a finite interval.  For our interval-oriented application of
 words we never utilise the word containing no letters (commonly denoted
 $\epsilon$ in the literature).

\subsection{Deterministic Finite-State Semi-Automata}

\label{deterministic-finite-state-semi-automata-subsec}

We define a \emphbf{deterministic finite-state semi-automaton} $S$ to be a
quadruple $(V_S, Q_S, q_S^I, \delta_S)$ consisting of a finite set of
propositional variables $V_S$, together with a finite, nonempty set of
\emphbf{control states} $Q_S=\{q_1, \ldots, q_m\}$, an \emphbf{initial control
  state} $q_S^I\in Q_S$ and a \emphbf{deterministic transition function}
$\delta_S\colon Q_S\times \Sigma_{V_S} \rightarrow Q_S$.  The sets $V_S$ and
$Q_S$ must be disjoint, i.e., $V_S\intersection Q_S=\emptyset$.  We use
propositional variables $q_1, \dots, q_m$ to denote control states since this
helps when expressing the semi-automaton's behaviour in $\PITL$.  A
\emphbf{run} on a finite word $\alpha_1\ldots \alpha_k$ in $\Sigma_{V_S}^+$
with $k$ atoms is a sequence of $k$ control states $q'_1\ldots q'_k$ all in
$Q_S$ with $q'_1=q_S^I$ and $\delta_S(q'_i,\alpha_i)=q'_{i+1}$ for each
$i\colon 1\le i\lt k$.  Hence the semi-automaton makes just $k-1$ transitions
and consequently ignores the details of the last atom $\alpha_k$.  Therefore
the semi-automaton differs from a conventional automaton which would have a
run with $k+1$ control states involving $k$ transitions and the examination of
all $k$ atoms.  Furthermore, the definition of a semi-automaton has no set of
\emph{final control states} and hence no \emph{acceptance condition}.  We
abbreviate the set of atoms $\Sigma_{V_S}$ as $\Sigma_S$ since the elements of
$\Sigma_{V_S}$ serve as $S$'s letters.

The semi-automaton $S$'s behaviour is expressible in temporal logic by
regarding each control state $q_i$ to be a propositional variable which is
true when $q_i$ is $S$'s current control state. Before showing how $S$'s runs
are expressed in $\PTL$, we first define a state formula $\init_S$ which
ensures that the initial control state is $q_S^I$ and also a transitional
formula $T_S$ in $\NLone$ which captures the behaviour of $\delta_S$:
\begin{center}
  \begin{tabular}{L@{\enspace}L}
    \init_S\colon
      & \displaystyle
        q_S^I \And \bigwedge_{q \in Q_S\colon q\ne q_S^I}\Not q \\[3pt]
    T_S\colon
      & \displaystyle
        \bigwedge_{q \in Q_S}
        \bigl((\Next q) \equiv
        \bigvee_{q' \in Q_S}
	\,
        \bigvee_{\alpha\in \Sigma_S\colon \delta_S(q',\alpha)=q}
            (q' \And \alpha)\bigr)
\dotspace.
  \end{tabular}
\end{center}  
If we assume finite time, then a run starting at $S$'s initial control state
is expressed as the $\PTL$ formula $\init_S \And \Box(\More \imp T_S)$ or
alternatively as the chain formula $\init_S \And (T_S \Until \Empty)$ in
$\PTLU$.

\subsection{Deterministic Finite-State Automata}

\label{deterministic-finite-state-automata-subsec}

Semi-automata do not have an acceptance test and hence do not have associated
\emph{accepting} runs.  We therefore now define a \emphbf{deterministic
  finite-state automaton} which includes an acceptance test.  As we shortly
illustrate, this can be constructed to recognise a given $\PITL$ formula in a
finite interval.  Let $M$ be a quintuple $(V_M, Q_M, q_M^I, \delta_M,
\tau_M)$.  The first four entries are as for a semi-automaton.  The last entry
$\tau_M \colon Q_M\rightarrow 2^{\Sigma_M}$ is a \emphbf{conditional
  acceptance function} from control states to sets of letters.  A run is the
same as for a semi-automaton.  Our notion of acceptance of a word does not use
a conventional \emph{set of final control states} but instead has the function
$\tau_M$ make all control states \emphbf{conditionally final}.  An
\emphbf{accepting run} on a finite word $\alpha_1\ldots \alpha_k$ in
$\Sigma_M^+$ with $k$ atoms is any run of $k$ control states $q'_1\ldots q'_k$
with $q'_k\in\tau_M(\alpha_k)$.  Therefore, a control state $q\in Q_M$ is
regarded as a final one only when the automaton sees an atom $\alpha$ with
$\alpha\in\tau_M(q)$.  A test for this is expressible as the state formula
$\acc_M$ defined below:
\begin{displaymath}
  \acc_M \colon\enspace 
    \bigvee_{q \in Q_M}\bigvee_{\alpha\in \tau_M(q)} (q \And \alpha)
\dotspace.
\end{displaymath}

If we assume finite time, an accepting run of $M$ starting at $M$'s initial
control state is expressed as the $\PTL$ formula $\init_M \And \Box(\More \imp
T_M) \And \Fin\!\acc_M$ or alternatively as the chain formula $\init_M \And
(T_M \Until (\acc_M \And \Empty))$ in $\PTLU$. As a result of our convention
for runs and accepting runs, the automaton $M$'s  operation requires one state
less than a conventional one to accept a word.  For example, it can accept
one-letter words without the need for any state transitions. In fact, such an
automaton $M$ only recognises words with at least one letter (i.e., in
$\Sigma_M^+$).  This is perfect when we utilise semi-automata and automata to
mimic $\PITL$ formulas since $\ITL$ intervals have at least one state.

The regular expressiveness of $\PITL$ with finite time ensures that any
$\PITL$ formula $B$ can be recognised by some $M$.  The set $V_B$ of
propositional variables in $B$ and the set $Q_M$ of $M$'s control states are
assumed to be distinct.  Formally, we have the next valid formula expressed in
$\QPITL$ (defined in Section~\ref{pitl-sec}):
\begin{equation*}
  \Valid \Finite \Implies
     B \EQUIV \Exists{q_1,\ldots,q_{\size{Q_M}}}
      \bigl(\init_M \Andd \Box(\More \imp T_M) \Andd \Fin\!\acc_M\bigr)
\dotspace.
\end{equation*}
For instance, below is a sample automaton $M$ to recognise finite intervals
satisfying the formula $(\Skip \And p)\chop \Skip \chop \Skip\SChopstar \chop
(\Empty \And \Not p)$, which is semantically equivalent to the $\PTL$ formula
$p \And \Next\Next\Diamond(\Empty \And \Not p)$:
\begin{equation}
  \label{deterministic-semi-automaton-and-automaton-1-eq}
  \begin{array}{l}
    V_M=\{p\} \text{\ (so $\Sigma_M=\{p, \Not p\}$)}
      \quad Q_M = \{ q_1, q_2, q_3, q_4 \}
      \quad q_M^I=q_1 \\
    \delta_M(q_1,p) = q_2 \quad \delta_M(q_1,\Not p) = q_4
      \quad \delta_M(q_2,p) =\delta_M(q_2,\Not p) = q_3 \\
    \delta_M(q_3,p) =\delta_M(q_3,\Not p) = q_3
      \quad \delta_M(q_4,p) =\delta_M(q_4,\Not p) = q_4 \\
    \tau_M(q_1)=\tau_M(q_2)=\tau_M(q_4)=\{\}
      \quad \tau_M(q_3)=\{\Not p\}
  \end{array}
\end{equation}
Here is an accepting run for the 5-letter word $p\,\Not p\, p\,p\, \Not
p$:\quad $q_1\, q_2\, q_3\, q_3\, q_3$:
\begin{center}
  \begin{tabular}{L@{\enspace}L}
    \init_M\colon & q_1 \And \Not q_2 \And \Not q_3 \And \Not q_4
      \quad \acc_M\colon\enspace q_3 \And \Not p \\
    T_M\colon & (\Next q_1)\equiv \False
      \,\Andd\, (\Next q_2)\equiv (q_1 \And p) \\
    & \Andd\, (\Next q_3)\equiv (q_2 \Or q_3)
        \,\Andd\, (\Next q_4)\equiv ((q_1 \And \Not p) \Or q_4) \\
    \multicolumn{2}{l}{Accepting run in $\PTL$:
      $\Finite \And \init_M \And \Box(\More \imp T_M) \And \Fin\!\acc_M$}
  \end{tabular}
\end{center}
Below are the values of $q_1, \ldots, q_4$ over an associated 5-state interval
in which $p$ has the behaviour $p\,\Not p\, p\, p\, \Not p$:
\begin{equation}
  \label{deterministic-semi-automaton-and-automaton-2-eq}
  \begin{array}{l}
   (\boldsymbol{q_1}, \Not q_2, \Not q_3, \Not q_4)
      \quad (\Not q_1, \boldsymbol{q_2}, \Not q_3, \Not q_4)
      \quad (\Not q_1, \Not q_2, \boldsymbol{q_3}, \Not q_4) \\
    \qquad (\Not q_1, \Not q_2, \boldsymbol{q_3}, \Not q_4)
      \quad (\Not q_1, \Not q_2, \boldsymbol{q_3}, \Not q_4)
\dotspace.
  \end{array}
\end{equation}
In each tuple, we show the unique active control state in boldface.  For
instance, $q_2$ is true in the second interval state since $q_1 \And p$ is
true in the first one.

\subsection{ATAs for Semi-Automata and Automata}

\label{atas-for-semi-automata-and-automata-subsec}

The runs of a deterministic semi-automaton or deterministic automaton from the
initial control state can alternatively be expressed with an ATA (defined in
\S\ref{auxiliary-temporal-assignments-subsec}).  We will consider the case for
a semi-automaton $S$, but the technique is identical for an automaton $M$.
Now $\PITL$ with finite time can express all regular languages in
$\Sigma_S^+$.  For each control state $q$ of $S$, the set of words in
$\Sigma_S^+$ for which $S$ starts in the initial control state $q_S^I$ and
ends in $q$ is regular.  The regular expressiveness of $\PITL$ with finite
time ensures that there exists some corresponding $\PITL$ formula $C_{S,q}$
which only has variables in the set $\VV_S$ and expresses this set of words.
In principle, such a formula can be obtained by adapting standard techniques
for constructing a regular expression from a conventional finite-state
automaton.  Now let the ATA $D_S$ denote the conjunction $ \bigwedge_{q \in
  Q_S}(q\Tassign C_{S,q})$. We express finite runs in $\PITL$ using $\Finite
\And \Bf D_S$.  Here is such an ATA for the earlier sample automaton
in~\eqref{deterministic-semi-automaton-and-automaton-1-eq}:
\begin{equation*}
  q_1{\Tassign}(\Empty \And p)
  \,\Andd\, q_2{\Tassign}(\Skip \And p)
  \,\Andd\, q_3{\Tassign}(\Skip \And p)\chop\Skip\chop\Skip\SChopstar
  \,\Andd\, q_4{\Tassign} (\More \And \Not p)
\dotspace.
\end{equation*}
Note that the case for $q_3$ simplifies to $q_3\Tassign (p \And
\Next\Next\True)$.  The 5-tuple sample run
in~\eqref{deterministic-semi-automaton-and-automaton-2-eq} reflects behaviour
in prefix subintervals for the previous illustrative word $p\,\Not p\, p\, p\,
\Not p$.  For example, $q_2$ is true in just the second interval state since
the 2-state prefix subinterval is the only prefix subinterval satisfying the
formula $\Skip \And p$.

For any deterministic automaton $M$, let $D_M$ denote some ATA obtained from
$M$ in exactly the same way as for a semi-automaton.

\subsection{Formal Equivalence of the Two Representations of Runs}

\label{formal-equivalence-of-the-two-representations-of-runs-subsec}

For finite time, the $\PITL$ formula $\Bf D_S$ expresses all runs of $S$
starting from its initial control state.  Hence for finite time this formula
is semantically equivalent to the previous formulas for this behaviour (e.g.,
the $\PTL$ formula $\init_S \And \Box(\More \imp T_S)$).  Consequently, the
next valid formula relates the two ways of expressing $S$'s runs:
\begin{equation}
  \label{auxiliary-temporal-assignments-1-eq}
  \textstyle
  \Valid
      \Finite \implies
        \Bigl(\Bf D_S \,\EQUIV\,
	   \bigl(\mathit{init}_S \And \Box(\More \imp T_S)\bigr)\Bigr)
\dotspace.
\end{equation}
The use of a single
example~\eqref{deterministic-semi-automaton-and-automaton-1-eq} for both
representations of $S$'s runs can be justified from this.  An automaton $M$'s
accepting runs can be expressed with $\Finite \And (\Bf D_M) \And
\Fin\!\acc_M$.  The $\QPITL$ formula below is valid for any $\PITL$ formula
$B$ and automaton $M$ which recognises $B$:
\begin{equation*}
  \Valid \Finite \Implies
     B \EQUIV
       \Exists{q_1,\ldots,q_{\size{Q_M}}}\bigl(\Bf D_M \Andd \Fin\!\acc_M\bigr)
\dotspace.
\end{equation*}

The valid $\PITLK$ formula~\eqref{auxiliary-temporal-assignments-1-eq} just
given relates two ways of representing in temporal logic the runs of a
finite-state semi-automaton (that is, $\Bf D_S$ and $\mathit{init}_S \And
\Box(\More \imp T_S)$).  It includes an explicit assumption about finite time.
The next Lemma~\ref{deductive-eqv-of-two-reps-of-auto-runs-lem} eliminates
this requirement and provides a way to re-express $\Bf D_S$ as an equivalent
$\PTL$ formula in deductions concerning infinite time.  The proof of
Lemma~\ref{deductive-eqv-of-two-reps-of-auto-runs-lem} only involves temporal
logic and requires no explicit knowledge about omega automata.

For the convenience of readers studying our deductions here and later on in
Section~\ref{reduction-of-pitl-to-ptlu-sec},
Table~\ref{list-of-pitl-theorems-and-derived-rules-mentioned-before-app-table}
lists every $\PITL$ theorem and derived rule explicitly mentioned somewhere
prior to Appendix~\ref{some-pitl-theorems-and-their-proofs-sec}.  The appendix
itself contains all needed $\PITL$ theorems and derived rules and as well as
their individual proofs.  \bgroup
\def\MyThmsOne#1&#2\EndMyThmsOne{%\csname #1-used-in-mythms\endcsname
  \expandafter\ifx \csname #1-used-in-mythms\endcsname\relax\else
  \refid{#1}&#2\\\fi} \ifx\MyThmsMacroA\relax \else
%   \message{2: Found MyThmsMacroA!!} %
  \begin{table*}%[h]
    \begin{center}
      \begin{tabular}{l@{\qquad}L}
        \MyThmsMacroA
      \end{tabular}
      \caption{$\PITL$ theorems and derived rules mentioned before Appendix~\ref{some-pitl-theorems-and-their-proofs-sec}}
      \label{list-of-pitl-theorems-and-derived-rules-mentioned-before-app-table}
    \end{center}
  \end{table*}
  \fi
\egroup

\begin{mylemma}
  \label{deductive-eqv-of-two-reps-of-auto-runs-lem}
  For any deterministic finite-state semi-automaton $S$, the next $\PITLK$
  equivalence involving $S$'s ATA $D_S$ and a $\PTL$ formula is a $\PITL$
  theorem:
  \begin{equation}
    \label{deductive-eqv-of-two-reps-of-auto-runs-2-eq}
    \textstyle
    \Theorem \Bf D_S
	\Equiv \bigl(\mathit{init}_S \And \Box(\More \imp T_S)\bigr)
\dotspace.
  \end{equation}
\end{mylemma}
\begin{proof}
  The validity of implication~\eqref{auxiliary-temporal-assignments-1-eq},
  together with completeness for $\PITL$ with finite time ensures
  that~\eqref{auxiliary-temporal-assignments-1-eq} is also a deducible
  theorem:
  \begin{equation*}
    \textstyle
    \Theorem
      \Finite \Implies
	\Bigl(\Bf D_S \,\EQUIV\,
	  \bigl(\mathit{init}_S \And \Box(\More \imp T_S)\bigr)\Bigr)
\dotspace.
  \end{equation*}
  We then deduce from that and Inference Rule~\refid{BfFGen} the next theorem:
  \begin{equation*}
    \textstyle
    \Theorem
      \Bf\Bigl(\Bf D_S \,\EQUIV\,
	\bigl(\mathit{init}_S \And \Box(\More \imp T_S)\bigr)\Bigr)
\dotspace.
  \end{equation*}
  From this and some interval-based temporal reasoning about $\Bf$ (using
  properties of the underlying modal system \ModalSys{K} -- see
  Appendix~\relax
  \ref{some-properties-of-bf-involving-the-modal-system-k-and-axiom-d-subsec})
  we can then deduce the equivalence below:
  \begin{equation*}
    \textstyle
    \Theorem \Bf D_S \Equiv \Bf \mathit{init}_S \Andd \Bf\Box(\More \imp T_S)
\dotspace.
  \end{equation*}
  Let us now re-express $\Bf \mathit{init}_S$ as the equivalent state formula
  $\mathit{init}_S$ (see PITL Theorem~\refid{BfState}):
  \begin{equation*}
    \textstyle
    \Theorem \Bf D_S \Equiv \mathit{init}_S \And \Bf\Box(\More \imp T_S)
\dotspace.
  \end{equation*}
  We also want to re-express $\Bf\Box(\More \imp T_S)$ as the $\PTL$ formula
  $\Box(\More \imp T_S)$.  This can be done by first re-expressing $\Bf\Box$
  as $\Box\Bf$ (see $\PITL$ Theorem~\refid{BfBoxEqvBoxBf}) to yield the
  equivalence below:
  \begin{equation}
    \label{deductive-eqv-of-two-reps-of-auto-runs-3-eq}
    \textstyle
    \Theorem \Bf D_S \Equiv \mathit{init}_S \And \Box\Bf(\More \imp T_S)
\dotspace.
  \end{equation}

  Let us now consider how to eliminate the operator $\Bf$ in the subformula
  $\Box\Bf(\More \imp T_S)$.  The fact that any $\NLone$ formula $T$ only sees
  an interval's first two states ensures that the next equivalence is valid
  and also deducible (see $\PITL$
  Theorem~\refid{DfMoreAndNLoneEqvMoreAndNLone}):
  \begin{equation*}
    \Theorem \Df(\More \And T) \Equiv \More \And T
\dotspace.
  \end{equation*}
  A dual form (see $\PITL$ Theorem~\refid{BfMoreImpNLoneEqvMoreImpNLone}) is
  readily deduced for use with $T_S$:
  \begin{equation*}
    \Theorem \Bf(\More \imp T_S) \Equiv \More \imp T_S
\dotspace.
  \end{equation*}
  We employ this with Derived Rule~\refid{BoxEqvBox} to obtain an equivalence
  for eliminating the $\Bf$ operator in $\Box\Bf(\More \imp T_S)$:
  \begin{equation}
    \label{deductive-eqv-of-two-reps-of-auto-runs-4-eq}
    \Theorem \Box\Bf(\More \imp T_S) \Equiv \Box(\More \imp T_S)
    \dotspace.
  \end{equation}
  Equivalence~\eqref{deductive-eqv-of-two-reps-of-auto-runs-2-eq}'s
  theoremhood, which is our immediate goal, then readily follows by simple
  propositional reasoning from the deduced
  equivalences~\eqref{deductive-eqv-of-two-reps-of-auto-runs-3-eq}
  and~\eqref{deductive-eqv-of-two-reps-of-auto-runs-4-eq}.
\end{proof}

\section{Compound Semi-Automata for Suffix Recognition}

\label{compound-semi-automata-for-suffix-recognition-sec}

Let a \emphbf{compound semi-automaton} $R$ be a vector of semi-automata $S_1,
\ldots, S_n$ for some $n\ge 1$ with disjoint sets of control states.  We take
$V_R$ to be the set of propositional variables in the semi-automata $S_1,
\ldots, S_n$ which are not also control states.  The purpose of $R$ is to
perform what we call \emphbf{suffix recognition}. This is a way to determine
which of an finite interval's suffix subintervals satisfy some given $\PITL$
formula $B$.  Suffix recognition is a stepping stone enabling us to subsequently
perform the \emph{infix recognition} already briefly mentioned in
\S\ref{overview-of-role-of-pitl-subsets-in-completeness-proof-subsec}.  Later on
in Section~\ref{reduction-of-pitl-to-ptlu-sec} this feature of $R$ ensures that
for a given $\PITLK$ formula $K$ with $m$ right-chops (previously defined in
\S\ref{right-chops-and-chain-formulas-subsec}), we can utilise $m$ such compound
semi-automata to obtain an ATA for infix recognition to replace the left sides
of $K$'s right-chops with $\PTLU$ chain formulas (also introduced in
\S\ref{right-chops-and-chain-formulas-subsec}).  The $n$ individual
semi-automata $S_1, \ldots, S_n$ in $R$ are meant to operate lockstep in
parallel and so simultaneously make state transitions. For each $i: 1\le i\lt
n$, we require for the set $V_{S_{i+1}}$, which contains propositional variables
examined by $S_{i+1}$, that $V_{S_{i+1}} \subseteq V_{S_i}\union Q_{S_i}$.
Hence the control states of $S_i$ are allowed occur within the letters for
$S_{i+1}$ and any semi-automata of higher index but not vice versa.  This
enables each semi-automaton to optionally observe control states of all
semi-automata with lower index when it makes transitions.  In our particular
construction of $R$, the set $V_R$ simply equals the set $V_B$ of propositional
variables in the $\PITL$ formula $B$ and also equals the lowest-indexed
semi-automata $S_1$'s set $V_{S_1}$ of propositional variables used to form the
atoms $\Sigma_{S_1}$.  Let $R$'s ATA $D_R$ be a conjunction of the ATAs for the
semi-automata $S_1, \ldots, S_n$.  It is not hard to check that $D_R$ obeys the
ATA requirement limiting where auxiliary variables can occur (as specified in
the definition of ATAs in \S\ref{auxiliary-temporal-assignments-subsec}) and is
therefore well-formed.

We perform suffix recognition by exploiting standard techniques originally
developed by McNaughton~\citeyear{McNaughton66} to construct deterministic
omega automata.  Choueka~\citeyear{Choueka74} later applied McNaughton's
insights to some constructions for automata on finite words.  Our discussion
here likewise concerns finite-time behaviour and avoids omega automata.
Furthermore, this section deals with semantic issues but not deductions.

wil\subsection{Overview of Construction of Compound Semi-Automaton}

The compound semi-automaton $R$ to suffix recognise $B$ is built from several
modified copies of a deterministic automaton running lockstep in parallel.  We
also define an associated chain formula $G_R$.  Here is a summary:
\begin{iteMize}{$\bullet$}
\item We initially construct $R$ and $G_R$ to just check whether $B$ is true
  in any given finite suffix subinterval of the overall \emph{finite} interval
  in which $R$ is run.  Consequently, $G_R$ can be used to mimic $B$.
\item We first construct a deterministic finite-state automaton $M$ (discussed
  in \S\ref{deterministic-finite-state-automata-subsec}) to recognise the
  regular language associated with $B$ in finite time.  Let $n$ be the number of
  control states, that is, $n=\size{Q_M}$.
%   \item Let us assume that $M$'s initial control states are never re-entered
%     in $M$'s automaton runs.
\item We do not use $M$ directly but instead construct $n+1$ semi-automata
  $S_1$, \ldots, $S_{n+1}$ based on $M$. The compound semi-automaton $R$ is a
  vector of them.
\item Our construction ensures that always at least one semi-automaton is in
  (its copy of) $M$'s initial control state and so available to start
  testing for $B$ in the suffix subinterval commencing at the current state.
\item A suffix subinterval satisfies $B$ iff there is exists a simulation of
  an accepting run of $M$ which starts in the subinterval's first state, ends
  in its last one (the same as the overall interval's final state) and is
  formed by combining up to $n+1$ pieces of runs of the semi-automata $S_1$,
  \ldots, $S_{n+1}$.  The successive partial runs are performed on
  semi-automata of decreasing index.
\end{iteMize}

\subsection{Construction of the Individual Semi-Automata}

Let us now consider the details of the $n+1$ semi-automata variants $S_1$,
\ldots, $S_{n+1}$ of $M$.  A semi-automaton $S_k$ has its own disjoint set
$Q_{S_k}=\{q^{S_k}_1,\ldots, q^{S_k}_n\}$ of copies of the $n$ control states
in $M$ and is initialised exactly as $M$ would be and hence starts in (its
copy of) $M$'s initial control state.  We let $S_k$ examine the control states
of semi-automata with lower index (i.e., $S_1, \dots, S_{k-1}$) when it makes
its transitions in lockstep with them.  Hence, the set of propositional
variables $V_{S_k}$ is the union of $V_M$ and $\bigcup_{1\le j\lt k}Q_{S_j}$
and all propositional variables in an atom $\alpha$ in $\Sigma_{S_k}$ are
therefore either in $V_M$ or are control states in the semi-automata $S_1,
\dots, S_{k-1}$.

We now define the \emphbf{transition function $\boldsymbol{\delta_{S_k}}$} of
each semi-automaton $S_k$ in $R$ for use when all of the semi-automata operate
in lockstep.  The transition function $\delta_{S_k}\colon Q_{S_k}\times
\Sigma_{S_k}\rightarrow Q_{S_k}$ is deterministic like $M$'s, but more
complicated.  For each pair $\langle q_i^{S_k},\alpha\rangle$ in
$Q_{S_k}\times \Sigma_{S_k}$, there are two distinct possible cases based on
the values of $q_i^{S_k}$ and $\alpha$.  We now define these cases and the
associated transitions:
\begin{iteMize}{$\bullet$}
\item \emphbf{The pair $\boldsymbol{\langle q_i^{S_k},\alpha\rangle}$ is
    active:} This occurs when for every $j\lt k$, the pair's atom $\alpha$
  assigns the control variable $q_i^{S_j}$ to be false.  It corresponds to a
  situation where $S_k$ is the semi-automaton of lowest index in $R$ currently
  in (its own copy $q_i^{S_k}$ of) $M$'s control state $q_i^M$ and itself
  also called \emphbf{active}.

  Let $\beta\in\Sigma_M$ be the atom in $\Sigma_M$ obtained from $\alpha$ by
  only using the propositional variables in $V_M$ and thereby ignoring the
  control variables in $\alpha$.  Now we have that
  $\delta_M(q_i^M,\beta)=q_j^M$ for some $q_j^M\in Q_M$.  Define the
  transition $\delta_{S_k}(q_i^{S_k},\alpha)$ to be the corresponding
  $q_j^{S_k}\in Q_{S_k}$.

\item \emphbf{The pair $\boldsymbol{\langle q_i^{S_k},\alpha\rangle}$ is
    inactive:} If the first case does not apply, then $S_k$ shares (its copy
  of) $M$'s control state $q_i^M$ with some semi-automaton of lesser index as
  seen by $S_k$ via the atom $\alpha$.  We define the transition
  $\delta_{S_k}(q_i^{S_k},\alpha)$ to equal the initial control state of
  $S_k$.  Hence $S_k$ makes a transition from its current control state to
  (its copy of) $M$'s initial control state so in effect reinitialises itself.
  Our construction of $R$ ensures that some other semi-automaton with lower
  index which is both active and presently in (its own copy of) the same
  control state $q_i^M$ of $M$ now indeed takes over from $S_k$.  We also say
  that $S_k$ is \emphbf{inactive} and that the two semi-automata \emphbf{merge}.
\end{iteMize}
Figure~\ref{sample-behaviour-of-compound-semi-automaton-fig} gives an example
of an deterministic automaton $M$ with four states and a run of an associated
compound semi-automaton with five semi-automata $S_1,\dots,S_5$.
\bgroup
\def\clap#1{\hbox to 0pt{\hss #1\hss}}
\def\U#1{\underline{#1}}
\def\X{\text{\llap{${\scriptstyle S_2}{\leftarrow}$}}}
\begin{figure*}
  \begin{center}
  \fbox{\vbox{\hsize=0.85\textwidth\flushleft\vspace{2pt} %
    \begin{tabular}{L}
        \text{Sample formula $B$}\colon
          (\Skip\And p)\chop\Skip\chop\Skip\SChopstar
            \chop(\Empty \And \Not p)
        \\\noalign{\vspace{5pt}}
        \text{Sample automaton $M$ for $B$ (already presented
              in~\eqref{deterministic-semi-automaton-and-automaton-1-eq}):}
          \\[2pt]
        \enspace
	\begin{array}{l}
          V_M=\{p\} \text{\ (so $\Sigma_M=\{p, \Not p\}$)}
            \quad Q_M = \{ q_1, q_2, q_3, q_4 \}
            \quad q_M^I=q_1 \\[1pt]
          \delta_M(q_1,p) = q_2 \quad \delta_M(q_1,\Not p) = q_4
            \quad \delta_M(q_2,p) =\delta_M(q_2,\Not p) = q_3 \\
          \delta_M(q_3,p) =\delta_M(q_3,\Not p) = q_3
            \quad \delta_M(q_4,p) =\delta_M(q_4,\Not p) = q_4 \\[1pt]
          \tau_M(q_1)=\tau_M(q_2)=\tau_M(q_4)=\{\}
            \quad \tau_M(q_3)=\{\Not p\} \\[1pt]
	  \init_M\colon\enspace q_1 \And \Not q_2 \And \Not q_3 \And \Not q_4
            \quad \acc_M\colon\enspace q_3 \And \Not p \\[1pt]
          \mbox{\rlap{$T_M\colon$}\hphantom{$\init_M\colon$\enspace}}
	    (\Next q_1)\equiv \False
            \,\Andd\, (\Next q_2)\equiv (q_1 \And p) \\
           \mbox{\hphantom{$\init_M\colon$\enspace}}
	      \Andd\, (\Next q_3)\equiv (q_2 \Or q_3)
              \,\Andd\, (\Next q_4)\equiv ((q_1 \And \Not p) \Or q_4)
          \end{array}
        \end{tabular} \\
    \vspace{7pt}
    \mbox{\quad Control state behaviour of each $S_k$
        in sample 8-state interval $\sigma$:}
    \vspace{3pt}
    \mbox{\qquad
    \begin{tabular}{CCCCCCC}
      \text{State in $\sigma$} & \text{$p$'s value}
         & S_1 & S_2 & S_3 & S_4 & S_5\\\hline\noalign{\vspace{2pt}}
      \sigma_0 & \Not p & \bf 1 & 1 & 1 & 1 & 1 \\
      \sigma_1 &      p & \bf 4 & \bf \U1 & 1 & 1 & 1 \\
      \sigma_2 & \Not p & \bf 4 & \bf \U2 & \bf 1 & 1 & 1 \\
      \sigma_3 &      p & \bf 4 & \bf \U3 & 4 & \bf \U1 & 1 \\
      \sigma_4 &      p & \bf 4 & \bf \U3 & \bf \U1 & \bf \U2 & 1 \\
      \sigma_5 & \Not p & \bf 4 & \bf \U3 & \bf \U2 & \X3 & \bf 1 \\
      \sigma_6 &      p & \bf 4 & \bf \U3 & \X3 & \bf 1 & \bf 4 \\
      \sigma_7 & \Not p & \bf 4 & \bf \U3 & \bf 1 & \bf 2 & 1 \\\hline
        \noalign{\vspace{7pt}}
      \multicolumn{7}{@{}c}{Value of $\acc'_k$ for each $S_k$ at end
              in state $\sigma_7$:}
        \\\noalign{\vspace{2pt}}
        & & \False & \True & \False & \False & \False
    \end{tabular}}
    \vspace{5pt}
    \begin{tabular}{l}
      Some explanations about the sample 8-state interval
	$\sigma_0\ldots\sigma_7$: \\
      \quad Only control states' indices are shown (e.g., 1 for $q_1$). \\
      \quad Active semi-automata are shown in \textbf{boldface}. \\
      \quad All control states used in any accepting runs of $M$ are
          $\underline{\text{underlined}}$. \\
      \quad ``${\scriptstyle S_2}{\leftarrow}$''
          shows merge into semi-automaton $S_2$ in accepting run for $M$.
    \end{tabular}
    \vspace{5pt}
    \begin{tabular}[t]{l}
      Compound accepting runs of $M$ to recognise $B$:
        \\[2pt]
      \quad Suffix subinterval $\sigma_1\ldots\sigma_7$
        ($S_2$: $\sigma_1\sigma_2\sigma_3\sigma_4\sigma_5\sigma_6\sigma_7$):
          $\underbrace{q_1,q_2,q_3,q_3,q_3,q_3,q_3}_{S_2}$ \\
      \quad Suffix subinterval $\sigma_3\ldots\sigma_7$
        ($S_4$: $\sigma_3\sigma_4$, $S_2$: $\sigma_5\sigma_6\sigma_7$):
          $\underbrace{q_1,q_2}_{S_4},\underbrace{q_3,q_3,q_3}_{S_2}$ \\
      \quad Suffix subinterval $\sigma_4\ldots\sigma_7$
        ($S_3$: $\sigma_4\sigma_5$, $S_2$: $\sigma_6\sigma_7$):
          $\underbrace{q_1,q_2}_{S_3},\underbrace{q_3,q_3}_{S_2}$
    \end{tabular}
    \vspace{2pt}}}
  \end{center}
  \caption{Sample behaviour of compound semi-automaton in 8-state interval}
  \label{sample-behaviour-of-compound-semi-automaton-fig}      
\end{figure*}
\egroup

Recall that our representation of $M$'s $n$ control states using $n$
propositional variables $q^M_1,\ldots, q^M_n$ has exactly one of the variables
being true at any time.  Hence we represent the $n$ control states for a
semi-automata $S_k$ using $n$ propositional variables $q^k_1, \ldots, q^k_n$.
Therefore the subset of atoms in $\Sigma_{S_k}$ extracted from $R$'s composite
runs always have exactly one variable $q^j_i$ true for each semi-automaton
$S_j$ with $j\lt k$.  This property of the runs follows by induction on
$k$. In contrast, the full set of atoms for $\Sigma_{S_k}$ includes for each
index $j$ with $j\lt k$ some \emphbf{pathological atoms} in which none or more
than one of the $q^j_i$ are true.  Nevertheless, actual runs of $S_k$ in $R$
never encounter such atoms so we need not concern ourselves with the precise
way $\delta_{S_k}$ is defined to handle them in transitions.

\subsection{Formalisation of Suffix Recognition in $\PITL$}

The following lemma formalises the finite-time behaviour of the compound
semi-automaton $R$ in $\PITL$ and uses an associated chain formula $G_R$ in
$\PTLU$ which we construct in the proof:
\begin{mylemma}
  \label{ata-and-df-chain-formula-lem}
  For any $\PITL$ formula $B$, there exists a compound semi-automaton $R$ with
  $V_R=V_B$ and associated ATA $D_R$ and chain formula $G_R$ such that $R$'s
  control variables are not in $B$ and the next implication is valid:
  \begin{equation}
    \label{ata-and-df-chain-formula-1-eq}
    \Valid \Finite \And \Bf D_R \Implies \Box (B \equiv G_R)
%% \qquad \text{(also $\ldots\imp\text{\emph{$\Bf$}}\Box\ldots$)}
\dotspace.
  \end{equation}
\end{mylemma}
This lemma provides a way to replace right-instances of a $\PITL$ formula $B$
by a chain formula $G_R$ in formulas restricted to finite time.  However, it
serves as basis for later replacing \emph{lefthand} sides of chops with chain
formulas.  The lemma is entirely semantic and so does not depend on any
particular axiom system or deductions. We will later readily deduce the
lemma's implication~\eqref{ata-and-df-chain-formula-1-eq} by invoking the
completeness for $\PITL$ with finite time to obtain immediate theoremhood of
the implication and some valid variants of it.  Hence, from the standpoint of
axiom systems and deductions, there is no need to know
Lemma~\ref{ata-and-df-chain-formula-lem}'s proof or even any further details
of $R$, $D_R$ and $G_R$.

\begin{proof}[Proof of Lemma~\ref{ata-and-df-chain-formula-lem}.]
  The construction for $R$ ensures that the set union $Q_{S_1}\union \cdots
  \union Q_{S_{n+1}}$ of control variables of the semi-automata $S_1$, \ldots,
  $S_{n+1}$ contains no elements of the set $V_B$ of propositional variables
  occurring in $B$.

  We will obtain the chain formula $G_R$ by mimicking an accepting run of $M$.
  This involves combining together pieces of runs from the some of the
  semi-automata $S_1$, \ldots, $S_{n+1}$.  It needs at most $n$ merges since
  when two semi-automata merge, only the one of lesser index continues testing.
  The chain formula $G_R$, when suitably combined with the compound
  semi-automaton $R$'s ATA, will capture the needed behaviour which we
  previously formalised in the
  implication~\eqref{ata-and-df-chain-formula-1-eq}.

  We first define state formulas to test for active and merging semi-automata
  and also introduce a modified acceptance test:
  \begin{iteMize}{$\bullet$}
  \item \emphbf{$\boldsymbol{\mathit{active}_k}$:} True iff semi-automaton
    $S_k$ is \emph{active}.
      \begin{equation*}
        \mathit{active}_k
          \Defeqv \bigvee_{1\le i\le n+1}
            \bigl(q^{S_k}_i \And \bigwedge_{1\le j\lt k}\Not q^{S_j}_i\bigr)
\dotspace.
      \end{equation*}
%       An active semi-automaton's transition to its next control state will be
%       based on the analogous transition for $M$.

    \item \emphbf{$\boldsymbol{\mathit{merge}_{j,k}}$:} True iff the active
      semi-automaton $S_j$ and inactive semi-automaton $S_k$ merge.
      \begin{equation*}
        \mathit{merge}_{j,k}
          \Defeqv \bigvee_{1\le i\le n+1}
		    (q^{S_j}_i \And q^{S_k}_i
		       \And \mathit{active}_j \And \Not\mathit{active}_k)
\dotspace.
      \end{equation*}
      It follows from the definition of an active semi-automaton that $j\lt
      k$.
%     The transition function for $S_k$ will re-initialise it to be available
%     for possibly examining whether some subsequent suffix subinterval
%     satisfies the $\PITL$ formula $B$.
  \item \emphbf{$\boldsymbol{\mathit{acc}'_k}$:} Let
    us also define a propositional test $\mathit{acc}'_k$ based on the state
    formula $\acc_M$ for checking $M$'s conditional acceptance test $\tau_M$.
    We use a substitution instance of $\acc_M$ to adapt it to $S_k$ and its
    own copies of $M$'s control states.
    \begin{equation*}
      \acc'_k  \Defeqv  (\acc_M)_{q^M_1,\ldots, q^M_n}^{q^{S_k}_1,\ldots, q^{S_k}_n}
  \dotspace.
    \end{equation*}
    Note that a semi-automaton $S$ has no conditional acceptance test $\tau_S$
    and indeed the role of $\mathit{acc}'_k$ here somewhat differs from that of
    $\acc_M$.
  \end{iteMize}
  As usual, for an individual semi-automaton $S_k$ in the compound
  semi-automaton $R$, the state formula $\init_{S_k}$ tests for the initial
  control state of $S_k$ and the $\NLone$ formula $T_{S_k}$ expresses the
  transition function $\delta_{S_k}$ of $S_k$ in temporal logic.
  
  Let us now inductively define for each pair $j,k: 1\le j\le k\le n+1$ a
  chain formula $G'_{k,j}$ to be true iff a run segment starts with currently
  active semi-automaton $S_k$ in some unspecified control state, involves
  exactly $j$ active automata (i.e., $j-1$ mergers) and ends with acceptance
  of the word seen.
  \begin{displaymath}
    \begin{array}{ll}
      G'_{k,1}\colon
        & (\mathit{active}_k \And T_{S_k})\Until (\mathit{acc}'_k \And \Empty) \\
      G'_{k,j+1}\colon
         & \displaystyle
	   (\mathit{active}_k \And T_{S_k}) \Until
             \bigvee_{1\le i\lt k}
               \bigl(\mathit{merge}_{i,k} \And G'_{i,j}\bigr)
\dotspace.
  \end{array}
  \end{displaymath}
  For example, the chain formula $\init_{S_1} \And \mathit{active}_1 \And
  G'_{1,1}$ corresponds to an accepting run of $M$ in which the semi-automaton
  $S_1$ recognises $B$ on its own.  The conjunction $\init_{S_2} \And
  \mathit{active}_2 \And G'_{2,2}$ corresponds to an accepting run of $M$
  involving first semi-automaton $S_2$ and then semi-automaton $S_1$.  The
  semi-automaton $S_2$ starts recognising $B$ and eventually merges with
  semi-automaton $S_1$ which completes the accepting run.

  Now let us construct from the chain formulas $G'_{k,j}$ the chain formula
  $G_R$ specifying an accepting run involving some of the $n+1$ semi-automata
  to recognise the $\PITL$ formula $B$. Like in the examples, we start in some
  active copy of $M$'s initial control state:
  \begin{displaymath}
       \textstyle
       G_R\colon
	\quad \bigvee_{1\le k\le n+1}
		\bigl(\init_{S_k} \And \mathit{active}_k
		  \And \bigvee_{1\le j\le k} G'_{k,j}\bigr)
\dotspace.
  \end{displaymath}
  The construction of the compound semi-automaton $R$ together with $D_R$ and
  $G_R$ ensures the desired validity of
  implication~\eqref{ata-and-df-chain-formula-1-eq}.
\end{proof}

%% \textbf{\Large ***** Resume from here.   3 May '11 ***** \\}

To assist readers, we list in Table~\ref{kinds-of-variables-table} a variety
of variables and where they are introduced.
\bgroup
\def\mySection{\S}
\begin{table*}%[h]
  \begin{center}
    \begin{tabular}{Ll@{}c}
       \text{Variable names} & \multicolumn{1}{c}{Category}
	  & Where defined \\\hline\noalign{\vspace{2pt}}
       A, A', B, C & Arbitrary $\PITL$ formulas
	 & \mySection\ref{pitl-sec} \\
       \alpha, \beta & Atoms (letters)
	 & \mySection\ref{deterministic-semi-automata-and-automata-sec} \\
       \acc_M & State formula for automaton $M$'s acceptance 
	 & \mySection\ref{deterministic-finite-state-automata-subsec} \\
       D, D' & Auxiliary temporal assignments (ATA)
	 & \S\ref{auxiliary-temporal-assignments-subsec} \\
       D_S, D_M, D_R & ATA for use in expressing runs
	   of $S$, $M$ and $R$
	 & \S\ref{atas-for-semi-automata-and-automata-subsec},
	   \S\ref{compound-semi-automata-for-suffix-recognition-sec} \\
       \delta_S, \delta_M & Deterministic transition function
	 & \S\ref{deterministic-finite-state-semi-automata-subsec},
	   \S\ref{deterministic-finite-state-automata-subsec} \\
        & for semi-automaton $S$ and automaton $M$ & \\
       G, G' & Chain formulas
	 & \S\ref{right-chops-and-chain-formulas-subsec} \\
       \init_S, \init_M & State formula to force the initial control state
	 & \S\ref{deterministic-finite-state-semi-automata-subsec},
	   \S\ref{deterministic-finite-state-automata-subsec} \\
        & of semi-automaton $S$ and automaton $M$ \\
       K, K' & $\PITLK$ formulas
	 & \S\ref{pitl-without-omega-iteration-subsec} \\
       M & Deterministic finite-state automaton
	 & \S\ref{deterministic-finite-state-automata-subsec} \\
       p, p', q, r & Propositional variables
	 & \mySection\ref{pitl-sec} \\
       Q_S, Q_M & Sets of control states of semi-automaton $S$
	 & \S\ref{deterministic-finite-state-semi-automata-subsec},
	   \S\ref{deterministic-finite-state-automata-subsec} \\
        & and automaton $M$ & \\
       R & Compound finite-state semi-automaton
	 & \mySection\ref{compound-semi-automata-for-suffix-recognition-sec} \\
       S & Deterministic finite-state semi-automaton
	 & \S\ref{deterministic-finite-state-semi-automata-subsec} \\
       \Sigma_{\VV}
	 & Atoms (letters) formed from variables in set $V$
	 & \mySection\ref{deterministic-semi-automata-and-automata-sec} \\
       \Sigma_S, \Sigma_M
	 & Atoms tested by semi-automaton $S$ and automaton $M$
	 & \S\ref{deterministic-finite-state-semi-automata-subsec},
	   \S\ref{deterministic-finite-state-automata-subsec} \\
       T, T' & $\NLone$ formulas & \S\ref{nl-one-formulas-subsec} \\
       T_S, T_M & $\NLone$ formula for transitions
	    of semi-automaton $S$
	 & \S\ref{deterministic-finite-state-semi-automata-subsec},
	   \S\ref{deterministic-finite-state-automata-subsec} \\
	 & and automaton $M$ & \\
       \tau_M & Conditional acceptance test for automaton $M$
	 & \S\ref{deterministic-finite-state-automata-subsec} \\
       V & Finite set of propositional variables
	 & \mySection\ref{pitl-sec} \\
       V_A, V_S, V_M, V_R
	 & Finite set of propositional variables in $\PITL$
	 & \S\ref{pitl-sec},
	   \S\ref{deterministic-finite-state-semi-automata-subsec},
	   \S\ref{deterministic-finite-state-automata-subsec},
	   \mySection\ref{compound-semi-automata-for-suffix-recognition-sec} \\
	 & formula $A$ and in atoms of semi-automaton $S$, \\
         & automaton $M$ and compound semi-automaton $R$ \\
       w, w' & State formulas & \S\ref{pitl-sec} \\
       X, X' & $\PTL$ formulas & \S\ref{pitl-sec} \\
       Y, Y' & $\PTLU$ formulas & \S\ref{ptl-with-until-subsec}
    \end{tabular}
    \caption{Naming conventions for different variables}
    \label{kinds-of-variables-table}
  \end{center}
\end{table*}
\egroup

\section{Reduction of PITL to PTL with Until}

\label{reduction-of-pitl-to-ptlu-sec}

Most of the remaining part of the $\PITL$ completeness proof concerns using
compound semi-automata to show right-completeness for $\PITLK$ by reduction to
$\PTLU$.  Recall from \S\ref{right-chops-and-chain-formulas-subsec} that any
chop construct in a formula $A$ is a right-chop iff it does not occur in
another chop's left operand or in a chop-star.

The $\PITL$ theorems mentioned here in proofs are found in
Table~\ref{list-of-pitl-theorems-and-derived-rules-mentioned-before-app-table}
in \S\ref{formal-equivalence-of-the-two-representations-of-runs-subsec} and
also Appendix~\ref{some-pitl-theorems-and-their-proofs-sec}.

\subsection{Application of Suffix Recognition, Right-Chops and Chain Formulas}

\label{application-of-suffix-recognition-right-chops-and-chain-formulas-subsec}

The next Lemma~\ref{infix-recognition-lem}, which employs the compound
semi-automaton $R$, generalises suffix recognition to \emphbf{infix
  recognition} for checking which of a (possibly infinite-time) interval's
finite-time infix subintervals satisfy some given $\PITL$ formula by instead
using a chain formula.
\begin{mylemma}
  \label{infix-recognition-lem}
  For any $\PITL$ formula $B$, there exists a compound semi-automaton $R$ with
  $V_R=V_B$, associated ATA $D_R$ and chain formula $G_R$ such that $R$'s
  control variables are not in $B$ and the next formula is a $\PITL$ theorem:
  \begin{equation}
    \label{infix-recognition-1-eq}
    \Theorem \Bf D_R \Implies \Box\Bf(B \equiv G_R)
\dotspace.
  \end{equation}
\end{mylemma}
\begin{proof}
  Lemma~\ref{ata-and-df-chain-formula-lem} ensures the validity of the
  implication below for some compound semi-automaton $R$, associated ATA $D_R$
  and chain formula $G_R$:
  \begin{displaymath}
    \Valid \Finite \And \Bf D_R \Implies \Box(B \equiv G_R)
%%    \qquad \text{(also $\ldots\imp\text{\emph{$\Bf$}}\Box\ldots$)}
\dotspace.
  \end{displaymath}
  This and completeness for $\PITL$ with finite time
  (Theorem~\ref{completeness-of-pitl-axiom-system-for-finite-time-thm})
  ensures the next implication's theoremhood:
  \begin{displaymath}
    \Theorem \Finite
	\Implies \bigl(\Bf D_R \,\implies\, \Box(B \equiv G_R)\bigr)
%%    \qquad \text{(also $\ldots\imp\text{\emph{$\Bf$}}\Box\ldots$)}
\dotspace.
  \end{displaymath}
  This and Inference Rule~\refid{BfFGen} yield the next formula:
  \begin{displaymath}
    \Theorem \Bf\bigl(\Bf D_R \implies \Box(B \equiv G_R)\bigr)
\dotspace.
  \end{displaymath}
  Simple reasoning about $\Bf$ (see $\PITL$ Theorem~\refid{BfImpDist}) results
  in the following:
  \begin{displaymath}
    \Theorem \Bf\Bf D_R \Implies \Bf\Box(B \equiv G_R)
\dotspace.
  \end{displaymath}
  We re-express $\Bf\Bf D_R$ as $\Bf D_R$ and commute $\Bf\Box$ (see $\PITL$
  Theorems~\refid{BfBfEqvBf} and~\refid{BfBoxEqvBoxBf}) to obtain our
  goal~\eqref{infix-recognition-1-eq}.
\end{proof}

The lemma below later plays a key role in reducing right-chops in a $\PITLK$
formula to $\PTLU$ formulas by first replacing their left sides with chain
formulas in $\PTLU$:
\begin{mylemma}
  \label{ata-and-ptlu-formula-lem}
  For any $\PITL$ formulas $B$ and $C$, there exists a compound semi-automaton
  $R$ with $V_R=V_B$, associated ATA $D_R$ and chain formula $G_R$ such that
  $R$'s control variables are not in $B$ or $C$ and the next formula is
  deducible as a right-theorem:
  \begin{equation}
    \label{ata-and-ptlu-formula-1-eq}
    \TheoremR \Bf D_R \Implies \Box\bigl((B\chop C) \equiv (G_R\chop C)\bigr)
    \dotspace.
  \end{equation}
\end{mylemma}
\begin{proof}
  Lemma~\ref{infix-recognition-lem} yields $R$, $D_R$, $G_R$ and the next
  implication for infix recognition of $B$:
  \begin{equation}
    \label{ata-and-ptlu-formula-2-eq}
%    \label{infix-recognition-1-eq}
    \Theorem \Bf D_R \Implies \Box\Bf(B \equiv G_R)
\dotspace.
  \end{equation}
  Note that this has no right variables.  We also employ the next implication
  which is an instance of $\PITL$ Theorem~\refid{BfChopEqvChop} and concerns
  interval-based reasoning about the left of chop:
  \begin{equation}
    \label{ata-and-ptlu-formula-3-eq}
    \TheoremR \Bf(B \equiv G_R) \Implies (B\chop C) \equiv (G_R\chop C)
\dotspace.
  \end{equation}
  Inference Rule~\refid{BoxGen} then obtains from
  implication~\eqref{ata-and-ptlu-formula-3-eq} the formula below:
  \begin{equation*}
    \TheoremR
      \Box\bigl(\Bf(B \equiv G_R)
	 \,\implies\, (B\chop C) \equiv (G_R\chop C)\bigr)
\dotspace.
  \end{equation*}
  This with $\PTL$-based reasoning involving the valid $\PTL$ formula $\Box(p
  \imp q) \imp \bigr((\Box p) \imp (\Box q)\bigr)$ with Axiom~\refid{VPTL},
  where $p$ is replaced by $\Bf(B \equiv G_R)$ and $q$ by $(B\chop C) \equiv
  (G_R\chop C)$, together with modus ponens results in the following:
  \begin{equation}
    \label{ata-and-ptlu-formula-4-eq}
    \TheoremR \Box\Bf(B \equiv G_R)
      \Implies \Box\bigl((B\chop C) \equiv (G_R\chop C)\bigr)
\dotspace.
  \end{equation}
  Implications~\eqref{ata-and-ptlu-formula-2-eq}
  and~\eqref{ata-and-ptlu-formula-4-eq} and simple propositional reasoning
  yield our goal~\eqref{ata-and-ptlu-formula-1-eq}.
\end{proof}

\begin{mylemma}
  \label{re-expressing-pitlk-with-chain-formulas-in-ptlu-lem}
  Any $\PITLK$ formula $K$ in which the left sides of all right chops are
  chain formulas is deducibly equivalent to some $\PTLU$ formula $Y$, that is,
  $\thmR K\equiv Y$.
\end{mylemma}
\begin{proof}
  Starting with $K$'s right-chops not nested in other right-chops, we
  inductively replace them by equivalent $\PTLU$ formulas. More precisely, if
  $n$ is the number of $K$'s right chops, then we use $n$ applications of
  Lemma~\ref{re-expressing-g-chop-x-in-ptlu-lem} and the Right Replacement
  Rule (Lemma~\ref{right-replacement-derived-rule-lemma-for-pitl-lem}) to
  show that $K$ is deducibly equivalent to some $\PTLU$ formula $Y$ (i.e.,
  $\thmR K \equiv Y$).
\end{proof}

%   This works since since any chop-stars in $K$ are very restricted ones of the
%   form $(\Skip \And T)\SChopstar$ found in the $\PITL$-based definition of
%   $\UntilOp$ given earlier in \S\ref{ptl-with-until-subsec} and itself used in
%   the definition of chain formulas in
%   \S\ref{right-chops-and-chain-formulas-subsec}.

For example, suppose $K$ is $(G_1\chop \Skip) \Or \bigl(G_2 \chop (G_3 \chop
w)\bigr)$ and hence has 3 right-chops.  We could start by first re-expressing
either $G_1\chop \Skip$ or $G_3\chop w$ by an equivalent $\PTLU$ formula.  For
instance, if $G_2$ is the chain formula $p \Until \Empty$ and $G_3$ is the
chain formula $q\Until \Empty$, then $G_3\chop w$ will be replaced by the
equivalent $\PTLU$ formula $q\Until w$. After this, $G_2 \chop (G_3 \chop w)$
will first reduce to $G_2 \chop (q\Until w)$ and finally to the $\PTLU$
formula $p \Until (q\Until w)$.

\subsection{Proof of the Main Completeness Theorem}

\label{proof-of-the-main-completeness-theorem-subsec}

We now establish right-completeness for $\PITLK$ and then use this to obtain
right-completeness for $\PITL$.
\begin{mylemma}
  \label{completeness-for-pitlk-lem}
  Any valid $\PITLK$ formula can be deduced as a right-theorem.
\end{mylemma}
\begin{proof}
  We show that a right-consistent $\PITLK$ formula $K$ is satisfiable.  Our
  proof transforms $K$ to a $\PTLU$ formula. Let $m$ equal the number of $K$'s
  right-chops. We employ $m$ compound semi-automata to obtain ATAs for
  systematically replacing the left operands of $K$'s right-chops by $\PTLU$
  chain formulas.  Note that if $m=0$, then $K$ has no chops but perhaps
  $\Skip$ so $K$ itself is in $\PTL$.  We will construct a sequence of $m+1$
  $\PITLK$ formulas $K'_1$, \ldots, $K'_{m+1}$.  In the final one $K'_{m+1}$,
  left operands of all right-chops are chain formulas so $K'_{m+1}$ is
  deducibly equivalent to some $\PTLU$ formula by
  Lemma~\ref{re-expressing-pitlk-with-chain-formulas-in-ptlu-lem}.  For
  example, suppose $K$ has the form $(B_1\chop w) \imp \bigl(B_2 \chop (B_3
  \chop \Skip)\bigr)$.  Then $K$ has 3 right-chops so $m$ equals 3 and $K'_4$
  has the form $(G_1\chop w) \imp \bigl(G_2 \chop (G_3 \chop \Skip)\bigr)$,
  where $G_1$, $G_2$ and $G_3$ in $K'_4$'s 3 right-chops' left sides are all
  chain formulas.

  Let $K'_1$ be $K$.  For each $i\colon 1\le i\le m$, we choose a right-chop
  in $K'_i$.  This has the form $B_i\chop K''_i$ for some $\PITL$ formula
  $B_i$ and $\PITLK$ formula $K''_i$.  Lemma~\ref{ata-and-ptlu-formula-lem}
  yields a compound semi-automaton $R'_i$, ATA $D_{R'_i}$ and a chain formula
  $G_{R'_i}$ for which the next right-theorem is deducible:
  \begin{equation}
    \label{completeness-of-pitl-axiom-system-1-eq}
    \TheoremR \Bf D_{R'_i}
      \Implies \Box\bigl((B_i\chop K''_i)\equiv (G_{R'_i}\chop K''_i)\bigr)
\dotspace.
  \end{equation}
  We employ Lemma~\ref{basic-limited-replacement-lemma-for-pitl-lem}
  concerning replacement of right-instances to relate $K'_i$ and $K'_{i+1}$ by
  replacing the selected $B_i\chop K''_i$ by $G_{R'_i}\chop K''$:
  \begin{equation*}
    \TheoremR \Box\bigl((B_i\chop K''_i)\equiv (G_{R'_i}\chop K''_i)\bigr)
      \Implies K'_i \equiv K'_{i+1}
\dotspace.
  \end{equation*}
  This and implication~\eqref{completeness-of-pitl-axiom-system-1-eq} together
  ensure the right-theorem $\thmR \Bf D_{R'_i} \imp (K'_i\equiv K'_{i+1})$.
  Without loss of generality, assume the control variables in the compound
  semi-automata $R'_1,\ldots, R'_m$ are distinct.  We deduce from the $m$
  implications $\thmR \Bf D_{R'_i} \imp (K'_i\equiv K'_{i+1})$ just mentioned
  the next right-theorem:
  \begin{equation}
    \label{completeness-of-pitl-axiom-system-2-eq}
    \textstyle
    \TheoremR \bigwedge_{1\le i\le m}(\Bf D_{R'_i}) \Implies  K\equiv K'_{m+1}
\dotspace.
  \end{equation}
  The left operand of each right-chop in $K'_{m+1}$ is a chain formula.  Hence
  by Lemma~\ref{re-expressing-pitlk-with-chain-formulas-in-ptlu-lem}, we can
  deduce the equivalence of $K'_{m+1}$ and some $\PTLU$ formula $Y$ to obtain
  the $\PITL$ right-theorem $\thmR K'_{m+1}\equiv Y$.  By this and
  implication~\eqref{completeness-of-pitl-axiom-system-2-eq}, the next
  implication is a right-theorem:
  \begin{equation}
    \label{completeness-of-pitl-axiom-system-3-eq}
    \textstyle
    \TheoremR \bigwedge_{1\le i\le m}(\Bf D_{R'_i}) \Implies  K\equiv Y
    \dotspace.
  \end{equation}
  Right-variables in the original formula $K$ do not occur in any $D_{R'_i}$
  since the construction of each $D_{R'_i}$ only involves the left sides of
  $K$'s right-chops.  The right-variables in $K$ are still right-variables in
  $Y$ and implication~\eqref{completeness-of-pitl-axiom-system-3-eq}.  Now
  $K$'s right-consistency and $m$ applications of
  Lemma~\ref{temporal-operator-bf-atas-and-consistency-lem} ensure the
  right-consistency of $K \And \bigwedge_{1\le i\le m}(\Bf D_{R'_i})$.  This
  is re-expressible as $K \And \Bf D'$, where the ATA $D'$ is the conjunction
  of the ATAs $D_{R'_1},\ldots,D_{R'_m}$ (we use $\PITL$
  Theorem~\refid{BfAndEqv}).  Hence the formula $K\And \Bf D'$ is
  right-consistent.  We deduce the equivalence of $\Bf D'$ and some $\PTL$
  formula $X$ as $\thm X\equiv \Bf D'$ by invoking
  Lemma~\ref{deductive-eqv-of-two-reps-of-auto-runs-lem} on the individual
  basic semi-automata in each $R'_i$ to re-express each one's runs in $\PTL$
  and then forming the conjunction of results.  Now $D'$ and $X$ have the same
  variables.  Hence the equivalence $X\equiv \Bf D'$ has no right-variables
  because of $\Bf D'$ and is a right-theorem (i.e., $\thmR X\equiv \Bf D'$).
  This with the equivalence $\thmR \Bf D'\equiv\bigwedge_{1\le i\le m}(\Bf
  D_{R'_i})$ and implication~\eqref{completeness-of-pitl-axiom-system-3-eq}
  then yield the equivalence of formulas $K \And \Bf D'$ and $Y \And X$ as a
  right-theorem.  Therefore the $\PTLU$ formula $Y \And X$, like $K \And \Bf
  D'$, is right-consistent and by right-completeness for $\PTLU$ (discussed in
  \S\ref{ptl-with-until-subsec}) is satisfiable as is $K$.
\end{proof}

We now prove our main result
Theorem~\ref{completeness-of-pitl-axiom-system-thm} about right-completeness
for $\PITL$:
\begin{proof}[Proof of Theorem~\ref{completeness-of-pitl-axiom-system-thm}.]
  Let $A$ be a right-consistent $\PITL$ formula.
  Lemma~\ref{completeness-for-pitlk-lem} ensures right-completeness for
  $\PITLK$. Hence by this and Lemma~\ref{reduction-of-pitl-pitlk-lem}, there
  exists some $\PITLK$ formula $K$ having the same variables and
  right-variables as $A$ and with the deducible equivalence $\thmR A\equiv K$.
  Now $K$ like $A$ is right-consistent and so satisfiable by
  right-completeness for $\PITLK$ (Lemma~\ref{completeness-for-pitlk-lem}).
  Hence $A$ is satisfiable.
\end{proof}
As we already remarked in
Section~\ref{right-instances-right-variables-and-right-theorems-sec}, the
completeness proof can be regarded as two parallel proofs.  The simpler one
uses the extra inference
rule~\eqref{optional-inference-rule-for-right-variables-1-eq} mentioned there
to avoid right-theorems and right-completeness.  The more sophisticated proof
uses right-theoremhood instead of the inference rule and ensures that any
valid $\PITL$ formula is not just a theorem but a right-theorem.

 \vspace{5pt}

 This concludes the $\PITL$ completeness proof.

\section{Some Observations about the Completeness Proof}

\label{some-observations-about-the-completeness-proof-sec}

We now consider various issues concerning the new $\PITL$ axiom system and
techniques employed in the completeness proof.  Most of the points address
questions previously raised by others.

\subsection{Alternative Axioms for PTL}

\label{alternative-axioms-for-ptl-subsec}

Axiom~\refid{VPTL} in
Table~\ref{axiom-system-for-pitl-with-finite-and-infinite-time-table} can
optionally be replaced by four lower level axioms.  Readers may wish to skip
over the details now given.  One of the lower level axioms is~\refid{Taut} in
Table~\ref{pitl-axiom-system-for-finite-time-table} permitting $\PITL$
formulas which are substitution instances of conventional (nonmodal)
tautologies.  For example, from the valid propositional formula $p \imp (p \Or
q)$ follows $\thm A\imp (A \Or B)$, for any $\PITL$ formulas $A$ and $B$.  The
other three axioms involve $\PTL$.  These are Axioms~\refid{FNextImpWeakNext}
and~\refid{FBoxInduct} found in
Table~\ref{pitl-axiom-system-for-finite-time-table} and also $\thm
\Skip\imp\Finite$.  The three Axioms~\refid{Taut}, \refid{FNextImpWeakNext}
and~\refid{FBoxInduct} together with the remaining $\PITL$ axioms and
inference rules in
Table~\ref{axiom-system-for-pitl-with-finite-and-infinite-time-table} then
suffice to derive a slight variant proposed by us in~\cite{Moszkowski04a} of
the complete $\PTL$ axiom system $\mathit{D^0\!X}$ for $\Next$ and $\Diamond$
(and $\Box$) of Gabbay et al.~\citeyear{GabbayPnueli80}, itself based on an
earlier one $\mathit{DX}$ of Pnueli~\citeyear{Pnueli77}.  We denote our
$\mathit{D^0\!X}$ variant here as $\mathit{D^0\!X'}$.  It permits both finite
and infinite time, whereas $\mathit{D^0\!X}$ assumes infinite time.  We
previously did an explicit deduction of $\mathit{D^0\!X'}$ in our completeness
proof for $\PITL$ with just finite time as described in~\cite{Moszkowski04a}.
However, for infinite time we need the additional axiom $\thm \Skip \imp
\Finite$ because Axiom~\refid{FiniteImpChopEmpty} (unlike
Axiom~\refid{FChopEmpty} in
Table~\ref{pitl-axiom-system-for-finite-time-table}) does not suffice on its
own to deduce $\thm \Skip\equiv\Next\Empty$ to re-express $\Skip$ using
$\Next$. Without $\thm \Skip \imp \Finite$, we can only deduce the $\PITL$
theorem $\thm \Finite \imp(\Skip\equiv\Next\Empty)$ from
Axiom~\refid{FiniteImpChopEmpty} together with the definition of $\Next$ in
terms of $\Skip$ and chop.  In addition, from $\mathit{D^0\!X'}$ (once
deduced), we can obtain $\thm (\Next\Empty) \imp \Finite$.  These two
implications combined with $\thm \Skip \imp \Finite$ and simple propositional
reasoning (involving Axiom~\refid{Taut} and modus ponens) yield our goal $\thm
\Skip\equiv \Next\Empty$.
% Here are suitable ones:
% \begin{align*}
%   & \theoremBEN \Next A \implies \Not\Next\Not A \\
%   & \theoremBEN A \And \Box(A \And \More \imp \Next A) \implies \Box A
% \dotspace.
% \end{align*}

\subsection{Feasibility of Reduction from PITL to PTL}

\label{feasibility-of-reduction-from-pitl-to-ptl-subsec}

Some people have expressed serious doubts about our proof's technical
feasibility owing to the significant gap in expressiveness between $\PITL$ and
$\PTL$.  We therefore believe it is worthwhile to emphasis that in spite of
this gap, any $\PITL$ formula can be represented by some $\PTL$ formula
containing auxiliary variables.  This is because conventional semantic
reasoning about omega-regular languages and omega automata ensures that for
any $\PITL$ formula $A$, there exist conventional nondeterministic omega
automata (such as B\"uchi automata) which recognise $A$.  For example, we
present in~\cite{Moszkowski00} a decidable version of quantified $\ITL$ which
includes $\QPITL$ (defined earlier in Section~\ref{pitl-sec}) as a subset and
then show how to encode formulas in B\"uchi automata.  Various deterministic
omega automata (e.g., with Muller, Rabin and Streett acceptance conditions)
are also suitable for this.  Such an automaton's accepting runs can be
trivially encoded by some $\PTL$ formula $X$ with auxiliary variables $p_1$,
$\ldots$, $p_n$ representing the automaton's control state.  Hence the $\PITL$
formula $A$ and the $\QPTL$ formula $\Exists{p_1\ldots p_n}X$ are semantically
equivalent, where $\exists$ is defined earlier in Section~\ref{pitl-sec}.
Furthermore, the (quantifier-free) $\PITL$ implication $X\imp A$ is valid and
consequently any model of $X$ can also serve as one for $A$.  Indeed the
technique of re-expressing formulas in omega-regular logics by means of
nondeterministic and deterministic omega automata expressed in versions of
$\PTL$ (subsequently enclosed in a simple sequence of existential quantifiers)
is central to the completeness proofs for $\QPTL$ variants by Kesten and
Pnueli~\citeyear{KestenPnueli2002} and French and
Reynolds~\citeyear{FrenchReynolds03}.  A related approach can be used to
reduce decidability of $\PTL$ with the (full) $\Until$ operator to $\PTL$
without $\Until$.  This works in spite of the fact that $\PTL$ with $\Until$
is strictly more expressive as proved by Kamp~\citeyear{Kamp68} (see also
Kr{\"o}ger and Merz~\citeyear{KroegerMerz08}).  We replace each $\Until$ in a
formula with an auxiliary variable which mimics its behaviour along the lines
of the two axioms for $\Until$ previously mentioned in
\S\ref{ptl-with-until-subsec}.  For example, when testing the satisfiability
of the formula $p \And \Next(p\Until q) \And \Not(p\Until q)$, we transform it
into the formula below with an extra auxiliary variable $r$:
  \begin{equation*}
    p \Andd \Next r \Andd \Not r
     \Andd \Box\bigl(r \EQUIV q \Or (p \And \Next r)\bigr)
     \Andd \Box(r \imp \Diamond q)
\dotspace.    
  \end{equation*}

\subsection{Benefits of Restricted Chop-Stars in Chain Formulas}

Lemma~\ref{completeness-for-pitlk-lem} states that any valid $\PITLK$ formula
can be deduced as a right-theorem.  Within the proof of this lemma, all
chop-star formulas found in the $\PITLK$ formula $K'_{m+1}$ only occur in
chain formulas. Such chop-star formulas therefore have the very restricted
form $(\Skip \And T)\SChopstar$ for expressing the $\PITL$-based version of
$\Until$ defined earlier in \S\ref{ptl-with-until-subsec} for $\PTLU$.  The
simplicity of these chop-star constructs greatly helps us to reduce $K'_{m+1}$
to the semantically equivalent $\PTLU$ formula $Y$ and show that their
equivalence is a deducible theorem.  Incidentally,
in~\cite{Moszkowski07\LMCSPaperBib} we prove that any $\PITL$ formula $(\Skip
\And T)\SChopstar$ can be expressed in $\PTL$ as $\Box(\More \imp T)$ and make
extensive use of this equivalence.  In contrast, arbitrary chop-star formulas
cannot necessarily be re-expressed as semantically equivalent $\PTL$ formulas.

\subsection{Thomas' Theorem and the Size of  Deductions}

Section~\ref{reduction-of-chop-omega-sec} uses Thomas' theorem to re-express a
$\PITL$ formula $A$ as a semantically equivalent $\PITLK$ formula $K$.  The
two known proofs of Thomas' theorem by Thomas himself~\citeyear{Thomas79} and
Choueka and Peleg~\citeyear{ChouekaPeleg83} unfortunately do not ensure that
$K$ is in some sense natural and succinct or even obtainable in a
computationally feasible way.  Therefore our completeness proof does not
guarantee simple deductions.  The main problem concerns the difficulties in
nontrivial transformations on the underlying omega automata representing
$\PITL$ formulas.  Other established completeness proofs for comparable
omega-regular logics with nonelementary complexity such as
$\QPTL$~\cite{KestenPnueli95,KestenPnueli2002,FrenchReynolds03} currently
share a similar fate. However, our proof bypasses an explicit embedding of the
intricate process of complementing nondeterministic omega automata.

\subsection{Justification for Using ATAs in the Completeness Proof}

Some readers will wonder why we need ATAs introduced in
\S\ref{atas-for-semi-automata-and-automata-subsec} and do not just use the
$\PTL$-based representation of semi-automata and automata presented
in~\S\ref{deterministic-finite-state-semi-automata-subsec}
and~\S\ref{deterministic-finite-state-automata-subsec}.  The main reason is
that, as far as we currently know, this requires a more intricate inference
rule than our $\PITL$-based one~\refid{BfAux}.  In particular, a $\PTL$-based
rule suitable for our purposes must permit the simultaneous introduction of
\emph{multiple} auxiliary propositional variables analogous to the one French
and Reynolds~\citeyear{FrenchReynolds03} were compelled to employ for $\QPTL$
without past time (see also~\cite{KroegerMerz08}).

\section{Existing Completeness Proofs for Omega-Regular Logics}

\label{existing-completeness-proofs-for-omega-regular-logics-sec}

We now compare our axiomatic completeness proof with related ones for other
omega-regular logics. Here is a list of a number of such formalisms:
\begin{iteMize}{$\bullet$}
\item Logics with nonelementary complexity:
\begin{iteMize}{$-$}
%% \item Omega-regular languages (and omega-regular expressions)
\item The \emph{Second-Order Theory of Successor
  ($\SOneS$)}~\cite{Buechi62\LMCSPaperBib}
\item \emph{Regular Logic}~\cite{Paech89} (This includes a $\PITL$ subset.)
\item  Various temporal logics with quantification:
\begin{iteMize}{$*$}
\item $\QPTL$ (with and without past time) (e.g., see~\cite{KroegerMerz08})
\item Quantified $\ITL$ with finite domains~\cite{Moszkowski00}
\end{iteMize}
\end{iteMize}
\item Logics with elementary complexity:
\begin{iteMize}{$-$}
\item \emph{Extended Propositional Linear-Time Temporal Logic
    ($\ETL$)}~\cite{Wolper83}
\item \emph{Linear-Time $\mu$-Calculus ($\nuTL$)}\relax
~\cite{BarringerKuiper86,BanieqbalBarringer89a}
\item \emph{Dynamic Linear Time Temporal Logic}\relax
~\cite{HenriksenThiagarajan99}
\end{iteMize}
\end{iteMize}
Kr{\"o}ger and Merz~\citeyear{KroegerMerz08} summarise $\QPTL$ and $\nuTL$ and
some axiomatisations.  See also the earlier surveys about the expressiveness
of various formalisms such as $\PTL$ and $\QPTL$ given by Lichtenstein et
al.~\citeyear{LichtensteinPnueli85} and Emerson~\citeyear{Emerson90}.  Like
$\SOneS$ and $\QPTL$, $\PITL$ has nonelementary complexity (e.g., see our
results in collaboration with J.~Halpern in~\cite{Moszkowski83a} (reproduced
in~\cite{Moszkowski04a})).  In contrast, $\ETL$ and $\nuTL$ have only
elementary complexity.

\subsection{Omega-Regular Logics with Nonelementary Complexity}

\label{omega-regular-logics-with-nonelementary-complexity-subsec}

Let us consider axiomatic completeness for omega-regular logics which, like
$\PITL$, have nonelementary complexity.  We later discuss some with elementary
complexity in~\S\ref{omega-regular-logics-with-elementary-complexity-subsec}.

We are not the first to consider a version of quantifier-free $\PITL$ with
infinite time.  Paech~\citeyear{Paech89} in a workshop paper presents
completeness proofs for Gentzen-style axiom systems for versions of a
\emph{Regular Logic} with branching-time and linear-time and both finite and
infinite time (see also~\cite{Paech88-techrep}).  The linear-time variant $\LRL$
can be regarded as $\PITL$ with the addition of a binary temporal operator
\emph{unless}.  Paech's framework is presented in a rather different way from
ours to accommodate both branching-time and linear-time models of time, with the
overwhelming emphasis on the branching-time one.  Perhaps more significantly,
the chop-star operator $A\ConventionalChopstar$ in $\LRL$ is limited, like
Kleene star, to finitely many iterations (we look at a closely related $\PITL$
subset, called by us $\PITLK$, in \S\ref{pitl-without-omega-iteration-subsec}).
Due to a theorem of Thomas~\cite{Thomas79} (which we discuss and use in
\S\ref{pitl-without-omega-iteration-subsec} and
Section~\ref{reduction-of-chop-omega-sec}), $\LRL$ has omega-regular
expressiveness, although it is less succinct than full $\PITL$.  Paech's
restricted chop-star does not support chop-omega's infinite iteration.  Indeed,
Thomas' theorem is not at all mentioned in the completeness proof and does not
serve as a bridge in the way we apply it in
Section~\ref{reduction-of-chop-omega-sec}.  Paech's stimulating and valuable
presentation is quite detailed, especially in the extended
version~\cite{Paech88-techrep}.  Nevertheless, in our opinion (based on many
years of experience with doing proofs in $\ITL$), its treatment of $\LRL$ needs
some clarification, as the following points demonstrate:
\begin{iteMize}{$\bullet$}
\item The unwinding of chop-star does not take into account that for induction
  over time to work in $\PITL$, individual iterations need to take at least
  two states.  This contrasts with our Axioms~\refid{SChopStarEqv}
  and~\refid{ChopOmegaInduct} in
  Table~\ref{axiom-system-for-pitl-with-finite-and-infinite-time-table} and an
  analogous one which Bowman and Thompson use in~\cite{BowmanThompson03}.
  Kono's tableaux-based decision procedure for $\PITL$~\cite{Kono95} likewise
  ensures that iterations have more than one state.
\item The proof system includes nonconventional rules requiring some temporal
  formulas to be in a form analogous to regular expressions.
\item The main proof concerns a branching-time semantics.
  In contrast, only a couple of sentences are devoted to extending the proof
  to a linear-time interval framework appropriate for $\LRL$.
\item The completeness proof uses constructions involving deterministic
  automata for finite words.  It also mentions Thomas' theorem which ensures
  omega-regular expressiveness of $\LRL$.  Now the proof by Choueka and
  Peleg~\citeyear{ChouekaPeleg83} of Thomas' theorem using standard
  deterministic omega automata quite clearly shows the link between $\LRL$ and
  these automata.  However Paech does not discuss how the $\LRL$ completeness
  proof relates to techniques previously developed by
  McNaughton~\citeyear{McNaughton66} and others for building deterministic
  omega automata from deterministic automata for finite words in order to
  recognise omega-regular languages. Some kind of explicitly described
  adaptation of such methods seems to us practically unavoidable.  In
  contrast, our proof quite clearly benefits from this work as we discuss in
  detail in \S\ref{compound-semi-automata-for-suffix-recognition-sec}.
\item Except for the $\LRL$ construct $L_0$ (the same as $\Empty$ in $\PITL$),
  no derived interval-oriented operators are defined (e.g, to examine prefix
  subintervals or to perform a test in a finite interval's final
  state). Moreover, it does not appear that the $\LRL$ proof system was ever
  used for anything.
\item One minor puzzling feature of the $\LRL$ axiom system is that in its
  stated form, the linear-time proof rules for Paech's unary construct $\Next A$
  (which is actually the weak-next operator $\WeakNext$ mentioned by us in
  Table~\ref{pitl-axiom-system-for-finite-time-table}) ensure that every state
  has a successor state.  This clearly forces the linear-time variant to be
  limited to infinite state sequences.  In practice, such a requirement is
  counterproductive for $\LRL$, which permits finite time and in particular has
  a primitive finite-time construct $L_1$ that is identical to our own construct
  $\Skip$ for two-state intervals.  The $\LRL$ formula
  $L_1\ConventionalChopstar$ is used in rules to force finite intervals.  The
  $\LRL$ proof rules for $\Next$ which impose infinite time clash with rules
  containing the formula $L_1\ConventionalChopstar$ and likewise with rules
  having $L_0$ to specify one-state intervals.  However, the difficulty with the
  $\LRL$ operator $\Next$ and infinite intervals seems to be an easily
  correctable oversight.
\end{iteMize}
Unfortunately, no subsequent versions of Paech's completeness proof for $\LRL$
with more explanations and clarifications have been published.  Indeed, the
difficulties faced at the time by Paech and others such as Rosner and
Pnueli~\citeyear{RosnerPnueli86} (discussed below) when attempting to develop
complete axiomatisations of versions of $\ITL$ with infinite time were such
that subsequent published work in this area did not appear until over ten
years later.  Incidentally, the manner of Paech's proof based on
\emph{Propositional Dynamic Logic}
($\PDL$)~\cite{FischerLadner79,HarelKozen2000} and the associated
\emph{Fischer-Ladner closures} suggests that it could have connections with
much later research by Henriksen and
Thiagarajan~\citeyear{HenriksenThiagarajan99} on axiomatising \emph{Dynamic
  Linear Time Temporal Logic}, a formalism combining $\PTL$ and $\PDL$ which
we shortly mention in
\S\ref{omega-regular-logics-with-elementary-complexity-subsec}.  On the other
hand, our own $\PITL$ completeness proof here and our earlier one for $\PITL$
with just finite time~\cite{Moszkowski04a} do not involve Fischer-Ladner
closures.

Completeness proofs for logics such as $\SOneS$~\cite{Siefkes70}, $\QPTL$ with
past time~\cite{KestenPnueli95,KestenPnueli2002} and without past
time~\cite{FrenchReynolds03} and one by us for quantified $\ITL$ with finite
domains~\cite{Moszkowski00} use quantified formulas encoding omega automata
and explicit deductions involving nontrivial techniques to complement them.
As we already noted in Section~\ref{introduction-sec}, our earlier axiomatic
completeness proof~\cite{Moszkowski00} for quantified $\ITL$ with finite
domains requires the use of quantifiers and does not work when formulas were
limited to have just propositional variables.  French and
Reynold's~\citeyear{FrenchReynolds03} axiom system for $\QPTL$ without past
time contains a nontrivial inference rule for introducing a variable number of
auxiliary variables.  This inference rule is required by the automata-based
completeness proof.

The axiomatic completeness proofs for the logics with quantification just
mentioned with nonelementary complexity involve using quantified auxiliary
variables to re-express a formula $A$ as another semantically equivalent
formula $\Exists{p_1\ldots p_n}X$, where $\exists$ for $\QPITL$ and $\QPTL$ is
defined earlier in Section~\ref{pitl-sec}.  Here $p_1,\dots, p_n$ are the
auxiliary variables and $X$ is a formula in a much simpler logical subset,
such as some version of (quantifier-free) $\PTL$.  Axiomatic completeness for
the subset is much easier to show than for the original logic.  Completeness
is then proved by the standard technique of demonstrating that any consistent
formula $A$ (i.e., not deducibly false) in the full logic is also satisfiable.
In particular, we deduce as a theorem the equivalence
$A\equiv\Exists{p_1\ldots p_n}X$.  Now from this, the assumed logical
consistency of $A$ and simple propositional reasoning, we readily obtain
consistency for $\Exists{p_1\ldots p_n}X$.  Standard reasoning about
quantifiers then ensures $X$ is consistent.  Completeness for the logical
subset yields a model for $X$ which can also serve as one for $A$.  Normally
in such completeness proofs, the formula $X$ encodes some kind of omega
automaton such as a nondeterministic B\"uchi automata. The details are not
relevant for our purposes here.  The deduction of the equivalence
$A\equiv\Exists{p_1\ldots p_n}X$ in these proofs has always involved
explicitly embedding nontrivial techniques for manipulating such omega
automata.

\begin{mycomment}
  In contrast to such established proofs for logics with nonelementary
  complexity, our approach does not need quantifiers.  We start by ensuring
  completeness for a subset of $\PITL$ formulas restricted to a version of
  $\PTL$ with an $\Until$ operator and called $\PTLU$.  From a consistent
  $\PITL$ formula $A$ we eventually obtain a $\PTL$ formula $X$ and $\PTLU$
  formula $Y$ which can include auxiliary propositional variables.  We ensure
  the implication $Y \And X\imp A$ is deducible as a $\PITL$ theorem and also
  that $Y\And X$, like $A$, is consistent.  Completeness for $\PTLU$ yields a
  model for $Y\And X$, which also serves as one for $A$.  Now if we assume
  that $p_1\ldots p_n$ are the auxiliary variables, then $A$ and
  $\Exists{p_1\ldots p_n}(Y\And X)$ are semantically equivalent and the
  $\QPITL$ formula $A\equiv\Exists{p_1\ldots p_n}(Y\And X)$ is valid.
  However, our proof, unlike the other established proofs for omega-regular
  logics with nonelementary complexity, does not need to deduce such a
  quantified formula and indeed cannot since basic $\PITL$ does not have
  quantifiers.
\end{mycomment}

In contrast to our approach, most of the established axiomatic completeness
proofs for logics with nonelementary complexity need quantifiers.  The one
exception is Paech's Regular Logic, which does not have quantifiers and in
linear time is like our $\PITLK$, the subset of $\PITL$ without chop-omega
defined earlier in \S\ref{pitl-without-omega-iteration-subsec}.  Our
quantifier-free proof also benefits from the hierarchical application of some
previously obtained semantic theorems and related techniques expressible as
valid formulas in restricted versions of $\PITL$ (such as $\PITL$ with just
finite time).  This largely spares us from explicit, tricky reasoning about
complementing omega automata.  Once we have ensured axiomatic completeness for
these versions of $\PITL$, valid formulas in them can be immediately deduced as
theorems. For example, we invoke (without proof) the theorem of Thomas at the
end of~\cite{Thomas79} to show that $\PITLK$ has the same expressiveness as full
$\PITL$.  Our completeness proof then combines this result with completeness for
$\PITLK$ to demonstrate that any $\PITL$ formula is deducibly equivalent to one
in $\PITLK$.

Our completeness proof for $\PITL$ with both finite and infinite time does not
actually require a proof of the axiomatic completeness of a version of $\PTL$
with this time model because Axiom~\refid{VPTL} in
Table~\ref{axiom-system-for-pitl-with-finite-and-infinite-time-table} includes
all substitution instances of valid $\PTL$ formulas.  For our purposes, even
axiomatic completeness for $\PTLU$ can be based on a reduction to $\PTL$ which
invokes Axiom~\refid{VPTL}.  However, as we noted in
\S\ref{alternative-axioms-for-ptl-subsec}, some alternative, lower level
axioms for the $\PITL$ axiom system can be used which would actually involve
the reliance on a complete $\PTL$ axiom system.  Our older axiom system for
$\PITL$ with just finite time in
Table~\ref{pitl-axiom-system-for-finite-time-table} includes explicit axioms
of this sort but of course can be readily modified to similarly use just a
version of Axiom~\refid{VPTL} for finite time.

Even if we choose to use the alternative axioms and therefore explicitly rely on
some provably complete $\PTL$ axiom system, the proofs are fairly easy to obtain
via tableaux and other means (e.g., see Gabbay et al.~\citeyear{GabbayPnueli80},
Lichtenstein and Pnueli~\citeyear{LichtensteinPnueli00}, Kr{\"o}ger and
Merz~\citeyear{KroegerMerz08} and
Moszkowski~\citeyear{Moszkowski07\LMCSPaperBib}).  Such methods often have
associated practical decision procedures which in many cases are not so hard to
implement.  This contrasts with the explicit encoding in deductions of much more
difficult automata-theoretic and combinatorical techniques to complement
omega-regular languages in completeness proofs for other omega-regular logics
with nonelementary complexity such as $\SOneS$~\cite{Siefkes70} and two versions
of $\QPTL$~\cite{KestenPnueli2002,FrenchReynolds03}.  Furthermore, the
completeness proofs for $\QPTL$ in any case also rely on reductions to some form
of axiomatic completeness for $\PTL$ (which, like in our presentation, can be
used without reproving it).  Those $\QPTL$ axiom systems could alternatively be
modified to include a suitable version of our Axiom~\refid{VPTL}.  So even if we
add a few extra axioms for $\PTL$, we still feel justified in regarding our
approach, which is partly based on invoking Thomas' theorem without having to
encode a proof of it in deductions, as indeed being much more implicit than
previous completeness proofs for omega-regular logics with nonelementary
complexity such as $\SOneS$ and $\QPTL$.

\begin{myremark}
  As noted above, unlike previous automata-based approaches, ours avoids
  explicitly defining omega automata and embedding various associated explicit
  deductions concerning complicated proofs of some known results about them.
  Nevertheless, omega automata can be used in a simple semantic argument
  ensuring that for any satisfiable $\PITL$ formula, there exists some
  satisfiable $\PTL$ formula which implies it.  This is because any
  omega-regular language can be recognised by such an automaton which itself
  is encodable in a $\QPTL$ formula of the form $\Exists{p_1\ldots p_n}X'$,
  for some $\PTL$ formula $X'$.  So for any $\PITL$ formula, there is some
  semantically equivalent $\QPTL$ formula of this kind and its quantifier-free
  part therefore implies the $\PITL$ formula.  Clearly, the $\PITL$ formula is
  satisfiable iff the $\PTL$ subformula is.
%   Here $X'$ being in $\PTL$ suffices since no $\Until$ constructs
%   are in fact needed, although we prefer to use them since they are
%   convenient.
\end{myremark}

Rosner and Pnueli's version of $\PITL$~\citeyear{RosnerPnueli86} with infinite
time and without chop-star is not an omega-regular logic since it has the
(more limited) expressiveness of conventional $\PTL$.  Nevertheless, it in
common with $\SOneS$, $\QPTL$ and $\PITL$ has nonelementary computational
complexity.  Rosner and Pnueli's complete axiom system includes a complicated
inference rule which requires the construction of a table.

\subsection{Omega-Regular Logics with Elementary Complexity}

\label{omega-regular-logics-with-elementary-complexity-subsec}

As we previously noted, $\ETL$, $\nuTL$ and Dynamic Linear Time Temporal Logic
have only elementary complexity.  Wolper~\citeyear{Wolper82a,Wolper83} proves
axiomatic completeness for $\ETL$ but Banieqbal and
Barringer~\citeyear{BanieqbalBarringer86} later present a correction to
Wolper's axiom system and proof requiring a table-based inference rule.
Walukiewicz~\citeyear{Walukiewicz95a} is the first to show axiomatic
completeness for the \emph{modal
  mu-calculus}~\cite{Kozen83,Stirling01,BradfieldStirling06} which subsumes
$\nuTL$.  Kaivola's~\citeyear{Kaivola95} subsequent less complicated
completeness proof for just $\nuTL$ uses a partially semantic approach which
has some similar aims to ours for $\PITL$, but is nevertheless technically
quite different.  It involves a clever normal form and tableaux.  Every
formula is shown to be deducibly equivalent to one in the normal form.  We
believe that our proof, although longer, is in certain respects more natural
and straightforward than even Kaivola's at the deductive level.

Dynamic Linear Time Temporal Logic combines $\PTL$ and \emph{Propositional
  Dynamic Logic (PDL)}~\cite{FischerLadner79,HarelKozen2000} in a linear-time
framework with infinite time.  The axiom system for this formalism has axioms
concerning a variety of transitions~\cite{HenriksenThiagarajan99}.  The
completeness proof is an adaptation of an earlier one for $\PDL$ by Kozen and
Parikh~\citeyear{KozenParikh81}.  It uses consistent sets of formulas.

\section{Future Work}

\label{future-work-sec}

Our plans include using the axiom system as a hierarchical basis for
completeness of $\PITL$ variants with weak chop and chop-star taken as
primitives as well as quantification.  Further possibilities include multiple
time granularities (see our work~\cite{Moszkowski95a} for finite time), a
temporal Hoare logic and also logics such as $\QPTL$ (by encoding within
$\QPTL$ a complete axiom system for quantified $\PITL$ instead using of omega
automata).  The last would show interval logics can be applied to point-based
ones.

% An $\ITL$-based axiomatic completeness for $\nuTL$ using an analogous encoding
% of $\nuTL$ in $\QPTL$ (e.g., using the one in~\cite{BanieqbalBarringer89a})
% together with an encoding of an axiom system for $\nuTL$~\cite{Kaivola95}
% currently appears to be quite difficult but worth exploring.  We also want to
% compare $\PITL$ with a version having $\nuTL$'s $\mu$ fixpoint operator
% (suggested by Colin Stirling).

In~\cite{Moszkowski04a}, we used semantic techniques to prove axiomatic
completeness for $\PITL$ with finite time by a simple reduction to an equally
expressiveness subset called by us \emph{Fusion Logic} and closely related to
\emph{Propositional Dynamic Logic
  (PDL)}~\cite{FischerLadner79,HarelKozen2000}.  Fusion Logic, like some
variants of PDL, uses discrete linear sequences of states instead of binary
relations as its semantic basis.  Some of the semantic techniques we presented
in Section~\ref{reduction-of-chop-omega-sec} for reducing $\PITL$ to its
expressively equivalent subset $\PITLK$ by eliminating instances of chop-omega
could shorten the completeness proof for Fusion Logic in~\cite{Moszkowski04a},
since that proof contains a similar elimination of chop-star by reduction down
to $\PTL$.  Furthermore, our completeness proof for $\PITL$ with just finite
time in~\cite{Moszkowski04a} uses a separate complete axiom system for Fusion
Logic. This now seems unnecessary for the overall completeness proof for
$\PITL$ with finite time.  Instead, the $\PITL$ axiom system should also
suffice for Fusion Logic in view of our positive experiences with the current
much more streamlined approach for $\PITL$ with infinite time.

The $\PITL$ operators $\Df$ and $\Bf$ for finite prefix subintervals play a
major role in our new completeness proof and appear worthy of more
consideration.  For example, we have recently studied techniques for reasoning
about them with \emph{time reversal}~\cite{Moszkowski11-TIME}.  This is a
natural mathematical way to exploit the symmetry of time in finite intervals. We
can show the validity of suitable finite-time formulas concerning $\Bf$ and
prefix subintervals from the validity of analogous ones for $\Box$ and suffix
subintervals which themselves might even be in conventional $\PTL$ with the
operator $\UntilOp$.  The time symmetry considered here only applies to finite
intervals. However, a valid finite-time formula obtained in this way can
sometimes then be generalised to infinite intervals.  One potential use of time
reversal is to provide an algorithmic reduction of suitable higher-level $\PITL$
formulas to lower-level $\PTL$ ones for model checking.  It also helps extend
compositional techniques we described
in~\cite{Moszkowski94,Moszkowski96,Moszkowski98}.

\begin{mycomment}
  More precisely, an additional $\PITL$ construct is used which has the form
  $A^r$ for any $\PITL$ formula $A$.  Let $\sigma\vld A^r$ be true iff
  $\sigma$ has finite length and $\sigma^r \vld A$, where $\sigma^r$ is
  $\sigma$'s reversal (i.e., the finite interval
  $\sigma_{\intlen{\sigma}}\ldots\sigma_0$).  Using this construct, valid
  formulas involving suffix subintervals (as in basic $\PTL$ with finite time
  where all temporal operators concern such suffixes) can with care be
  interpreted in reverse to obtain useful valid $\PITL$ formulas concerning
  prefix subintervals.

  Likewise, testing for satisfiability for such formulas with $\Bf$ can be
  reduced to it for formulas with $\Box$.  In some cases, the formulas
  containing $\Bf$ can then be extended to deal with the infinite time
  behaviour of $\Bf$ as well.  Time reversal work wells with compositional
  fixpoints of the $\ITL$ operator chop-star (i.e., any formula $A$ for which
  $A\equiv A\ConventionalChopstar$ is valid) which we previous described
  in~\cite{Moszkowski94,Moszkowski96,Moszkowski98}.  Axiomatic completeness
  for $\PITL$ with the operator for time reversal does not appear to be
  difficult to achieve.

  We regard time reversal as potentially offering intriguing evidence
  suggesting that certain subsets of $\ITL$ might offer a fundamentally more
  appropriate and intuitive compositional framework for analysis of concurrent
  systems than the better established point-based logics.  Subsets of $\PITL$
  easily reducible to point-based formulas via time reversal would, like the
  propositional point-based logics, have elementary complexity.  From the
  interval-oriented perspective, traditional point-based temporal logics would
  then be regarded as being lower level, more computationally oriented
  frameworks having an important but nevertheless subordinate role.  Indeed,
  the independently made comments we already noted in our introduction in
  Section~\ref{introduction-sec} by B{\"a}umler et
  al.~\citeyear{BaeumlerSchellhorn09} and Olderog and
  Dierks~\citeyear{OlderogDierks2008} concerning their preference of
  interval-based temporal logics over point-based ones seems quite consistent
  with this observation.
\end{mycomment}

\begin{mycomment}
  Interestingly, in the context of model checking, Vardi~\citeyear{Vardi01}
  analogously regards branching-time temporal logics as being lower level,
  less natural and less compositional than linear-time point-based ones but
  offering a computationally attractive framework which can be profitably
  extended to linear-time temporal logic by means of reductions.  Yet when
  Vardi discusses his view of the ``ultimate'' temporal logic specification
  language for model checking, he only briefly mentions that there exist
  proponents of interval-based linear time logics.  He himself does not seem
  to regard such notations as having much to offer to the wider temporal logic
  community.  Perhaps such an assessment could at some time in the future
  benefit from a re-evaluation which considers computationally feasible
  systematic transformations for interval formulas to point-based ones, much
  like the reductions Vardi mentions from linear time to branching time.  For
  some applications, one might even come to regard solely point-based
  specifications and proofs about transition systems as analogous to the
  unstructured use of gotos criticised by Dijkstra~\citeyear{Dijkstra68}.
  Indeed B{\"a}umler et al.~\citeyear{BaeumlerSchellhorn09} note the following
  concerning algorithms analysed by the KIV interactive theorem prover:
  \begin{quote}
    The line numbers are given for explanatory purposes only. They are not
    used in KIV.
  \end{quote}
  In contrast, conventional point-based proofs often make extensive useful of
  line numbers or locations in programs when reasoning about state
  transitions.
\end{mycomment}

\section*{Conclusions}

We have presented a simple axiom system for $\PITL$ with infinite time and
proved completeness using a semantic framework and reductions to finite time
and $\PTL$.  Our axiom system is demonstrably simpler than the one which Paech
presents for $\LRL$, even though we support omega-iteration and $\LRL$ does
not.  Moreover, the explicitly stated deductions in our proof can be regarded
as being technically less complex then others for quantified omega-regular
logics with nonelementary complexity such as $\SOneS$ and $\QPTL$.  This is
because known completeness proofs for those logics involve an \emph{explicit}
deductive embedding of proofs of theorems about complementing omega-regular
languages and require reasoning about nontrivial algorithms (typically
utilising quantifier-based encodings of omega automata).  Such completeness
proofs therefore do not merely use one such theorem but incorporate
significant aspects of its complicated proof, in effect \emph{reproving it.}
In contrast, we simply invoke Thomas' theorem without referring to how it is
proved.  In our opinion, this conforms much more to the conventional
mathematical practice of using previously established theorems, even
hard-to-prove ones, as modular ``black boxes''.  However, we appreciate that
some readers will argue about the significance of this technical point.

The overall results we have described in our new completeness proof seem to
complement our recent analysis of $\PTL$ using
$\PITL$~\cite{Moszkowski07\LMCSPaperBib}.  One surprise during the development
of our completeness proof concerned how much explicit deductions could be
minimised by application of valid properties proved with semi-automata and
automata on finite words.  Another unexpected benefit arose from the insights
into time reversal.

% Furthermore, as noted in the previous Section~\ref{future-work-sec}, our new
% $\PITL$ axiom system could even serve as a basis for simpler axiom systems and
% proofs of axiomatic completeness for some point-based propositional temporal
% logics such as $\QPTL$ without past time.  In addition, some of the $\PITL$
% operators used in our proof together with techniques for time reversal might
% offer natural and computationally practical reductions from interval-oriented
% specifications to lower level ones in more conventional point-based temporal
% logics.

% In connection with this, it seems reasonable to finish off with a recent
% inspiring remark which Professor Zhou Choachen, one of the originators of the
% Duration Calculus~\cite{ZhouHoare91,ZhouHansen2004} and a member of the
% Chinese Academy of Sciences, made in~\cite[page 610]{Zhou07} after mentioning
% various successful applications of interval-based notations:
% \begin{quote}
%   I believe that Interval Logic deserves an important role in computing
%   science.
% \end{quote}

\section*{Acknowledgements}

We thank Antonio Cau, Dimitar Guelev, Helge Janicke, Colin Stirling, Georg
Struth and the anonymous referees for suggestions.  Shirley Craig provided
outstanding library services and deserves special mention.

\iffalse
  % Force no space between references.
  \let\stdbibitem=\bibitem

  \def\bibitem#1{\itemsep=0pt %
  \parsep=0pt %
  \topsep=0pt %
  \partopsep=0pt %
  \parskip=0pt %
  \stdbibitem{#1}}
\fi

\bibliography{\jobname-strings,biblio}
%\bibliography{completeness-simple-journal-submission-strings,\jobname,biblio}

%\clearpage

\appendix

\section{Some PITL theorems and Their Proofs}

\label{some-pitl-theorems-and-their-proofs-sec}

\InTextBodyfalse  % needed for MyThms support
\immediate\write\MyThmsOut{\noexpand\MyThmsMacro}  % needed for MyThms support

\bgroup

\parskip=0.75em %

\defid[noxrefs,bold=false]{Assump}{assump.,,}
%\def\Assump{\textrm{assump.}}
%\let\Assump=\relax
%\dummyid{Assump}
%% \def\given{\textrm{given}}
%% \dummyid{given}
\defid[noxrefs]{Prop}{Prop,,}
\defid[noxrefs]{ImpChain}{$\imp$-chain,,}
\defid[noxrefs]{EqvChain}{$\equiv$-chain,,}
\defid[noxrefs]{PITLF}{PITLF,,}
%  WARNING: The following name clobbers the normal \FLV macro.
%  The refid system should be fixed to avoid this.
%\defid[noxrefs]{FLV}{$\noexpand\FLV$,,}
\begin{comment}
  \defid[noxrefs]{Meta}{Meta,,}
  \defid[noxrefs]{Quant}{Quant,,}
  \defid[noxrefs]{Induct}{Induct,,}
\end{comment}

%% The following emacs lisp code edits proofs taken from imp*.tex files.
\begin{mycomment}
(progn
  (beginning-of-buffer)
  (replace-string "\r" "" nil)  ;; Remove every "^M" (carriage return)
  (beginning-of-buffer)
  (replace-regexp "\\\\citerul[e]" "\\\\refid" nil)
  (beginning-of-buffer)
  (replace-string ";" "\\chop " nil)
  (beginning-of-buffer)
  (replace-regexp "\\\\theorem[ ]" "\\\\theoremBEN " nil)
  (beginning-of-buffer)
  (my-query-replace-regexp "D[i]" "Df" nil)
  (beginning-of-buffer)
  (replace-string "Dfst" "Dist" nil)
  (beginning-of-buffer)
  (my-query-replace-regexp "B[i]" "Bf" nil)
  (beginning-of-buffer)
  (my-query-replace-regexp "B[f]Gen" "BfFGen" nil)
  (beginning-of-buffer)
  (replace-regexp "\\bS'\\b" "A'" nil)
  (beginning-of-buffer)
  (replace-regexp "\\bS\\b" "A" nil)
  (beginning-of-buffer)
  (replace-regexp "\\bT'\\b" "B'" nil)
  (beginning-of-buffer)
  (replace-regexp "\\bT\\b" "B" nil)
  )
\end{mycomment}

This appendix gives a representative set of $\PITL$ theorems and derived
inference rules together with their proofs.  Many are used either directly or
indirectly in the completeness proof for $\PITL$ with both finite and infinite
time.  We have partially organised the material, particularly in
\S\ref{some-properties-of-bf-involving-the-modal-system-k-and-axiom-d-subsec},
along the lines of some standard modal logic
systems~\cite{Chellas80,HughesCresswell96}.

The $\PITL$ theorems and derived rules have a shared index sequence (e.g.,
\refid{BfChopImpChop}--\refid{BoxChopEqvChop} are followed by \refid{BfGen}
rather than \textbf{DR1}).  We believe that this convention simplifies
locating material in this appendix and also in
Table~\ref{list-of-pitl-theorems-and-derived-rules-mentioned-before-app-table}
found earlier in
\S\ref{formal-equivalence-of-the-two-representations-of-runs-subsec}.

Proof steps can refer to axioms, inference rules,
previously deduced theorems, derived inference rules and also the following:
\begin{iteMize}{$\bullet$}
\item \myhypertarget{refid.Assump.0}{\refid{Assump}}: Assumptions which are
  regarded as being previously deduced.

\item \myhypertarget{refid.Prop.0}{\refid{Prop}}: Conventional nonmodal
  propositional reasoning (by restricted application of Axiom~\refid{VPTL})
  and modus ponens.

\item \myhypertarget{refid.ImpChain.0}{\refid{ImpChain}}: A chain of
  implications.

\item \myhypertarget{refid.EqvChain.0}{\refid{EqvChain}}: A chain of
  equivalences.

\item[] In principle, \refid{ImpChain} and \refid{EqvChain} are subsumed by
  \refid{Prop} but are used here to make the reasoning more explicit.

\item \myhypertarget{refid.PITLF.0}{\refid{PITLF}}: Our assumption of
  axiomatic completeness for $\PITL$ with just finite time permits any valid
  implication of the form $\Finite \imp A$.
\end{iteMize}

\subsection{Some Basic Properties of Chop} 

\ifmylmcs \hfill \fi

We now consider deducing various simple properties of chop and the associated
operators $\Df$, $\Bf$, $\Diamond$ and $\Box$ which have a wide range of uses.

\statetheorem{BfChopImpChop&
\Theorem \Bf(A\implies A')  \Implies  (A\chop B) \implies (A'\chop B)}%
\begin{proveit}
1&B \Implies B   &\refid{Prop} \\
2&\Box(B \implies B )   &1,\refid{BoxGen} \\
3&\Bf(A\implies A') \And \Box(B\implies B) \Implies
   (A\chop B) \implies (A'\chop B)   &\refid{BfAndBoxImpChopImpChop} \\
4&\Bf(A\implies A') \Implies  (A\chop B) \implies (A'\chop B)
   &2,3,\refid{Prop}
\end{proveit}

\statetheorem{BoxChopImpChop&
\Theorem \Box(B\implies B')  \Implies  (A\chop B) \implies (A\chop B')}%
\begin{proveit}
1&\Finite \implies (A \imp A)   &\refid{Prop} \\
2&\Bf(A \implies A)   &1,\refid{BfFGen} \\
3&\Bf(A\implies A) \And \Box(B\implies B') \Implies
   (A\chop B) \implies (A\chop B')   &\refid{BfAndBoxImpChopImpChop} \\
4&\Box(B\implies B') \Implies   (A\chop B) \implies (A\chop B')
   &2,3,\refid{Prop}
\end{proveit}

\statetheorem{BoxChopEqvChop&
\Theorem \Box(B\equiv B')  \Implies  (A\chop B) \equiv (A\chop B')}%
\begin{proveit}
1&\Box(B\equiv B') \Equiv \Box(B\imp B') \And \Box(B'\imp B)
  &\refid{VPTL} \\
2&\Box(B\implies B')  \Implies  (A\chop B) \implies (A\chop B')
  &\refid{BoxChopImpChop} \\
3&\Box(B'\implies B)  \Implies  (A\chop B') \implies (A\chop B)
  &\refid{BoxChopImpChop} \\
4&\Box(B\equiv B')  \Implies  (A\chop B) \equiv (A\chop B')  &2,3,\refid{Prop}
\end{proveit}

The following derived variant of Inference Rule~\refid{BfFGen} omits the
subformula $\Finite$:
\statederivedrule{BfGen&
\theoremBEN A \infer \theoremBEN \Bf A}%
\begin{proveit}
1&A &\refid{Assump} \\
2&\Finite \implies A &1,\refid{Prop} \\
3&\Bf A &2, \refid{BfFGen}
\end{proveit}
The derived inference rule~\refid{BfGen} can also be referred to
as~\textbf{$\Bf$Gen} (analogous to the inference rule~\refid{BoxGen}).

\statederivedrule{LeftChopImpChop&
\theoremBEN A\implies A' \infer \theoremBEN (A\chop B) \implies (A'\chop B)}%
\begin{proveit}
1&A\Implies A'   &\refid{Assump} \\
2&\Bf(A\implies A')   &1,\refid{BfGen} \\
3&\Bf(A\implies A') \Implies (A\chop B) \implies (A'\chop B)
   &\refid{BfChopImpChop} \\
4&A\chop B \Implies A'\chop B   &2,3,\refid{MP}
\end{proveit}

\statederivedrule{LeftChopEqvChop&
\theoremBEN A \equiv A' \infer \theoremBEN (A\chop B) \equiv (A'\chop B)}%
\begin{proveit}
1&A\Equiv A'&\refid{Assump} \\
2&A\Implies A'&1,\refid{Prop} \\
3&A\chop B \Implies A'\chop B&2,\refid{LeftChopImpChop} \\
4&A'\Implies A&1,\refid{Prop} \\
5&A'\chop B \Implies A\chop B&4,\refid{LeftChopImpChop} \\
6&A\chop B \Equiv A'\chop B&3,5,\refid{Prop}
\end{proveit}

\statederivedrule{DfImpDf&
\theoremBEN A\implies B \infer \theoremBEN \Df A \implies \Df B}%
\begin{proveit}
1&A\Implies B   &\refid{Assump} \\
2&A\chop \True \Implies B\chop \True   &1,\refid{LeftChopImpChop} \\
3&\Df A \Implies \Df B   &2,\defof{\Df}
\end{proveit}

\statederivedrule{DfEqvDf&
\theoremBEN A\equiv B \infer \theoremBEN \Df A \equiv \Df B}%
\begin{proveit}
1&A\Equiv B &\refid{Assump} \\
2&A\chop \True \Equiv B\chop \True &1,\refid{LeftChopEqvChop} \\
3&\Df A \Equiv \Df B &2,\defof{\Df}
\end{proveit}

\statederivedrule{RightChopImpChop&
\theoremBEN B\implies B' \infer \theoremBEN (A\chop B) \implies (A\chop B')}%
\begin{proveit}
1&B \Implies B'   &\refid{Assump} \\
2&\Box(B \implies B')   &\refid{BoxGen} \\
3&\Box(B \implies B') \Implies (A\chop B) \implies (A\chop B')
   &\refid{BoxChopImpChop} \\
4&A\chop B \Implies A\chop B'   &2,3,\refid{MP}
\end{proveit}

\statederivedrule{RightChopEqvChop&
\theoremBEN B\equiv B' \infer \theoremBEN (A\chop B) \EQUIV (A\chop B')}%
\begin{proveit}
1&B \Equiv B'   &\refid{Assump} \\
2&B \Implies B'   &1,\refid{Prop} \\
3&A\chop B \Implies A\chop B'   &2,\refid{RightChopImpChop} \\
4&B' \Implies B   &1,\refid{Prop} \\
5&A\chop B' \Implies A\chop B   &4,\refid{RightChopImpChop} \\
6&A\chop B \Equiv A\chop B' &3,5,\refid{Prop}
\end{proveit}

\statederivedrule{DiamondEqvDiamond&
\theoremBEN A\equiv B \infer \theoremBEN \Diamond A \equiv \Diamond B}%
\begin{proveit}
1&A\Equiv B &\refid{Assump} \\
2&\True\chop A \Equiv \True\chop B &1,\refid{RightChopEqvChop} \\
3&\Diamond A \Equiv \Diamond B &2,\defof{\Diamond}
\end{proveit}

\statederivedrule{BoxEqvBox&
\theoremBEN A\equiv B \infer \theoremBEN \Box A \equiv \Box B}%
\begin{proveit}
1&A \Equiv B  &\refid{Assump} \\
2&\Not A \Equiv \Not B  &1,\refid{Prop} \\
3&\Diamond\Not A \Equiv \Diamond\Not B  &2,\refid{DiamondEqvDiamond} \\
4&\Not\Diamond\Not A \Equiv \Not\Diamond\Not B  &3,\refid{Prop} \\
5&\Box A \Equiv \Box B  &4,\defof{\Box}
\end{proveit}

\statederivedrule{BoxImpInferBoxImpBox&
\theoremBEN \Box A\,\implies\, B \infer \theoremBEN \Box A\,\implies\,\Box B}%
\begin{proveit}
1&\Box A\Implies B  &\refid{Assump} \\
2&\Box(\Box A\implies B)  &1,\refid{BoxGen} \\
3&\Box(\Box A\implies B) \Implies (\Box A \implies \Box B)  &\refid{VPTL} \\
4&\Box A \Implies \Box B  &2,3,\refid{MP}
\end{proveit}

\statetheorem{AndChopA&
\Theorem (A \And A')\chop B \Implies A\chop B}%
\begin{proveit}
1&A \And A' \Implies A &\refid{Prop} \\
2&(A \And A')\chop B \Implies A\chop B &1,\refid{LeftChopImpChop}
\end{proveit}

\statetheorem{AndChopB&
\Theorem (A \And A')\chop B \Implies A'\chop B}%
\begin{proveit}
1&A \And A' \Implies A' &\refid{Prop} \\
2&(A \And A')\chop B \Implies A'\chop B &1,\refid{LeftChopImpChop}
\end{proveit}

\statetheorem{AndChopImpChopAndChop&
\Theorem (A \And A')\chop B \Implies (A\chop B) \Andd (A'\chop B)}%
\begin{proveit}
1&(A \And A')\chop B \Implies A\chop B &\refid{AndChopA} \\
2&(A \And A')\chop B \Implies A'\chop B &\refid{AndChopB} \\
3&(A \And A')\chop B \Implies (A\chop B) \Andd (A'\chop B) &1,2,\refid{Prop}
\end{proveit}

\statetheorem{AndChopCommute&
\Theorem (A \And A')\chop B \Equiv (A' \And A)\chop B}%
\begin{proveit}
1&A \And A' \Equiv A' \And A  &\refid{Prop} \\
2&(A \And A')\chop B \Equiv (A' \And A)\chop B  &1,\refid{LeftChopEqvChop}
\end{proveit}

\statetheorem{OrChopEqv&
\Theorem (A\Or A')\chop B \Equiv (A\chop B) \;\Or\; (A'\chop B)}%
The proof for $\supset$ is immediate from axiom \refid{OrChopImp}.
Here is the proof for $\subset$:
\vspace{3pt}
\begin{proveit}
1&A \Implies A\Or A'  &\refid{Prop} \\
2&A\chop B \Implies (A\Or A')\chop B  &1,\refid{LeftChopImpChop} \\
3&A' \Implies A\Or A'  &\refid{Prop} \\
4&A'\chop B \Implies (A\Or A')\chop B  &3,\refid{LeftChopImpChop} \\
5&(A\chop B) \;\Or\; (A\chop B') \Implies (A\Or A')\chop B  &2,4,\refid{Prop}
\end{proveit}

\statetheorem{ChopImpDf&
\Theorem A\chop B \Implies \Df A}%
\begin{proveit}
1&B \Implies \True  &\refid{Prop} \\
2&A\chop B \Implies A\chop \True  &1,\refid{RightChopImpChop} \\
3&A\chop B \Implies \Df A  &2,\defof{\Df}
\end{proveit}

\statetheorem{DfEmpty&
\Theorem \Df\Empty}%
\begin{proveit}
1&\Empty\chop\True \Equiv \True &\refid{EmptyChop} \\
2&\Empty\chop\True \Implies \Df \Empty &\refid{ChopImpDf} \\
3&\Df\Empty &1,2,\refid{Prop}
\end{proveit}

\statetheorem{ChopImpDiamond&
\Theorem A\chop B \Implies \Diamond B}%
\begin{proveit}
1&A \Implies \True   &\refid{Prop} \\
2&A\chop B \Implies \True\chop B   &1,\refid{LeftChopImpChop} \\
3&A\chop B \Implies \Diamond B   &2,\defof{\Diamond}
\end{proveit}

\subsection{Some Properties of \texorpdfstring{$\Bf$}{Box-f} involving the Modal System K and Axiom D}

\label{some-properties-of-bf-involving-the-modal-system-k-and-axiom-d-subsec}

\ifmylmcs \hfill \fi

The two pairs of operators $\Box$ and $\Diamond$ and $\Bf$ and $\Df$ obey
various standard properties of modal logics.  Axiom~\refid{VPTL} helps
streamline reasoning involving $\Box$ and $\Diamond$.  The situation with
$\Bf$ and $\Df$ is quite different since they lack a comparable axiom.
Therefore, it is especially beneficial to review some conventional modal
systems which assist in organising various useful deductions involving $\Bf$
and $\Df$.

Table~\ref{some-standard-modal-systems-table} summarises some relevant modal
systems, various associated axioms and inference rules.
Chellas~\citeyear{Chellas80} and Hughes and
Cresswell~\citeyear{HughesCresswell96} give more details.
\begin{table*}[h]
\begin{center}
\begin{tabular}{llLc}
\multicolumn{1}{c}{System}& & \multicolumn{1}{c}{Axiom or inference rule}
  & \multicolumn{1}{c}{Axiom or rule name} \\\hline\noalign{\vspace{2pt}}
\ModalSys{K}: & & M\,A \Defeqv \Not L\,\Not A & M-def\\
 & \hphantom{\ModalSys{K}} plus
   & \theoremBEN L(A\imp B) \implies (L\,A \imp L\,B) & K \\
 & \hphantom{\ModalSys{K}} plus & \theoremBEN A \infer \theoremBEN L\,A & N \\
\ModalSys{T}: & \ModalSys{K} plus & \theoremBEN L\,A \implies A & T \\
\ModalSys{S4}: & \ModalSys{T} plus & \theoremBEN L\,A \implies LL\,A & 4 \\
\ModalSys{KD4}: & \ModalSys{K} plus 4 and
 & \theoremBEN L\,A \implies M\,A & D \\
\end{tabular}
\end{center}
\caption{Some standard modal systems}
\label{some-standard-modal-systems-table}
\end{table*}

Within $\PITL$, as in $\PTL$, the operator $\Box$ can be regarded as the
conventional unary \emph{necessity} modality $L$ and the operator $\Diamond$
as the dual \emph{possibility} operator $M$.  The two operators together
fulfil the requirements of the modal system \ModalSys{S4}. We do not need to
explicitly prove versions of the \ModalSys{S4} axioms in
Table~\ref{some-standard-modal-systems-table} for $\Box$ and $\Diamond$.
Rather, any $\PITL$ formula which is a substitution instance of a valid
\ModalSys{S4} formula involving $\Box$ and $\Diamond$ can be readily deduced
using the $\PITL$ proof system's Axiom~\refid{VPTL}.  Similarly, inference
rules based on \ModalSys{S4} can be obtained with Axiom~\refid{VPTL},
Inference Rule~\refid{BoxGen} (which corresponds to the inference rule N of
\ModalSys{S4}) and modus ponens.  Moreover, the $\PITL$ proof system's
Axiom~\refid{VPTL} permits using \emph{any} $\PITL$ formula which is a
substitution instance of some valid $\PTL$ formula which can also contain the
$\PTL$ operator $\Next$. In view of all this, we do not give much further
consideration to aspects of \ModalSys{S4} with $\Box$ and $\Diamond$.

In contrast to $\Box$, the $\PITL$ operator $\Bf$ does not have a
comprehensive axiom analogous to~\refid{VPTL}.  Therefore, we need to
explicitly prove in the $\PITL$ axiom system various modal properties of $\Bf$
and its dual $\Df$.  If only finite time is allowed, then $\Bf$ and $\Df$ act
as an \ModalSys{S4} system.  However, $\Bf$ with infinite time permitted does
not fulfil the requirements of \ModalSys{S4}, or even those of the weaker
modal system~\!\ModalSys{T}, because Axiom T fails.  Instead, $\Bf$ with
infinite time fulfils the requirements of the modal system \ModalSys{KD4}
which is strictly weaker than \ModalSys{S4}.

Here is a list of \ModalSys{KD4}'s axioms and inference rules and related
$\PITL$ proofs for $\Bf$:
\begin{center}
\begin{tabular}{lLl}
  K & \theoremBEN L(A\imp B) \implies (L\,A \imp L\,B)
    & Theorem~\refid{BfImpDist} \\
  N & \theoremBEN A \infer \theoremBEN L\,A & Derived Inf.~Rule~\refid{BfGen} \\
  D & \theoremBEN L\,A \implies M\,A & Theorem~\refid{BfImpDf} \\
  4 & \theoremBEN L\,A \implies LL\,A & Theorem~\refid{BfImpBfBf}
\end{tabular}
\end{center}
If only finite time is allowed, then the implication D does not need to be
regarded as an explicit axiom since it can be inferred from any proof system
for \ModalSys{S4}.

\begin{myremark}
  It is also worth noting that the related operators $\Bi$ and $\Di$ (defined
  using weak chop in Table~\ref{pitl-derived-operators-table} in
  Section~\ref{pitl-sec}) obey the modal system \ModalSys{S4} even when
  infinite time is permitted.  However, we prefer to work with $\Bf$ and $\Df$
  since the use of strong chop simplifies the overall $\PITL$ completeness
  proof.
\end{myremark}

Conventional model logics usually take $L$, not $M$, to be primitive.  When
we deduce standard modal properties for $\Bf$ and $\Df$ in our $\PITL$ axiom
system, we let $M$, which corresponds to $\Df$, be primitive and define $L$ to
be $M$'s dual (i.e., $L\,A \defeqv \Not M\,\Not A$).  This $M$-based approach
goes well with the $\PITL$ axioms for chop.  Chellas~\citeyear{Chellas80}
discusses some alternative axiomatisations of modal systems with $M$ as the
primitive although none correspond directly to ours.  For the system
\ModalSys{K}, we can deduce implication\relax
~\eqref{some-properties-of-bf-involving-the-modal-system-k-and-axiom-d-1-eq}
below for $\Bf$ and $\Df$ (see Theorem~\refid{BfImpDfImpDf} later on) and then
obtain from it together some other reasoning the more standard axiom K just
presented which only mentions $L$:
\begin{equation}
  \label{some-properties-of-bf-involving-the-modal-system-k-and-axiom-d-1-eq}
  \Theorem  L(A\imp B) \implies (M\,\!A \implies M\,\!B)
\dotspace.
\end{equation}

The operators $\Box$ and $\Bf$ together yield a \emph{multi-modal logic} with
two necessity constructs $L$ and $L'$ which are commutative:
\begin{equation*}
  \Theorem LL'\, A \Equiv L'L\,A
\dotspace.
\end{equation*}
This corresponds to our Theorem~\refid{BfBoxEqvBoxBf} given later on.

Below are various theorems and derived inference rules about $\Bf$ and $\Df$
for obtaining the axioms M-def (Theorem~\refid{Mdef}) and K
(Theorem~\refid{BfImpDist}) found in the modal system \ModalSys{K}.  The
associated inference rule N was already proved above as Derived Inference
Rule~\refid{BfGen}.  We also prove the modal axiom D
(Theorem~\refid{BfImpDf}).

In the next proof's final step, recall that \refid{EqvChain} indicates a chain
of equivalences:
\statetheorem{Mdef&
\Theorem \Df A \Equiv \Not\Bf\Not A}%
\begin{proveit}
1&A \Equiv \Not\Not A &\refid{Prop} \\
2&\Df A \Equiv \Df\Not\Not A &1,\refid{DfEqvDf} \\
3&\Df\Not\Not A \Equiv \Not\Not\Df\Not\Not A &\refid{Prop} \\
4&\Df\Not\Not A \Equiv \Not\Bf\Not A &3,\defof{\Bf} \\
5&\Df A \Equiv \Not\Bf\Not A &2,4,\refid{EqvChain}
\end{proveit}

\statetheorem{BfImpDfImpDf&
\Theorem \Bf(A \implies B) \Implies \Df A \implies \Df B}%
\begin{proveit}
1&\Bf(A\implies B)  \Implies (A\chop \True) \implies (B\chop \True)
   &\refid{BfChopImpChop} \\
2&\Bf(A \implies B) \Implies \Df A \implies \Df B
   &1,\defof{\Df}
\end{proveit}

\statetheorem{BfContraPosImpDist&
\Theorem \Bf(\Not B \implies \Not A) \Implies (\Bf A) \implies (\Bf B)}%
\begin{proveit}
1&\Bf(\Not B \implies \Not A) \Implies (\Df\Not B) \implies (\Df\Not A)
   &\refid{BfImpDfImpDf} \\
2&\Bf(\Not B \implies \Not A) \Implies (\Not\Df\Not A) \implies (\Not\Df\Not B)
   &1,\refid{Prop} \\
3&\Bf(\Not B \implies \Not A) \Implies (\Bf A) \implies (\Bf B)
   &2,\defof{\Bf}
\end{proveit}

\statetheorem{BfImpDist&
\Theorem \Bf(A\implies B) \Implies (\Bf A) \implies (\Bf B)}%
\begin{proveit}
1&(A \implies B) \Implies (\Not B \implies \Not A)   &\refid{Prop} \\
2&\Not(\Not B \implies \Not A) \Implies \Not(A \implies B)
   &1,\refid{Prop} \\
3&\Bf\bigl(\Not(\Not B \implies \Not A) \implies \Not(A \implies B)\bigr)
   &2,\refid{BfGen} \\
4&\begin{myarray}
    \Bf\bigl(\Not(\Not B \implies \Not A) \implies \Not(A \implies B)\bigr) \\
      \Implies \Bf(A \implies B) \implies \Bf(\Not B \implies \Not A)
   \end{myarray}
   &\refid{BfContraPosImpDist} \\
5&\Bf(A \implies B) \Implies \Bf(\Not B \implies \Not A)
   &3,4,\refid{MP} \\
6&\Bf(\Not B \implies \Not A) \Implies (\Bf A) \implies (\Bf B)
   &\refid{BfContraPosImpDist} \\
7&\Bf(A\implies B) \Implies (\Bf A) \implies (\Bf B)   &5,6,\refid{ImpChain}
\end{proveit}

\statederivedrule{BfImpBfRule&
\Theorem A \implies B \infer \Theorem \Bf A \implies \Bf B}%
\begin{proveit}
1&A \Implies B   &\refid{Assump} \\
2&\Bf(A \implies B)   &1,\refid{BfGen} \\
3&\Bf(A\implies B) \Implies (\Bf A) \implies (\Bf B)   &\refid{BfImpDist} \\
4&\Bf A \Implies \Bf B   &2,3,\refid{MP}
\end{proveit}

\statederivedrule{BfEqvBfRule&
\theoremBEN A\equiv B \infer \theoremBEN \Bf A \equiv \Bf B}%
\begin{proveit}
1& A \Equiv B  &\refid{Assump} \\
2& A \Implies B  &1,\refid{Prop} \\
3& \Bf A \Implies \Bf B  &2,\refid{BfImpBfRule} \\
4& B \Implies A  &1,\refid{Prop} \\
5& \Bf B \Implies \Bf A  &4,\refid{BfImpBfRule} \\
6& \Bf A \Equiv \Bf B  &3,5,\refid{Prop}
\end{proveit}

\statetheorem{BfAndEqv&
\Theorem \Bf(A \And B) \Equiv \Bf A \And \Bf B}%
\begin{proveit}
1&(A \And B) \Implies A  &\refid{Prop} \\
2&\Bf(A \And B) \Implies \Bf A  &1,\refid{BfImpBfRule} \\
3&(A \And B) \Implies B  &\refid{Prop} \\
4&\Bf(A \And B) \Implies \Bf B  &3,\refid{BfImpBfRule} \\
5&A \Implies (B \implies (A \And B))  &\refid{Prop} \\
6&\Bf A \Implies \Bf(B \implies (A \And B))  &5,\refid{BfImpBfRule} \\
7&\Bf(B \implies (A \And B)) \Implies \bigl(\Bf B \implies \Bf(A \And B)\bigr)
  &\refid{BfImpDist} \\
8&\Bf A \And \Bf B \Implies \Bf(A \And B)  &6,7,\refid{Prop} \\
9&\Bf(A \And B) \Equiv \Bf A \And \Bf B  &2,4,8,\refid{Prop}
\end{proveit}

\statetheorem{BfEqvSplit&
\Theorem \Bf(A \equiv B) \Equiv \Bf(A\imp B) \Andd \Bf(B\imp A)}
\begin{proveit}
1& (A \equiv B) \Equiv (A\imp B) \Andd (B\imp A)  &\refid{Prop} \\
2& \Bf(A \equiv B) \Equiv \Bf\bigl((A\imp B) \And (B\imp A)\bigr)  &1,\refid{BfEqvBfRule} \\
3& \Bf\bigl((A\imp B) \And (B\imp A)\bigr)
     \Equiv \Bf(A\imp B) \Andd \Bf(B\imp A)  &\refid{BfAndEqv} \\
4& \Bf(A \equiv B) \Equiv \Bf(A\imp B) \Andd \Bf(B\imp A)  &2,3,\refid{EqvChain}
\end{proveit}

\statetheorem{BfChopEqvChop&
\Theorem \Bf(A\equiv A')  \Implies  (A\chop B) \equiv (A'\chop B)}%
\begin{proveit}
1&\Bf(A\equiv A') \Equiv \Bf(A\imp A') \Andd \Bf(A'\imp A)
  &\refid{BfEqvSplit} \\
2&\Bf(A\imp A')  \Implies  (A\chop B) \implies (A'\chop B)
  &\refid{BfChopImpChop} \\
3&\Bf(A'\imp A)  \Implies  (A'\chop B) \implies (A\chop B)
  &\refid{BfChopImpChop} \\
4&\Bf(A\equiv A')  \Implies  (A\chop B) \equiv (A'\chop B)  &1--3,\refid{Prop}
\end{proveit}

\statetheorem{BfImpDfEqvDf&
\Theorem \Bf(A \equiv B) \Implies \Df A \equiv \Df B}
\begin{proveit}
1& \Bf(A \equiv B) \Implies (A\chop\True) \equiv(B\chop\True)  &\refid{BfChopEqvChop} \\
2& \Bf(A \equiv B) \Implies \Df A \equiv \Df B  &1,\defof{\Df}
\end{proveit}

\statederivedrule{FiniteImpDfEqvDfRule&
\theoremBEN \Finite \implies (A\equiv B) \infer \theoremBEN \Df A \equiv \Df B}%
\begin{proveit}
1&\Finite \Implies (A\equiv B)  &\refid{Assump} \\
2&\Bf(A\equiv B)  &1,\refid{BfFGen} \\
3&\Bf(A \equiv B) \Implies \Df A \equiv \Df B
   &\refid{BfImpDfEqvDf} \\
4&\Df A \Equiv \Df B  &2,3,\refid{MP}
\end{proveit}

\statetheorem{BfImpDf&
\Theorem \Bf A \Implies \Df A}%
\begin{proveit}
1&A \Implies (\Empty \implies A) &\refid{Prop} \\
2&\Bf A \Implies \Bf(\Empty \implies A) &1,\refid{BfImpBfRule} \\
3&\Bf(\Empty \implies A) \Implies (\Df\Empty \implies \Df A)
  &\refid{BfImpDfImpDf} \\
4&\Bf A \Implies (\Df\Empty \implies \Df A) &2,3,\refid{ImpChain} \\
5&\Df\Empty &\refid{DfEmpty} \\
6&\Bf A \Implies \Df A &4,5,\refid{Prop}
\end{proveit}

\statetheorem{DfOrEqv&
\Theorem \Df(A \Or B) \Equiv \Df A \,\Or\, \Df B}%
\begin{proveit}
1&(A\Or B)\chop \True \Equiv (A\chop\True) \Or (B\chop\True)  &\refid{OrChopEqv}
  \\
2&\Df(A\Or B) \Equiv \Df A \,\Or\, \Df B  &1,\defof{\Df}
\end{proveit}

%\ifmylmcs \hfill \fi

\statetheorem{BfAndChopImport&
\Theorem \Bf A \And (A'\chop B) \Implies (A \And A')\chop B}%
\begin{proveit}
1&A \Implies (A' \implies A\And A')   &\refid{Prop} \\
2&\Bf A \Implies \Bf (A' \implies A\And A')   &1,\refid{BfImpBfRule} \\
3&\Bf (A' \implies A\And A') \Implies   (A'\chop B) \implies (A \And A')\chop B
   &\refid{BfChopImpChop} \\
4&\Bf A \And (A'\chop B) \Implies (A\And A')\chop B   &2,3,\refid{Prop}
\end{proveit}

\subsection{Some Properties of Chop,  \texorpdfstring{$\Df$}{Diamond-f}  and \texorpdfstring{$\Bf$}{Box-f} with State Formulas}

\ifmylmcs \hfill \fi

\statetheorem{DfState&
\Theorem \Df w \Equiv w}%
\begin{proveit}[Proof for $\supset$]
1&\Not w \Implies \Bf\Not w&\refid{StateImpBf} \\
2&\Not w\Implies \Not\Df\Not\Not w&1,\defof{\Bf} \\
3&\Df\Not\Not w \Implies w&2,\refid{Prop} \\
4&w \Implies \Not\Not w&\refid{Prop} \\
5&\Df w \Implies \Df\Not\Not w  &4,\refid{DfImpDf} \\
6&\Df w \Implies w&3,5,\refid{ImpChain}
\end{proveit}
\begin{proveit}[Proof for $\subset$]
1&w \Implies \Bf w &\refid{StateImpBf} \\
2&\Bf w \Implies \Df w &\refid{BfImpDf} \\
3&w \Implies \Df w &1,2,\refid{ImpChain}
\end{proveit}

\statetheorem{BfState&
\Theorem \Bf w \Equiv w}%
\begin{proveit}
1&\Df \Not w \Equiv \Not w &\refid{DfState} \\
2&\Not\Df \Not w \Equiv w &1,\refid{Prop} \\
3&\Bf w \Equiv w &2,\defof{\Bf}
\end{proveit}

\statetheorem{StateChop&
\Theorem w\chop A \Implies w}%
\begin{proveit}
1&w\chop A \Implies \Df w &\refid{ChopImpDf} \\
2&\Df w \Equiv w &\refid{DfState} \\
3&w\chop A \Implies w &1,2,\refid{Prop}
\end{proveit}

\statetheorem{StateChopExportA&
\Theorem (w \And A)\chop B \Implies w}%
\begin{proveit}
1&w \And A \Implies w &\refid{Prop} \\
2&(w \And A)\chop B \Implies w \chop B &1,\refid{LeftChopImpChop} \\
3&w\chop B \Implies w &\refid{StateChop} \\
4&(w \And A)\chop B \Implies w &2,3,\refid{ImpChain}
\end{proveit}

The following lets us move a state formula into the left side of chop:
\statetheorem{StateAndChopImport&
\Theorem w \And (A\chop B) \Implies (w \And A)\chop B}%
\begin{proveit}
1&w\Implies \Bf w&\refid{StateImpBf} \\
2&w \And (A\chop B) \Implies \Bf w \And (A\chop B)   &1,\refid{Prop} \\
3&\Bf w \And (A\chop B) \Implies (w \And A)\chop B   &\refid{BfAndChopImport} \\
4&w \And (A\chop B) \Implies (w \And A)\chop B   &2,3,\refid{ImpChain}
\end{proveit}
We can easily combine this with theorem \refid{StateChopExportA} to deduce the
equivalence below:
\statetheorem{StateAndChop&
\Theorem (w \And A)\chop B \Equiv w \And (A\chop B)}%
\begin{proveit}
1&(w \And A)\chop B \Implies w &\refid{StateChopExportA} \\
2&(w \And A)\chop B \Implies (w\chop B) \Andd (A\chop B)
   &\refid{AndChopImpChopAndChop} \\
3&(w \And A)\chop B \Implies w\Andd (A\chop B) &1,2,\refid{Prop} \\
4&w \And (A\chop B) \Implies (w \And A)\chop B &\refid{StateAndChopImport} \\
5&w \And (A\chop B) \Equiv (w \And A)\chop B &3,4,\refid{Prop}
\end{proveit}

Below is a useful corollary of \refid{StateAndChop} used in decomposing the
left side of chop:
\statetheorem{StateAndEmptyChop&
\Theorem (w \And \Empty)\chop A \Equiv w \And A}%
\begin{proveit}
1&(w \And \Empty)\chop A \Equiv w \And (\Empty\chop A)
   &\refid{StateAndChop} \\
2&\Empty\chop A \Equiv A   &\refid{EmptyChop} \\
3&(w \And \Empty)\chop A \Equiv w \And A   &1,2,\refid{Prop}
\end{proveit}

\begin{hidetheorems}
The following is a simple corollary of \refid{StateAndEmptyChop}:
\statetheorem{EmptyAndStateChop&
\Theorem (\Empty \And w)\chop A \Equiv w \And A}%
\begin{proveit}
1&(\Empty \And w)\chop A \Equiv (w \And \Empty)\chop A
  &\refid{AndChopCommute} \\
2&(w \And \Empty)\chop A \Equiv w \And A  &\refid{StateAndEmptyChop} \\
3&(\Empty \And w)\chop A \Equiv w \And A  &1,2,\refid{EqvChain}
\end{proveit}

\statetheorem{StateAndDf&
\Theorem \Df(w \And A) \Equiv w \And \Df A}%
\begin{proveit}
1&(w \And A)\chop \True \Equiv w \And (A\chop \True)  &\refid{StateAndChop} \\
2&\Df (w \And A) \Equiv w \And \Df A  &1,\defof{\Df}
\end{proveit}
\end{hidetheorems}

\begin{mycomment}
\statederivedrule{StateImpBfGen&
\theoremBEN w \implies A \infer \theoremBEN w \implies \Bf A}%
\begin{proveit}
1&w \Implies A  &\refid{Assump} \\
2&\Not A \Implies \Not w  &1,\refid{Prop} \\
3&\Df\Not A \Implies \Df\Not w  &2,\refid{DfImpDf} \\
4&\Df\Not w \Equiv \Not w  &\refid{DfState} \\
5&\Df\Not A \Implies \Not w  &3,4,\refid{Prop} \\
6&w \Implies \Not\Df\Not A  &5,\refid{Prop} \\
7&w \Implies \Bf A  &6,\defof{\Bf}
\end{proveit}
\end{mycomment}

\begin{mycomment}
The following theorem can be used to do induction over time with chop:
\statetheorem{ChopAndNotChopImp&
\Theorem A\chop B \And \Not(A\chop B') \Implies A\chop (B \And \Not B')}%
\begin{proveit}
1&B \Implies (B\And \Not B') \Or B'&\refid{Prop} \\
2&A\chop B \Implies A\chop (B\And \Not B') \Or A\chop B'&
   1,\refid{LeftChopImpChop} \\
3&A\chop B \And \Not (A\chop B') \Implies A\chop (B \And \Not B')
  &2,\refid{Prop}\\
\end{proveit}
\end{mycomment}

\subsection{Some Properties of \texorpdfstring{$\Bf$}{Box-f} involving the Modal System K4}

\label{some-properties-of-bf-involving-the-modal-system-k4-subsec}

\ifmylmcs \hfill \fi

We now consider how to establish for the $\PITL$ operator $\Bf$ the axiom
``4'' ($\PITL$ Theorem~\refid{BfImpBfBf}) found in the modal systems \emph{K4}
and \ModalSys{S4}.

\statetheorem{DfDfEqvDf&
\Theorem \Df\Df A \Equiv \Df A}%
\begin{proveit}
1&(A\chop \True)\chop \True \Equiv A\chop (\True\chop\True)
   &\refid{ChopAssoc} \\
2&\Df\True \Equiv \True &\refid{DfState} \\
3&(\True\chop\True) \Equiv \True &2,\defof{\Df} \\
4&A\chop (\True\chop\True) \Equiv A\chop\True &3,\refid{LeftChopEqvChop} \\
5&(A\chop \True)\chop \True \Equiv A\chop\True&1,4,\refid{EqvChain} \\
6&\Df\Df A \Equiv \Df A &5,\defof{\Df}
\end{proveit}

\statetheorem{DfNotEqvNotBf&
\Theorem \Df\Not A \Equiv \Not\Bf A}%
\begin{proveit}
1&\Bf A \Equiv \Not\Df\Not A &\defof{\Bf} \\
2&\Df \Not A \Equiv \Not\Bf A &1,\refid{Prop}
\end{proveit}

\statetheorem{DfDfNotEqvNotBfBf&
\Theorem \Df\Df\Not A \Equiv \Not\Bf\Bf A}%
\begin{proveit}
1&\Df\Not A \Equiv \Not\Bf A  &\refid{DfNotEqvNotBf} \\
2&\Df\Df\Not A \Equiv \Df\Not\Bf A  &1,\refid{DfEqvDf} \\
3&\Df\Not\Bf A \Equiv \Not\Bf\Bf A  &\refid{DfNotEqvNotBf} \\
4&\Df\Df\Not A \Equiv \Not\Bf\Bf A  &2,3,\refid{EqvChain}
\end{proveit}

\statetheorem{BfBfEqvBf&
\Theorem \Bf\Bf A \Equiv \Bf A}%
\begin{proveit}
1&\Df\Df\Not A \Equiv \Df\Not A  &\refid{DfDfEqvDf} \\
2&\Df\Df\Not A \Equiv \Not\Bf\Bf A  &\refid{DfDfNotEqvNotBfBf} \\
3&\Not\Bf\Bf A \Equiv \Df\Not A  &1,2,\refid{Prop} \\
4&\Df\Not A \Equiv \Not\Bf A  &\refid{DfNotEqvNotBf} \\
5&\Not\Bf\Bf A \Equiv \Not\Bf A  &3,4,\refid{EqvChain} \\
6&\Bf\Bf A \Equiv \Bf A  &5,\refid{Prop}
\end{proveit}

\statetheorem{BfImpBfBf&
\Theorem \Bf A \Implies \Bf\Bf A}%
\begin{proveit}
1&\Bf\Bf A \Equiv \Bf A &\refid{BfBfEqvBf} \\
2&\Bf A \Implies \Bf\Bf A &1,\refid{Prop}
\end{proveit}

\subsection{Properties Involving the PTL Operator \texorpdfstring{$\Next$}{Next}}

\ifmylmcs \hfill \fi

\statetheorem{NextChop&
\Theorem (\Next A)\chop B \Equiv \Next(A\chop B)}%
\begin{proveit}
1&(\Skip\chop A)\chop B \Equiv \Skip\chop (A\chop B)   &\refid{ChopAssoc} \\
2&(\Next A)\chop B \Equiv \Next(A\chop B)   &1,\defof{\Next}
\end{proveit}

\statetheorem{StateAndNextChop&
\Theorem (w \And \Next A)\chop B \Equiv w \And \Next(A\chop B)}%
\begin{proveit}
1&(w \And \Next A)\chop B \Equiv w \And \bigl((\Next A)\chop B\bigr)
   &\refid{StateAndChop} \\
2&(\Next A)\chop B \Equiv \Next(A\chop B)   &\refid{NextChop} \\
3&(w \And \Next A)\chop B \Equiv w \And \Next(A\chop B)   &1,2,\refid{Prop}
\end{proveit}

\statetheorem{DfStateAndNextEqv&
\Theorem \Df(w \And \Next w') \Equiv w \And \Next w'}%
\begin{proveit}
1&(w \And \Next w')\chop\True \Equiv w \And \Next(w'\chop\True)
   &\refid{StateAndNextChop} \\
2&\Df(w \And \Next w') \Equiv w \And \Next\Df w' &1,\defof{\Df} \\
3&\Df w' \Equiv w' &\refid{DfState} \\
4&\Skip\chop \Df w' \Equiv \Skip\chop w' &3,\refid{RightChopEqvChop} \\
5&\Next\Df w' \Equiv \Next w' &4,\defof{\Next} \\
6&\Df(w \And \Next w') \Equiv w \And \Next w'  &2,5,\refid{Prop}
\end{proveit}

\subsection{Some Properties of \texorpdfstring{$\Bf$}{Box-f} Together with \texorpdfstring{$\Box$}{Box}}

\ifmylmcs \hfill \fi

We make use of the following analogue of Theorem~\refid{DfNotEqvNotBf} for
$\Diamond$ and $\Box$:
\statetheorem{DiamondNotEqvNotBox&
\Theorem \Diamond\Not A \Equiv \Not\Box A}%
\begin{proveit}
1&\Diamond \Not A \Equiv \Not\Box A &\refid{VPTL}
\end{proveit}

\statetheorem{DfDiamondEqvDiamondDf&
\Theorem \Df\Diamond A \Equiv \Diamond\Df A}%
\begin{proveit}
1&(\True\chop A)\chop \True \Equiv \True\chop (A\chop\True)
   &\refid{ChopAssoc} \\
2&(\Diamond A)\chop \True \Equiv \Diamond (A\chop\True)  &1,\defof{\Diamond} \\
3&\Df\Diamond A \Equiv \Diamond\Df A  &2,\defof{\Df}
\end{proveit}

\statetheorem{DfDiamondNotEqvNotBfBox&
\Theorem \Df\Diamond\Not A \Equiv \Not\Bf\Box A}%
\begin{proveit}
1&\Diamond\Not A \Equiv \Not\Box A &\refid{DiamondNotEqvNotBox} \\
2&\Df\Diamond\Not A \Equiv \Df\Not\Box A &1,\refid{DfEqvDf} \\
3&\Df\Not\Box A \Equiv \Not\Bf\Box A &\refid{DfNotEqvNotBf} \\
4&\Df\Diamond\Not A \Equiv \Not\Bf\Box A &2,3,\refid{EqvChain}
\end{proveit}

\statetheorem{DiamondDfNotEqvNotBoxBf&
\Theorem \Diamond\Df\Not A \Equiv \Not\Box\Bf A}%
\begin{proveit}
1&\Df\Not A\Equiv \Not\Bf A &\refid{DfNotEqvNotBf} \\
2&\Diamond\Df\Not A \Equiv \Diamond\Not\Bf A
  &1,\refid{DiamondEqvDiamond} \\
3&\Diamond\Not\Bf A \Equiv \Not\Box\Bf A
   &\refid{DiamondNotEqvNotBox} \\
4&\Diamond\Df\Not A \Equiv \Not\Box\Bf A &2,3,\refid{EqvChain}
\end{proveit}

\statetheorem{BfBoxEqvBoxBf&
\Theorem \Bf\Box A \Equiv \Box\Bf A}%
\begin{proveit}
1&\Df\Diamond\Not A \Equiv \Diamond\Df\Not A
   &\refid{DfDiamondEqvDiamondDf} \\
2&\Df\Diamond\Not A \Equiv \Not\Bf\Box A &\refid{DfDiamondNotEqvNotBfBox} \\
3&\Diamond\Df\Not A \Equiv \Not\Box\Bf A &\refid{DiamondDfNotEqvNotBoxBf} \\
4&\Bf\Box A \Equiv \Box\Bf A &1-3,\refid{Prop}
\end{proveit}

\subsection{Some Properties of Chop-Star}

\ifmylmcs \hfill \fi

We now consider some theorems and derived rules concerning chop-star.

\statederivedrule{ImpMoreChopStarEqvRule&
\theoremBEN A\imp \More
  \infer \theoremBEN
    A\SChopstar \EQUIV \Empty \,\Or\, (A\chop A\SChopstar)}%
\begin{proveit}
1&A\Implies \More &\refid{Assump} \\
2&A\And\More \Equiv A &1,\refid{Prop} \\
3&(A\And\More)\chop A\SChopstar \Equiv A\chop A\SChopstar
  &2,\refid{LeftChopEqvChop} \\
4&A\SChopstar \Equiv \Empty \,\Or\, \bigl((A\And \More)\chop A\SChopstar\bigr)
  &\refid{SChopStarEqv} \\
5&A\SChopstar \Equiv \Empty \,\Or\, (A\chop A\SChopstar)
  &3,4,\refid{Prop}
\end{proveit}

\statederivedrule{ImpMoreChopStarChopEqvRule&
\theoremBEN A\imp \More
  \infer \theoremBEN
    A\SChopstar\chop B \EQUIV B \,\Or\, \bigl(A\chop(A\SChopstar\chop B)\bigr)}%
\begin{proveit}
1&A\Implies \More &\refid{Assump} \\
2&A\SChopstar \Equiv \Empty \,\Or\, (A\chop A\SChopstar)
  &1,\refid{ImpMoreChopStarEqvRule} \\
3&A\SChopstar\chop B \Equiv (\Empty \,\Or\, (A\chop A\SChopstar))\chop B
  &2,\refid{LeftChopEqvChop} \\
4&(\Empty \Or (A\chop A\SChopstar))\chop B
 \Equiv (\Empty\chop B) \Or \bigl((A\chop A\SChopstar)\chop B\bigr)
  &\refid{OrChopEqv} \\
5&\Empty\chop B \Equiv B  &\refid{EmptyChop} \\
6&(A\chop A\SChopstar)\chop B \Equiv A\chop (A\SChopstar\chop B)
  &\refid{ChopAssoc} \\
7&A\SChopstar\chop B \Equiv B \,\Or\, \bigl(A\chop (A\SChopstar\chop B)\bigr)
  &3--6,\refid{Prop}
\end{proveit}

\statetheorem{SChopStarEqvSChopstarChopEmptyOrChopOmega&
\Theorem A\SChopstar \Equiv (A\SChopstar\chop\Empty) \Or A\Chopomega}%
\begin{proveit}
1&\Finite \Or \Not\Finite  &\refid{Prop} \\
2&\Finite \Or \Inf  &1,\defof{\Inf} \\
3&\Finite  \Implies  (A\SChopstar\chop \Empty) \equiv A\SChopstar
  &\refid{FiniteImpChopEmpty} \\
4&\Inf \Implies A\SChopstar\equiv (A\SChopstar\And \Inf) &\refid{Prop} \\
5&\Inf \Implies A\SChopstar\equiv A\Chopomega
  &4,\defof{\text{chop-omega}} \\
6&A\SChopstar \Equiv (A\SChopstar\chop\Empty) \Or A\Chopomega
  &2,3,5,\refid{Prop}
\end{proveit}

\subsection{Some Properties Involving a Reduction to PITL with Finite Time}

\ifmylmcs \hfill \fi

We now present some derived inference rules which come in useful when
completeness for $\PITL$ with finite time is assumed (see
Theorem~\ref{completeness-of-pitl-axiom-system-for-finite-time-thm}).  Recall
that any valid implication of the form $\Finite \imp A$ is allowed and that we
designate such a step by using \refid{PITLF}.  $\PITL$
Theorem~\refid{BfFinStateEqvBox} below illustrates this technique.

\statederivedrule{FiniteImpBfImpBfRule&
\theoremBEN \Finite \implies (A \imp B)
  \infer \theoremBEN \Bf A \implies \Bf B}%
\begin{proveit}
1&\Finite \Implies (A \implies B)  &\refid{Assump} \\
2&\Bf(A\implies B) &1,\refid{BfFGen} \\
3&\Bf(A\implies B) \Implies (\Bf A \implies \Bf B) &\refid{BfImpDist} \\
4&\Bf A \Implies \Bf B  &2,3,\refid{MP}
\end{proveit}

\statederivedrule{FiniteImpBfEqvBfRule&
\theoremBEN \Finite \implies (A \equiv B)
  \infer \theoremBEN \Bf A \EQUIV \Bf B}%
\begin{proveit}
1&\Finite \Implies (A \equiv B)  &\refid{Assump} \\
2&\Finite \Implies (A \implies B)  &1,\refid{Prop} \\
3&\Bf A \Implies \Bf B  &2,\refid{FiniteImpBfImpBfRule} \\
4&\Finite \Implies (B \implies A)  &1,\refid{Prop} \\
5&\Bf B \Implies \Bf A  &4,\refid{FiniteImpBfImpBfRule} \\
6&\Bf A \Equiv \Bf B &3,5,\refid{Prop}
\end{proveit}

The next theorem's proof involves the application of the previous derived
inference rule together with completeness for $\PITL$ with just finite time:
\statetheorem{BfFinStateEqvBox&
\Theorem \Bf\Fin w \Equiv \Box w}%
\begin{proveit}
1&\Bf\Bf\Fin w \Equiv \Bf\Fin w &\refid{BfBfEqvBf} \\
2&\Bf\Fin w \Equiv \Bf\Bf\Fin w &1,\refid{Prop} \\
3&\Finite \Implies \bigl((\Bf\Fin w) \EQUIV \Box w\bigr) &\refid{PITLF} \\
4&\Bf\Bf\Fin w \Equiv \Bf\Box w &3,\refid{FiniteImpBfEqvBfRule} \\
5&\Bf\Box w \Equiv \Box\Bf w &\refid{BfBoxEqvBoxBf} \\
6&\Bf w \Equiv w &\refid{BfState} \\
7&\Box\Bf w \Equiv \Box w &6,\refid{BoxEqvBox} \\
8&\Bf\Fin w \Equiv \Box w &2,4,5,7,\refid{EqvChain}
\end{proveit}

An alternative proof of Theorem~\refid{BfFinStateEqvBox} can be given
without~\refid{PITLF} by first deducing the dual equivalence
$\bigl(\Df\Diamond(\Empty\And w)\bigr) \equiv \Diamond w$, for any state
formula $w$.

\egroup

\subsection{Some Properties of Skip, Next And Until} 

Recall from \S\ref{nl-one-formulas-subsec} that $\NLone$ formulas are exactly
those $\PTL$ formulas in which the only temporal operators are unnested
$\Next$s (e.g., $p \Or \Next\Not p$ but not $p \Or \Next\Next\Not p$).  The
next theorem holds for any $\NLone$ formula $T$:
\statetheorem{DfMoreAndNLoneEqvMoreAndNLone& \Theorem \Df(\More \And T) \Equiv
  \More \And T}%
\begin{proof}
  We use Axiom~\refid{VPTL} to re-express $\More \And T$ as a logically
  equivalent disjunction $\bigvee_{1\le i\le n}(w_i \And \Next w'_i)$ for some
  natural number $n\ge 1$ and $n$ pairs of state formulas $w_i$ and $w'_i$:
  \begin{equation}
    \label{DfMoreAndNLoneEqvMoreAndNLone-1-eq}
    \Theorem \More \And T \Equiv \bigvee_{1\le i\le n}(w_i \And \Next w'_i)
\dotspace.
  \end{equation}
  Now by Theorem~\refid{DfStateAndNextEqv} any conjunction $w \And \Next w'$
  is deducibly equivalent to $\Df(w \And \Next w')$.  Therefore the
  disjunction in~\eqref{DfMoreAndNLoneEqvMoreAndNLone-1-eq} can be
  re-expressed as $\bigvee_{1\le i\le n}\Df(w_i \And \Next w'_i)$:
  \begin{equation}
    \label{DfMoreAndNLoneEqvMoreAndNLone-2-eq}
    \Theorem \bigvee_{1\le i\le n}(w_i \And \Next w'_i)
      \Equiv \bigvee_{1\le i\le n}\Df(w_i \And \Next w'_i)
\dotspace.
  \end{equation}
  Then by $n-1$ applications of Theorem~\refid{DfOrEqv} and some simple
  propositional reasoning, the righthand operand of this equivalence is itself
  is deducibly equivalent to $\Df\bigl(\bigvee_{1\le i\le n}(w_i \And \Next
  w'_i)\bigr)$:
  \begin{equation}
    \label{DfMoreAndNLoneEqvMoreAndNLone-3-eq}
    \Theorem \bigvee_{1\le i\le n}\Df(w_i \And \Next w'_i)
      \Equiv \Df\bigl(\bigvee_{1\le i\le n}(w_i \And \Next w'_i)\bigr)
\dotspace.
  \end{equation}
  The chain of the three
  equivalences~\eqref{DfMoreAndNLoneEqvMoreAndNLone-1-eq}--\relax
  \eqref{DfMoreAndNLoneEqvMoreAndNLone-3-eq} yields the following:
  \begin{equation*}
    \Theorem \More \And T
      \Equiv \Df\bigl(\bigvee_{1\le i\le n}(w_i \And \Next w'_i)\bigr)
\dotspace.
  \end{equation*}
  We then apply Derived Rule~\refid{DfEqvDf} to the first
  equivalence~\eqref{DfMoreAndNLoneEqvMoreAndNLone-1-eq}:
  \begin{equation*}
    \Theorem \Df(\More \And T)
      \Equiv \Df\bigl(\bigvee_{1\le i\le n}(w_i \And \Next w'_i)\bigr)
\dotspace.
  \end{equation*}
  The last two equivalences with simple propositional reasoning yield our
  goal~\refid{DfMoreAndNLoneEqvMoreAndNLone}.
\end{proof}

Here is a corollary of the previous $\PITL$
Theorem~\refid{DfMoreAndNLoneEqvMoreAndNLone} for any $\NLone$ formula $T$:
\statetheorem{BfMoreImpNLoneEqvMoreImpNLone&
\Theorem \Bf(\More \imp T) \Equiv \More \imp T}%
\begin{proveit}
1&\Bf(\More \imp T) \Equiv \Not\Df\Not(\More \imp T)  &\defof{\Bf} \\
2&\Not(\More \imp T) \Equiv \More \And \Not T  &\refid{Prop} \\
3&\Df\Not(\More \imp T) \Equiv \Df(\More \And \Not T)
  &2,\refid{DfEqvDf} \\
4&\Df(\More \And \Not T) \Equiv \More \And \Not T
  &\refid{DfMoreAndNLoneEqvMoreAndNLone} \\
5&\Df\Not(\More \imp T) \Equiv \More \And \Not T &3,4,\refid{EqvChain} \\
6&\Bf(\More \imp T) \Equiv \Not(\More \And \Not T) &1,5,\refid{Prop} \\
7&\Not(\More \And \Not T) \Equiv \More \imp T  &\refid{Prop} \\
8&\Bf(\More \imp T) \Equiv \More \imp T &6,7,\refid{EqvChain}
\end{proveit}

\statetheorem{MoreAndNLoneImpBfMoreImpNLone&
\Theorem \More \Andd T \Implies \Bf(\More \imp T)}%
\begin{proveit}
1&\Bf(\More \imp T) \Equiv \More \imp T
  &\refid{BfMoreImpNLoneEqvMoreImpNLone} \\
2&\More \And T \Implies \Bf(\More \imp T)  &1,\refid{Prop}
\end{proveit}

\statetheorem{BfSkipImpAndNextImpAndSkipAndChop&
\Theorem \Bf(\Skip \imp A) \Andd \Next B \Implies (\Skip \And A)\chop B}%
\begin{proveit}
1&\Bf(\Skip \imp A) \Andd (\Skip\chop B)
    \Implies \bigl((\Skip \imp A) \And \Skip\bigr)\chop B
  &\refid{BfAndChopImport} \\
2&(\Skip \imp A) \Andd \Skip \Implies \Skip \Andd A  &\refid{Prop} \\
3&\bigl((\Skip \imp A) \And \Skip\bigr)\chop B \Implies (\Skip \And A)\chop B
  &2,\refid{LeftChopImpChop} \\
4&\Bf(\Skip \imp A) \Andd (\Skip\chop B) \Implies (\Skip \And A)\chop B
  &1,3,\refid{Prop} \\
5&\Bf(\Skip \imp A) \Andd \Next B \Implies (\Skip \And A)\chop B
  &4,\defof{\Next}
\end{proveit}

\statetheorem{BfMoreImpImpBfSkipImp&
\Theorem \Bf(\More \imp A) \Implies \Bf(\Skip \imp A)}%
\begin{proveit}
1&\More \Implies \Skip  &\refid{VPTL} \\
2&(\More \imp A) \Implies (\Skip \imp A)  &1,\refid{Prop} \\
3&\Bf(\More \imp A) \Implies \Bf(\Skip \imp A)
  &2,\refid{BfImpBfRule}
\end{proveit}

\statetheorem{BfMoreImpAndNextImpAndSkipAndChop&
\Theorem \Bf(\More \imp A) \Andd \Next B \Implies (\Skip \And A)\chop B}%
\begin{proveit}
1&\Bf(\More \imp A) \Implies \Bf(\Skip \imp A)
  &\refid{BfMoreImpImpBfSkipImp} \\
2&\Bf(\Skip \imp A) \Andd \Next B \Implies (\Skip \And A)\chop B
  &\refid{BfSkipImpAndNextImpAndSkipAndChop} \\
3&\Bf(\More \imp A) \Andd \Next B \Implies (\Skip \And A)\chop B
  &1,2,\refid{Prop}
\end{proveit}

\statetheorem{DfSkipAndNLoneEqvMoreAndNLone&
\Theorem \Df(\Skip \And T) \Equiv \More \And T}%
\begin{proveit}
1&\Finite \Implies \Df(\Skip \And T) \EQUIV (\More \And T)  &\refid{PITLF} \\
2&\Df\Df(\Skip \And T) \Equiv \Df(\More \And T)
  &1,\refid{FiniteImpDfEqvDfRule} \\
3&\Df\Df(\Skip \And T) \Equiv \Df(\Skip \And T)  &\refid{DfDfEqvDf} \\
4&\Df(\More \And T) \Equiv \More \And T
  &\refid{DfMoreAndNLoneEqvMoreAndNLone} \\
5&\Df(\Skip \And T) \Equiv \More \And T  &2--4,\refid{Prop}
\end{proveit}

\begin{hidetheorems}
\statetheorem{DfSkipAndNLoneImpBfSkipImpNLone&
\Theorem \Df(\Skip \And T) \Implies \Bf(\Skip \imp T)}%
\begin{proveit}
1&\Df(\Skip \And T) \Equiv \More \And T
  &\refid{DfSkipAndNLoneEqvMoreAndNLone} \\
2&\Df(\Skip \And \Not T) \Equiv \More \And \Not T
  &\refid{DfSkipAndNLoneEqvMoreAndNLone} \\
3&\More \And T \Implies \Not(\More \And T)  &\refid{Prop} \\
4&\Df(\Skip \And T) \Implies \Not\Df(\Skip \And \Not T)  &1--3,\refid{Prop} \\
5&\Skip \And \Not T \Equiv \Not(\Skip \imp T)  &\refid{Prop} \\
6&\Df(\Skip \And \Not T) \Equiv \Df\Not(\Skip \imp T)  &5,\refid{DfEqvDf} \\
7&\Not\Df(\Skip \And \Not T) \Equiv \Not\Df\Not(\Skip \imp T)
  &6,\refid{Prop} \\
8&\Df(\Skip \And T) \Implies \Not\Df\Not(\Skip \imp T)  &4,7,\refid{Prop} \\
9&\Df(\Skip \And T) \Implies \Bf(\Skip \imp T)  &8,\defof{\Bf}
\end{proveit}
\end{hidetheorems}

\statetheorem{NLoneAndSkipChopEqvNLoneAndNext&
\Theorem (\Skip \And T)\chop A \Equiv T \Andd \Next A}%
\vspace{6pt}
\begin{proveit}[Proof for $\supset$]
1&(\Skip \And T)\chop A \Implies \Df(\Skip \And T)
  &\refid{ChopImpDf} \\
2&\Df(\Skip \And T) \Equiv \More \And T
  &\refid{DfSkipAndNLoneEqvMoreAndNLone} \\
3&(\Skip \And T)\chop A \Implies T  &1,2,\refid{Prop} \\
4&(\Skip \And T)\chop A \Implies \Skip\chop A
  &\refid{AndChopA} \\
5&(\Skip \And T)\chop A \Implies \Next A
  &4,\defof{\Next} \\
6&(\Skip \And T)\chop A \Implies T \Andd \Next A
  &3,5,\refid{Prop}
\end{proveit}
\vspace{6pt}
\begin{proveit}[Proof for $\subset$]
1&\Next A \Implies \More  &\refid{VPTL} \\
2&\More \Andd T \Implies \Bf(\More \imp T)
  &\refid{MoreAndNLoneImpBfMoreImpNLone} \\
3&T \Andd \Next A \Implies \Bf(\More \imp T)  &1,2,\refid{Prop} \\
4&\Bf(\More \imp T) \Andd \Next A \Implies (\Skip \And T)\chop A
  &\refid{BfMoreImpAndNextImpAndSkipAndChop} \\
5&T \Andd \Next A \Implies (\Skip \And T)\chop A
  &3,4,\refid{Prop}
\end{proveit}

\statetheorem{UntilEqv&
\Theorem T\Until A \Equiv A \,\Or\, \bigr(T \And \Next(T\Until A)\bigr)}%
\begin{proveit}
1&\Skip\And T \Implies \More  &\refid{VPTL} \\
2&(\Skip \And T)\SChopstar\chop A
  \Equiv A \,\Or\,
    \bigl((\Skip \And T)\chop ((\Skip \And T)\SChopstar\chop A)\bigr)
  &1,\refid{ImpMoreChopStarChopEqvRule} \\
3&T\Until A \Equiv A \,\Or\, \bigl((\Skip \And T)\chop (T \Until A)\bigr)
 &2,\defof{\Until} \\
4&(\Skip \And T)\chop (T \Until A) \Equiv T \Andd \Next(T \Until A)
  &\refid{NLoneAndSkipChopEqvNLoneAndNext} \\
5&T\Until A \Equiv A \,\Or\, \bigr(T \And \Next(T\Until A)\bigr)
 &3--4,\refid{Prop} 
\end{proveit}

\statetheorem{UntilImpDiamond&
\Theorem T\Until A \Implies \Diamond A}%
\begin{proveit}
1&(\Skip \And T)\SChopstar\chop A \Implies \Diamond A &\refid{ChopImpDiamond} \\
2&T\Until A \Implies \Diamond A &1,\defof{\Until}
\end{proveit}

\immediate\write\MyThmsOut{\noexpand\EndMyThmsMacro}

%%%(((RefId support: Start of part 2 of 2

\ifPrintRefIdXrefs \let\MyDummyMacro=\relax \else
\let\MyDummyMacro=\enddocument \fi \MyDummyMacro

\begingroup

\small

\newif\ifMySortByLatexName

% For description of \MyIdEntry, see first definition.
\def\MyIdEntry,#1,#2,#3,#4,#5\myend{\relax
   \MyIdEntryThreeColumns{#1}{#2}{#3}{#4}{#5}\cr
   \MyIdShowXrefs{#1}%
}

% Format three columns when sorting by LaTeX ids:
\def\MyIdEntryThreeColumnsLaTeXId#1#2#3#4#5{%
   \MyIdHyperLink{#1}\relax
   &#5\relax
   &\MyIdShowDefLocs{#1}\relax
}

\def\MyIdShowXrefs#1{\expandafter\ifx\csname #1-disable-xrefs\endcsname\relax
      \expandafter\ifx\csname #1-xref-max\endcsname\relax
      \else
        \noalign{\nobreak\hangindent=2\parindent \hangafter=1 xrefs:
                \hsize=0.9\hsize\MyPrintXrefs{#1}}
      \fi
   \fi
}

% Format undefined entries when sorting by LaTeX ids:
\def\MyIdUndefinedLaTeXId,#1,#2,#3,{%
   \myhyperlink{undefid.#1}{#1???}&\textbf{ ***Undef***}%
   &\MyIdShowDefLocs{#1}\cr
   \MyIdShowXrefs{#1}%
}

% Format three columns when sorting by generated ids:
\def\MyIdEntryThreeColumnsGeneratedTriplet#1#2#3#4#5{%
   \let\MyIdEntryhypertarget=\MyHyperDummy
   \let\MyIdEntryhyperlink=\MyHyperDummy
   \expandafter\ifx\csname #3\endcsname\relax
     \let\MyIdEntryhypertarget=\myhypertarget
     \let\MyIdEntryhyperlink=\myhyperlink
   \fi
   \MyIdEntryhypertarget{symtab.#1.0}{}\relax
   \MyIdEntryhyperlink{refid.#1.0}{#5}\relax
   &\MyIdShowDefLocs{#1}\relax
   &#1\relax
}

% Show locations where id is defined
% #1=id
\def\MyIdShowDefLocs#1{{\relax
   \expandafter\ifx\csname #1-iddef-max\endcsname\relax
      \expandafter\ifx\csname #1-undefined\endcsname\relax
         \mywarning{ITL id #1 has no page numbers}%
      \fi
      ---[no pagenum]%
   \else
      \tempcount=1 %
      \loop
         \myhyperlink{iddef.#1.\the\tempcount}
            {\csname #1-iddef.\the\tempcount\endcsname}\relax
      \ifnum\csname #1-iddef-max\endcsname>\tempcount
         \advance\tempcount by 1\relax
            ,
      \repeat
   \fi
   \MyIdShowHardcopyDefLocs{#1}%
}}

% Show locations where id is defined in hardcopy version
\def\MyIdShowHardcopyDefLocs#1{\relax
   \ifprintmyids
      \expandafter\ifx\csname #1-printed-ref\endcsname\relax
      \else
         \ \footnotesize\selectfont(\csname #1-printed-ref\endcsname
           \expandafter\ifx\csname #1-printed-pagenums\endcsname\relax
           \else
              :\csname #1-printed-pagenums\endcsname
           \fi
         )
      \fi
   \fi
}

% Print cross-references for a given id
% #1=id
\def\MyPrintXrefs#1{\relax
   \expandafter\ifx\csname #1-xref-max\endcsname\relax
   \else
      \tempcount=1\relax
      \loop
         \myhyperlink{myxref.#1.\the\tempcount}
            {\csname #1-xref-\the\tempcount.pagenum\endcsname}\relax
         \expandafter\ifx\csname #1-xref-\the\tempcount.in-proof\endcsname
                   \relax
      %      \typeout{***Cannot find #1-xref-\the\tempcount.in-proof}\relax
         \else
            \ (\csname\csname #1-xref-\the\tempcount.in-proof\endcsname
                   \endcsname)\relax
         \fi
      \ifnum\csname #1-xref-max\endcsname>\tempcount
         \advance\tempcount by 1\relax
            ,
      \repeat
   \fi
}

\def\MyIdHyperLink#1{\relax
   \expandafter\ifx\csname #1-iddef-max\endcsname\relax
      % This id has no page numbers so essentially undefined
      \myhypertarget{refid.#1.0}{#1(?)}\relax
   \else
      \myhyperlink{refid.#1.0}{#1}\relax
   \fi
}

% \MyIdXrefPageNum: Macro for handling page numbers to cross references.
% #1=id
% #2=xref index
% #3=page number
\def\MyIdXrefPageNum,#1,#2,#3,{
   \typeout{***refid #1 used on page #3 with xref index=#2}
   \expandafter\xdef\csname #1-xref-max\endcsname{#2}
   \expandafter\gdef\csname #1-xref-#2.pagenum\endcsname{#3}
}

% \MyIdXrefInProof: Macro for handling cross-references to other ids.
% #1=id
% #2=xref index
% #3=cross-referenced id
\def\MyIdXrefInProof,#1,#2,#3,
   {\relax
      \typeout{Defining #1-xref-#2.in-proof for #1 to reference #3}
      \expandafter\gdef\csname #1-xref-#2.in-proof\endcsname{#3}
   }

% Macro def for undefined references
\def\MyIdUndefined,#1,#2,#3,{}

\def\MyIdUndefinedShowXrefByLaTeXId,#1,#2,#3,{{\relax
   \parskip=0pt %
   \tempdimen=\MyIdsUndefinedXrefMaxWidthDimen
   \advance\tempdimen by 2em %
   \hangindent=\tempdimen \hangafter=1\raggedright\leavevmode
   \myhypertarget{undefid.#1}{%
      \hbox to \MyIdsUndefinedXrefMaxWidthDimen{#1}}\qquad
   \expandafter\ifx\csname #1-disable-xrefs\endcsname\relax
      \MyPrintXrefs{#1}%
   \fi
   \par
}}

% \MyIdUndefinedGetMaxWidth: Used to determine max width of undefined ids
\def\MyIdUndefinedGetMaxWidth,#1,#2,#3,{%
   \setbox\tempbox=\hbox{#1}%
   \ifdim\wd\tempbox>\tempdimen
      \tempdimen=\wd\tempbox
   \fi
}

\def\MyIdsIndexFileName{\jobname.myids_i}

% Could use Latex's \IfFileExists{filename}{...if-part...}{...else-part...}
% and \InputIfFileExists{filename}
% See the definitions in latex.ltx
\newread\MyIdsIndexIn \openin\MyIdsIndexIn = \MyIdsIndexFileName
\newif\ifHaveFile
\newif\ifFirstInList
\ifeof\MyIdsIndexIn \HaveFilefalse \else \HaveFiletrue \fi
\closein\MyIdsIndexIn
\ifHaveFile
   \input{\MyIdsIndexFileName}
   % Process cross references:
   \MyIdsXrefsUnsorted
   \clearpage
   \begin{longtable}[l]{lll}
      \multicolumn{3}{c}{\large Ids sorted by names in \LaTeX{} file} \\\hline
      \endhead
      \noalign{\global\let\MyIdEntryThreeColumns=\MyIdEntryThreeColumnsLaTeXId
         \global\let\MyIdUndefined=\MyIdUndefinedLaTeXId}
      \MyIdsIndexSortByLatexId
   \end{longtable}
   \clearpage
   \begin{longtable}[l]{lll}
      \multicolumn{3}{c}{\large Ids sorted by generated names} \\\hline
      \endhead
      \noalign{\global\let\MyIdEntryThreeColumns
                =\MyIdEntryThreeColumnsGeneratedTriplet}
      \MyIdsIndexSortByGeneratedTriplet
   \end{longtable}
   \let\MyIdUndefined=\MyIdUndefinedGetMaxWidth
   \tempdimen=0pt %
   \MyIdsUndefinedSortByLatexId  % Dummy invocation just to get max id width.
   \ifdim\tempdimen>0pt % See if have undefined ids from previous run
      \edef\MyIdsUndefinedXrefMaxWidthDimen{\the\tempdimen}
      \clearpage
      \hbox to \textwidth{\hfil Undefined ids sorted by name\hfil}
      \vspace{2pt}
      \hrule
      \vspace{2pt}
      \let\MyIdUndefined=\MyIdUndefinedShowXrefByLaTeXId
      \MyIdsUndefinedSortByLatexId
   \else
      \ifmyidundefined
         \begin{center}
            {\large ***My warning: Undefined ids during current run***}
            \mywarning{Undefined ids during current run}
         \end{center}
      \fi
   \fi

   \clearpage

   \begin{center}
      \large Topological sort of ids
   \end{center}

   \FirstInListtrue

   % \x is used in macro \MyIdsTopologicalSort's generated data
   \def\x#1{\ifFirstInList\global\FirstInListfalse\else, \fi
      \expandafter\ifx\csname #1-ref\endcsname\relax
         \InsertHyphens{#1???}\relax
      \else
         \InsertHyphens{#1} (\csname #1-ref\endcsname)\relax
      \fi
   }

   \expandafter\ifx\csname MyIdsTopologicalSortMessages\endcsname\relax
   \else
      \halign{#\hfil\cr
         \textbf{***Error during typological sort using ``tsort'':} \cr
         \MyIdsTopologicalSortMessages
      }
   \fi

   \MyIdsTopologicalSort

\fi

\endgroup

% The following \fi seems extraneous:
%\fi

%%%)))RefId support: End of part 2 of 2

\end{document}

%%% Local Variables: 
%%% mode: latex
%%% mode: reftex
%%% version-control: t
%%% kept-new-versions: 5000
%%% fill-column: 80
%%% TeX-source-specials-mode: 1
%%% reftex-plug-into-AUCTeX: t
%%% TeX-master: t
%%% End: 